%% file: thesis.tex
\documentclass[reqno,a4paper,12pt,abstract=on]{scrreprt}

\usepackage[a4paper,hmargin=2cm,tmargin=3cm,bmargin=3cm]{geometry}
\usepackage[utf8x]{inputenc}
\usepackage[T1]{fontenc}
\usepackage{mathptmx}

\usepackage{fullpage}
\usepackage{amsmath}
\usepackage{mathrsfs}
\usepackage{scrpage2}
\usepackage{enumerate}
\usepackage{amsthm}
\usepackage{color}
\usepackage[usenames,dvipsnames,svgnames,table]{xcolor}
\usepackage{multicol}
\usepackage{mathtools}
\usepackage{etoolbox}
\usepackage{amssymb}
\usepackage{setspace}
\usepackage{amsmath,amssymb,amstext,amsthm,amscd,mathrsfs,eucal,bm,slashed}
\usepackage{graphicx,color}
\usepackage{array}
\usepackage{cite}
\usepackage{hyperref}
\hypersetup{
  pdftitle   = {Asymptotic Expansions and Bondi Positivity in Higher Dimensional Relativity},
  pdfkeywords = {Bondi mass, asymptotic flatness, higher dimensions, spinors, asymptotic symmetries, MPhil},
  pdfauthor  = {Alex Thorne},
  pdfcreator = {\LaTeX}
}

\allowdisplaybreaks
\newcommand{\defeq}{\vcentcolon=}

\newcommand{\scri}{\mathscr{I}}

\newcommand{\Lie}{\pounds}

\theoremstyle{definition}
\newtheorem{theorem}{Theorem}[section]
\newtheorem{lemma}[theorem]{Lemma}
\newtheorem{corollary}[theorem]{Corollary}
\newtheorem{note}{Note}
\theoremstyle{definition}
\newtheorem{definition}{Definition}
\theoremstyle{remark}
\newtheorem*{remark}{Remark}

\numberwithin{equation}{subsection}

\newtoggle{thesis}
\newtoggle{complete}

\let\originalsection\section
\renewcommand{\section}{\chapter}
\renewcommand{\subsection}{\originalsection}

\numberwithin{equation}{section}
\toggletrue{thesis}
\togglefalse{complete}

\begin{document}

\title{Asymptotic Expansions and Bondi Positivity in Higher Dimensional Relativity}
\subtitle{MPhil Thesis, 2013}
\author{Alexander Thorne
\thanks{School of Mathematics, Cardiff University, Senghennydd Road, Cardiff CF24 4AG, Wales, UK. \texttt{ThorneA@cardiff.ac.uk}}
}
\date{\today}
\maketitle


\input{./Abstract.tex}

\tableofcontents

\input{./Preliminary}

\input{./Introduction.tex}
\input{./AsymptoticFlatness.tex}
\input{./EinsteinEquations.tex}
\input{./AsymptoticSymmetries.tex}
\input{./BondiMass.tex}
\input{./Spinors.tex}

\input{./DiracEquation.tex}
\input{./BondiPositivity4D.tex}
\input{./Conclusion.tex}

\appendix
\input{./Notes.tex}
\input{./GRTensors.tex}
\input{./GRSpinors.tex}

\input{./Bibliography.tex}
\end{document}

%% file: Abstract.tex
\begin{minipage}{\linewidth}

\begin{abstract}
The positivity of the Bondi mass has been proven in 4 dimensions, but in higher dimensions the issue remains open.
The formalism of the present paper has been developed to investigate this and is well suited to the higher dimensional case.
At null infinity, we make asymptotic expansions of the metric components in conformal Gaussian null coordinates, and use the vacuum Einstein equations to solve for the expansion coefficients.
We find simple coordinate formulae for the Bondi mass, news and flux in terms of the expansion coefficients of the metric components.
We also make expansions of the generator of an asymptotic symmetry and obtain expressions for its expansion coefficients in terms of those of the metric components.
We make a spinorial asymptotic expansion of a solution to the Dirac equation in the four dimensional case, and use it to give a clean proof that the Bondi mass is positive in four dimensions.
The generalisation to higher dimensions has been investigated, but not yet resolved: indeed, the spinor mass appears to be divergent in the higher dimensional case, because of terms which appear to be singular in the limit at infinity.
Such apparently singular terms are common in many of our calculations, but can generally be shown to be zero before the limit is taken.
However, it is not yet clear this can be done in the case of the spinor mass in higher dimensions, so further investigation is required.
\end{abstract}

\renewcommand{\abstractname}{Acknowledgements}
\begin{abstract}
I would like to thank my supervisor Prof. Stefan Hollands for his help and guidance throughout my course.
I would also like to thank the School of Mathematics, Cardiff University for financial support.
\end{abstract}

\end{minipage}

%% file: Preliminary.tex
\iftoggle{complete}{
\subsection*{Issues}
Below is a list of issues still needing to be resolved.
\begin{itemize}
\item Understand and verify the definition of $\bar\psi$,
\begin{equation*}
\bar\psi = (-\psi^\star_-, \psi^\star_+)
\end{equation*}
\item Verify that $\nabla_a\bar\psi = \bar\nabla_a\psi$ for an arbitrary derivative operator.
\item Why do we need to impose $N^a\nabla_a\psi=0$ to relate $\nabla_r\psi$ to $\nabla_u\psi$?
Shouldn't setting $u=\frac{1}{2}r$ inherently relate the derivaties (at least the ordinary partial derivatives $\partial_r$ and $\partial_u$? 
Is $N^a\nabla_a\psi=0$ in fact just the generalisation of the conventional relationship between derivatives on a hypersurface to a curved space?
\item Check in coordinates that calculating $\nabla_a\xi^a$ yields the same result whether calculated in spinor or spinor free notation. Check this bearing in mind the latest corrections to the spinor free method (addition of missing Christoffel symbols).
\item Derive expression for connection coefficients $\omega_a$
\end{itemize}
}{}

\iftoggle{complete}{
\subsection*{Reading}
Below are listed the books and papers I have used during my course.

\subsubsection*{Papers}
\begin{itemize}
\item ``GR in the radiating regime'', Chrusciel, Franendiener, Friedrich
\end{itemize}

\subsubsection*{Books}
}{}

\iftoggle{complete}{
\subsection*{Uncertainty}
Different colours are used to mark varying levels of uncertainty in calculations and arguments.

\begin{center}
\begin{tabular}{ l | l }
Colour & Meaning\\
\hline
\color{gray}{gray} & Initial calculation, unchecked, highly uncertain\\
\color{green}{green} & Checked once with minor or no corrections, quite uncertain\\
\color{blue}{blue} & Checked several times with minor or no corrections, quite confident\\
\color{black}{black} & Checked extensively, with no corrections or mistakes found, highly confident\\
\color{orange}{orange} & Major corrections applied at least once, highly uncertain\\
\color{red}{red} & Contradictions found or major corrections ongoing (almost) certainly incorrect\\
\end{tabular}
\end{center}
}{}

%% file: Introduction.tex
\section{Introduction}
There are two measures of the total energy of an asymptotically flat spacetime.
The ADM (Arnowitt-Deser-Misner) mass \cite{Arnowitt} is a number defined at spatial infinity and represents the total energy of an idealised isolated system.
This quantity is readily generalised to the higher dimensional case.
The positivity of the ADM mass is a well established result \cite{Schoen79a,Witten81,Schoen79b,Schoen81}.

The Bondi mass \cite{Bondi62} is defined on cross sections of null infinity; it is a function of time in that it depends on the particular choice of cross section.
Intuitively, the Bondi mass represents the energy of an isolated system at a particular instant of time, after some energy has been lost by gravitational radiation through null infinity.
The positive mass conjecture for the Bondi mass states that the Bondi mass of a spacetime is positive on every cross section of null infinity.
Intuitively, this means that the energy of the isolated system does not become negative at any future time, i.e. more energy cannot be radiated away than was specified by the initial ADM mass of the system.
This has been proven in four dimensions \cite{Schoen82,Horowitz82,Ludvigsen81,Ludvigsen82},
but has not yet been generalised to the higher dimensional case.
Indeed, until recently, there was not even an expression for the Bondi mass in a higher dimensional asymptotically flat spacetime \cite{Hollands:2003ie}, see also \cite{Tanabe:2011es,Tanabe:2009va}.

The formalism of this paper has been developed to investigate the positive mass conjecture in even dimension\footnote{Odd spacetime dimensions $d$ are not considered since Penrose's framework of conformal null infinity requires $d$ even \cite{Hollands:2004ac}.} $d\geq4$.
A major difficulty in general is to obtain a systematic expansion of the Witten spinor, but the formalism here is well suited to this issue.
Indeed we can obtain the expansion of such a spinor in any even dimension $d\geq4$, though in higher dimensions the positivity result itself has not yet been settled.
In the four dimensional case, the spinorial expansions lead to a clean proof of positivity, given the global existence of the Witten spinor.
In higher dimensions, however, terms which are singular at null infinity arise in the limit for the spinor mass, causing it to be divergent, and the proof to fail.
If this is not due to a calculation error, then it would seem that a spinorial proof in higher dimensions is not possible.
This may indicate that the Bondi mass positivity result simply does not hold in the higher dimensional case.
On the other hand, it is possible that corrections to the calcultions will resolve the issue and the four dimensional proof will indeed generalise.

This paper does not address the existence of the Witten spinor, but focuses instead on determining its asymptotics under the assumption that a suitable spinor does exist.
Since our slice has the same geometry as in the AdS case \cite{Hertog:2005hm}, the existence proof is likely to be similar.
This has not, however, been considered in any detail here.

In the first half of this paper, the main goal is to obtain a simple coordinate expression for the Bondi mass in higher (even) dimensions.
Once the definition of asymptotic flatness has been settled, we construct Gaussian null coordinates \cite{Hollands:2006rj,Penrose} at null infinity and make asymptotic expansions of the metric components.
A major part of the analysis is to solve recursively for these expansions using the vacuum Einstein equations.
To ensure that the recursion can proceed, we must impose suitable conditions on the low order metric expansion coefficients (these conditions form part of our definition of asymptotic flatness).
By adapting the corresponding formulae of \cite{Hollands:2003ie} to our gauge, we arrive at suitable definitions of the Bondi mass, news and flux. We are then able to obtain simple coordinate expressions for these quantities using the asymptotic expansions of the metric coefficients.

In the second half of the paper, the main goal is to provide a spinorial proof of the positivity of the Bondi mass in four dimensions.
We make an asymptotic expansion of a solution to the Dirac equation and define the ``spinor mass'' \cite{Hollands:2005wt}, which can be shown to be positive by a standard argument.
The asymptotic expansion of the spinor leads to a coordinate expression for the spinor mass.
Since we now have coordinate formulae for both the Bondi mass and spinor mass, it is easy to see that they are equal, and therefore that the Bondi mass is positive.
The paper concludes with a discussion of the remaining open issues in the generalisation to higher dimensions, how they may be resolved and the consequences should this prove impossible.

%% file: AsymptoticFlatness.tex
\section{Asymptotic Flatness}
\label{sec:asymptotic-flatness}
Having defined asymptotic flatness in higher even dimension $d\geq4$, Gaussian null coordinates \cite{Hollands:2006rj,Penrose} are constructed at future null infinity.
Asymptotic expansions of the metric components are then made and conditions imposed on the lowest order expansion coefficients.

\subsection{Weak Asymptotic Flatness}
\label{sec:asymptotic-flatness-weak}
In this section, we introduce the  notion of ``weak'' asymptotic flatness. We specify the basic smoothness properties required of an asymptotically flat spacetime and the conformal transformation used to map to the unphysical spacetime.

\begin{definition}
\label{def:weak-asymptotic-flatness}
Let $g_{ab}$ be a Lorentzian metric on manifold $M$ of even dimension $d\geq4$. 
A conformal transformation with conformal factor $\Omega$ maps the ``physical'' spacetime $(M,g_{ab})$ to the ``unphysical'' spacetime ($\tilde{M},\tilde{g}_{ab})$ where $\tilde{g}_{ab}$ is given by $g_{ab}=\Omega^{-2} \tilde{g}_{ab}$.
We say the spacetime is \emph{weakly asymptotically flat} if we can attach a boundary $\scri$ to $M$ such that
\begin{enumerate}[(i)]
\item $\tilde{g}_{ab}$ is smooth on $\tilde{M}$ (up to and including $\scri$) ,
\item $\Omega$ is smooth on $\tilde{M}$ ,
\item $\Omega = 0$ and $\tilde\nabla_a\Omega \neq 0$ at $\scri$ .
\end{enumerate}
\end{definition}
\begin{remark}
It was shown in \cite{Anderson:2004pz} that there exists a large class of spacetimes which are weakly asymptotically flat under this definition.
\end{remark}
\begin{remark}
Throughout the paper, tensor quantities associated with the unphysical spacetime are marked with a tilde $\tilde{}\ $.
Indices on physical tensors (i.e. without a tilde) are raised and lowered with the physical metric $g_{ab}$, whereas indices on unphysical tensors (i.e. with a tilde) are raised and lowered with the unphysical metric $\tilde{g}_{ab}$.
\end{remark}
\begin{theorem}
If the spacetime is vacuum, $R_{ab}=0$, in a neighbourhood of null infinity, $\mathscr{I}$, then $\mathscr{I}$ can be split into two disjoint sets, $\scri^+$ and $\scri^-$, which we call future null infinity and past null infinity respectively.
\end{theorem}
\begin{proof}
Null infinity $\mathscr{I}$ is a surface of constant $\Omega$, so its normal is $n^a = g^{ab}\nabla_b\Omega$.
It follows from the vacuum Einstein equations that this normal is null \cite[pp.~278-9]{Wald}, and therefore $\mathscr{I}$ is a null surface.
Thus $\mathscr{I} = \partial\tilde{M}$ is the null boundary of $\tilde{M}$, and it follows \cite[pp.~222-3]{Hawking} that, near $\mathscr{I}$, $\tilde{M}$ is either in the past or the future of $\partial\tilde{M}$.
This means that $\mathscr{I}$ can be split into two disjoint sets, future null infinity $\mathscr{I}^+$ and past null infinity $\mathscr{I}^-$, on which future directed and past directed null geodesics respectively have their endpoints.
It can also be argued \cite[p.~223]{Hawking} that $\mathscr{I}$ consists of only the two connected sets $\mathscr{I}^+$ and $\mathscr{I}^-$.
\end{proof}
\begin{remark}
Throughout the remainder of the paper, calculations are performed at future null infinity, $\scri^+$.
However, analogous calculations could be performed at past null infinity, $\scri^-$.
\end{remark}

\subsection{Gaussian Null Coordinates}
\label{sec:asymptotic-flatness-coordinates}
In this section we construct so called ``Gaussian null coordinates'' \cite{Hollands:2006rj,Penrose} from null geodesics at future null infinity.
We choose these coordinates and fix our gauge in such a way that the metric takes on a relatively simple form.

\begin{theorem}
Future null infinity, $\scri^+$, is a null surface with normal $n^a = \tilde{g}^{ab}\tilde\nabla_b\Omega$.
The null geodesics generated by $n^a$ on $\scri^+$ are affinely parameterised and non intersecting.
\end{theorem}
\proof{Given in \cite[pp.~278-9]{Wald}}
\begin{remark}
The proof in \cite{Wald} shows that the geodesics have zero expansion, shear and twist \cite[p.~217]{Wald}.
This is why the last part of the theorem (that the geodesics do not intersect on $\scri^+$) is true.
\end{remark}

\begin{theorem}
\label{thm:gaussian-null-coordinates}
Let $(M,g_{ab})$ be a weakly asymptotically flat spacetime with corresponding unphysical spacetime $(\tilde{M},\tilde{g}_{ab})$.
We can construct Gaussian null coordinates $r,u,x^A$ at future null infinity\footnote{The coordinates are constructed from null geodesics and thus do not exist globally in general since these geodesics may overlap. We are only concerned with asymptotic expansions, however, so coordinates defined locally at future null infinity are sufficient.}
such that the physical metric takes the form
\begin{equation}
\label{eq:metric-gaussian-null-coordinates-nogauge}
g_{ab} = \Omega^{-2}\tilde{g}_{ab} 
       = \Omega^{-2}[2(dr_{(a} - \alpha du_{(a} - \beta_A dx^A_{(a})du_{b)} + \gamma_{AB} dx^A_a dx^B_b]\quad ,
\end{equation}
where $\alpha$, $\beta_A$ and $\gamma_{AB}$ are functions of $r,u,x^A$ defined on $\tilde{M}$ such that on $\scri^+$,
\begin{equation}
\label{eq:gnc-conditions-scri}
\alpha = 0 \ , \quad \beta_A = 0 \quad .
\end{equation}
\end{theorem}
\begin{proof}
Denote the affine parameter on geodesics of $n^a$ by $u$, so that
\begin{equation}
n^a = (\partial_u)^a \equiv \left(\frac{\partial}{\partial u}\right)^a
\end{equation}
is a coordinate vector.
Choose a cross section $\Sigma(0,0)$ of $\scri^+$ on which we set $u=0$ and choose $(d-2)$ ``angular'' coordinates\footnote{In Minkowski and Schwarzschild, these coordinates are the angular coordinates on a $(d-2)$ sphere. They are not angular coordinates in general, but for convenience we refer to them as ``angular''.} $x^A$ with corresponding coordinate vector fields
\begin{equation}
p_A^a = (\partial_A)^a \equiv \left(\frac{\partial}{\partial x^A}\right)^a \quad .
\end{equation}
The coordinates $x^A$ are extended to $\scri^+$ by parallel transport along the geodesics of $n^a$.
Since the vector fields $p_A^a$ are tangent to $\scri^+$ and $n^a$ is its normal,
\begin{equation}
\label{eq:n-pA-orthogonal}
\tilde{g}_{ab}n^a p_A^b = 0 \quad .
\end{equation}

\noindent
Let $l^a$ be a null vector on $\scri^+$ such that
\begin{equation}
\label{eq:gnc-conditions-l}
\tilde{g}_{ab} l^a n^b = 1 \ , \quad \tilde{g}_{ab} l^a p_A^b = 0 \quad .
\end{equation}
Such a vector exists and is unique
since it has $d$ components and we have imposed $d$ conditions.
The vector $l^a$ may be parallel transported along itself,
\begin{equation}
l^a\tilde\nabla_a l^b = 0 \quad ,
\end{equation}
to the rest of $\tilde{M}$ and it is null everywhere since
\begin{equation}
l^a \tilde\nabla_a ( \tilde{g}_{bc} l^b l^c ) = 2\tilde{g}_{bc} l^c l^a \tilde\nabla_a l^b = 0 \quad .
\end{equation}
We choose the affine parameter $r$ along geodesics of $l^a$ as a coordinate, so that
\begin{equation}
l^a = (\partial_r)^a \equiv \left(\frac{\partial}{\partial r}\right)^a
\end{equation}
is a coordinate vector, and set $r=0$ on $\scri^+$.
We have thus constructed coordinates $(r,u,x^A)$ on $\tilde{M}$.
They may in fact only be well defined in a neighbourhood of $\scri^+$, since the geodesics of $l^a$ may intersect globally.
This does not affect our calculations, however, and subsequently we shall say ``throughout $\tilde{M}$'' when we mean ``in a neighbourhood of $\mathscr{I}^+$''.
%
%

Finally, it follows from the nullness of $l^a$ that throughout $\tilde{M}$,
\begin{align}
l^a \tilde\nabla_a ( \tilde{g}_{bc} l^b n^c )
&= \tilde{g}_{bc} n^c l^a \tilde\nabla_a l^b + \tilde{g}_{bc} l^b l^a \tilde\nabla_a n^c \nonumber \\
&= \tilde{g}_{bc} l^b n^a \tilde\nabla_a l^c \quad , \ \text{since} \ [n^a, l^b] = 0 \nonumber \\
&= \frac{1}{2} n^a \tilde\nabla_a ( \tilde{g}_{bc} l^b l^c ) = 0 \quad , \\
l^a \tilde\nabla_a ( \tilde{g}_{bc} l^b p_A^c )
&= \tilde{g}_{bc} p_A^c l^a\tilde\nabla_a l^b + \tilde{g}_{bc} l^a l^b \tilde\nabla_a p_A^c \nonumber \\
&= \tilde{g}_{bc} l^b p_A^a \tilde\nabla_a l^c \quad , \ \text{since} \ [p_A^a, l^b] = 0 \nonumber \\
& = \frac{1}{2} p_A^a \tilde\nabla_a ( \tilde{g}_{bc} l^b l^c ) = 0 \quad .
\end{align}
Thus, conditions (\ref{eq:gnc-conditions-l}) hold throughout $\tilde{M}$ and the metric takes on the desired form.
It follows from the nullness of $n^a$ on $\scri^+$ that
\begin{equation}
\alpha = -\frac{1}{2}\tilde{g}_{ab} n^a n^b = 0 \quad ,
\end{equation}
on $\scri^+$,
and from equation (\ref{eq:n-pA-orthogonal}) that
\begin{equation}
\beta_A = -\tilde{g}_{ab} n^a p_A^b =0 \quad 
\end{equation}
on $\scri^+$,
where it was used that $n^a=(\partial_u)^a$ on $\mathscr{I}^+$.
\end{proof}
\begin{theorem}
\label{thm:gaussian-null-coordinates-gauge}
The conformal factor and coordinates $r,u,x^A$ can be chosen to give so called ``conformal Gaussian null coordinates'' \cite{Ishibashi:2007kb} in which the conformal factor is $\Omega=r$ and the metric takes the form
\begin{equation}
g_{ab} = r^{-2}\tilde{g}_{ab} = r^{-2}[2(dr_{(a} - \alpha du_{(a} - \beta_A dx^A_{(a})du_{b)} + \gamma_{AB} dx^A_a dx^B_b] \quad .
\end{equation}
\end{theorem}
\begin{proof}
\medskip\noindent
We want to make a conformal coordinate transformation,
\begin{subequations}
\begin{align}
\Omega &\rightarrow r' \quad, \\
r &\rightarrow r' = r'(r,u,x^A) \quad , \\
u &\rightarrow u' = u \quad , \\
x^A &\rightarrow x'{}^A = x^A \quad .
\end{align}
\end{subequations}
(Note that $r$ is the only coordinate which is changed.)
Under this transformation, we require that the metric,
\begin{equation}
g_{ab} = \Omega^{-2}\tilde{g}_{ab} = \Omega^{-2}[2(dr_{(a} - \alpha du_{(a} - \beta_A dx^A_{(a})du_{b)} + \gamma_{AB} dx^A dx^B] \quad , \\
\end{equation}
becomes
\begin{align}
g_{ab} = r'{}^{-2}\tilde{g}'_{ab} = r'{}^{-2}[2(dr'_{(a} - \alpha' du'_{(a} - \beta'_A dx'{}^A_{(a})du'_{b)} + \gamma'_{AB} dx'{}^A dx'{}^B] \quad .
\end{align}
Note that $\tilde{g}'_{ab} = (r'/\Omega)^2 \tilde{g}_{ab}$.
Comparing the leading terms of the two expressions and contracting with $l^a n^b$, we find (since $du_a=du_a'$)
\begin{align}
\label{eq:asymptotic-flatness-coordinates-conformal-gauge-condition}
\frac{\partial r'}{\partial r} &= \frac{r'{}^2}{\Omega^2} \quad .
\end{align}
Defining $\omega = \Omega/r$, this can be rewritten as
\begin{equation}
\frac{\partial(1/r')}{\partial(1/r)} = \frac{1}{\omega^2} \quad .
\end{equation}
Integration yields a solution,
\begin{equation}
1/r' = h + \int_{1/r_0}^{1/r} \frac{1}{\omega^2} \ d(1/\dot{r}) \quad ,
\end{equation}
where the integration constant $h$ is a function of $u,x^A$ given by
\begin{equation}
h=1/r'\big|_{r=r_0} \quad .
\end{equation}
Provided the integral exists, this gives a function $r' = r'(r,u,x^A)$ in the interior of $\tilde{M}$ up to choice of $h$
(we choose not to remove this gauge freedom and leave $h$ unspecified).
We can smoothly extend $r'$ to $\scri^+$ by setting it to $0$ there.
Now we must check that this conformal transformation preserves conditions (\ref{eq:gnc-conditions-scri}) and (\ref{eq:gnc-conditions-l}).
The coordinate vector fields transform under the vector transformation law as
\begin{subequations}
\begin{align}
\label{eq:l-prime-halfway}
l'{}^a 
\iftoggle{complete}{
= \left(\frac{\partial}{\partial r'}\right)^a
       &= \left(\frac{\partial}{\partial r}\right)^a \frac{\partial r}{\partial r'}
	+ \left(\frac{\partial}{\partial u}\right)^a \frac{\partial u}{\partial r'}
	+ \left(\frac{\partial}{\partial x^A}\right)^a \frac{\partial x^A}{\partial r'}\\
       &= \left(\frac{\partial}{\partial r}\right)^a \frac{\partial r}{\partial r'}
	u
	+ \left(\frac{\partial}{\partial u}\right)^a \frac{\partial u'}{\partial r'}
	+ \left(\frac{\partial}{\partial x^A}\right)^a \frac{\partial x'{}^A}{\partial r'}\\
}{}
       &= \frac{\partial r}{\partial r'} l^a \quad , \\
\label{eq:n-prime-halfway}
n'{}^a 
\iftoggle{complete}{
= \left(\frac{\partial}{\partial u'}\right)^a
       &= \left(\frac{\partial}{\partial r}\right)^a \frac{\partial r}{\partial u'}
	+ \left(\frac{\partial}{\partial u}\right)^a \frac{\partial u}{\partial u'}
	+ \left(\frac{\partial}{\partial x^A}\right)^a \frac{\partial x^A}{\partial u'}\\
       &= \left(\frac{\partial}{\partial r}\right)^a \frac{\partial r}{\partial u'}
	+ \left(\frac{\partial}{\partial u}\right)^a \frac{\partial u'}{\partial u'}
	+ \left(\frac{\partial}{\partial x^A}\right)^a \frac{\partial x'{}^A}{\partial u'}\\
}{}
       &= \frac{\partial r}{\partial u'}l^a + n^a \quad , \\
\label{eq:pA-prime-halfway}
p'_A{}^a
\iftoggle{complete}{
= \left(\frac{\partial}{\partial x'{}^A}\right)^a
       &= \left(\frac{\partial}{\partial r}\right)^a \frac{\partial r}{\partial x'{}^A}
	+ \left(\frac{\partial}{\partial u}\right)^a \frac{\partial u}{\partial x'{}^A}
	+ \left(\frac{\partial}{\partial x^B}\right)^a \frac{\partial x^B}{\partial x'{}^A}\\
       &= \left(\frac{\partial}{\partial r}\right)^a \frac{\partial r}{\partial x'{}^A}
	+ \left(\frac{\partial}{\partial u}\right)^a \frac{\partial u'}{\partial x'{}^A}
	+ \left(\frac{\partial}{\partial x^B}\right)^a \frac{\partial x'{}^B}{\partial x'{}^A}\\
       &= \frac{\partial r}{\partial x'{}^A} l^a + \delta^B{}_A p_B^a\\
}{}
       &= \frac{\partial r}{\partial x'{}^A} l^a + p_A^a \quad ,
\end{align}
\end{subequations}
where we used that $u=u'$,$x^A=x'{}^A$.
Since $r=0$ on $\scri^+$,
\begin{equation}
\frac{\partial r}{\partial u'} = \frac{\partial r}{\partial x'{}^A} = 0
\end{equation}
on $\scri^+$, and therefore by equations (\ref{eq:n-prime-halfway}) and (\ref{eq:pA-prime-halfway}), $n'{}^a = n^a$ and $p'_A{}^a = p_A^a$ on $\scri^+$.
It then follows that on $\scri^+$,
\begin{align}
& \tilde{g}'_{ab} n'{}^a n'{}^b = (r'/\Omega)^2 \ \tilde{g}_{ab} n^a n^b = 0 \quad , \\
& \tilde{g}'_{ab} n'{}^a p'_A{}^b = (r'/\Omega)^2 \ \tilde{g}_{ab} n^a p_A^b = 0 \quad .
\end{align}
That is, $n'{}^a$ is null and orthogonal to $p'_A{}^a$ at $\scri^+$.
Thus, conditions (\ref{eq:gnc-conditions-scri}) still hold: $\alpha'=\beta'_A=0$ on $\scri^+$.

\noindent
Now, by equation (\ref{eq:asymptotic-flatness-coordinates-conformal-gauge-condition}), equation (\ref{eq:l-prime-halfway}) becomes
\begin{equation}
\label{eq:l-prime}
l'{}^a = \left(\frac{\partial r}{\partial r'}\right) \ l^a = \frac{\Omega^2}{r'{}^2} \ l^a \quad .
\end{equation}
It follows that $l'{}^a$ is null throughout $\tilde{M}$, since
\begin{align}
\tilde{g}'_{ab}l'{}^al'{}^b = \frac{\Omega^2}{r'{}^2} \tilde{g}_{ab} l^a l^b = 0 \quad .
\end{align}
In addition, on $\scri^+$, we see that
\begin{align}
& \tilde{g}'_{ab} l'{}^a p'_A{}^b = \tilde{g}_{ab} l^a p_A^b = 0 \quad ,
& \tilde{g}'_{ab} l'{}^a n'{}^b = \tilde{g}_{ab} l^a n^b = 1 \quad .
\end{align}
These two properties are precisely conditions (\ref{eq:gnc-conditions-l}), and it remains only to check that they are satisfied throughout $\tilde{M}$.
To do this, we show that they are parallel transported along the geodesics of $l'{}^a$.

Under the conformal transformation, the derivative operator obeys \cite[p.~446]{Wald},
\begin{equation}
\tilde\nabla'_a t^b = \tilde\nabla_a t^b + \tilde{C}^b_{ac} t^c \quad ,
\end{equation}
where
\begin{align}
\label{eq:conformal-derivative}
\tilde{C}^b{}_{ac}
\iftoggle{complete}{
&= \frac{1}{2}\tilde{g}'{}^{bd}(\tilde\nabla_a\tilde{g}'_{cd} + \tilde\nabla_c\tilde{g}'_{ad} - \tilde\nabla_d\tilde{g}'_{ac})\\
&= \frac{1}{2}\left(\frac{r'}{\Omega}\right)^{-2}\tilde{g}^{bd}\left[4\tilde{g}_{d(c}\frac{r'}{\Omega}\tilde\nabla_{a)}\frac{r'}{\Omega} - 2\tilde{g}_{ac}\frac{r'}{\Omega}\tilde\nabla_d\frac{r'}{\Omega}\right]\\
}{}
&= \frac{\Omega}{r'}\left[2\delta^b{}_{(a}\tilde\nabla_{c)}\frac{r'}{\Omega} - \tilde{g}_{ac}\tilde{g}^{bd}\tilde\nabla_d\frac{r'}{\Omega}\right] \quad .
\end{align}

\medskip\noindent
We can now calculate using equations (\ref{eq:l-prime}) and (\ref{eq:conformal-derivative}) that $l'{}^a$ satisfies the affinely parametrised geodesic equation,
\begin{align}
l'{}^a \tilde\nabla'_a l'{}^b 
\iftoggle{complete}{ \nonumber \\
& = l'{}^a\tilde\nabla'_a\left(\frac{\Omega^2}{r'{}^2}l^b\right) 
&= l^b l'{}^a \tilde\nabla_a \left(\frac{\Omega}{r'}\right)^2 + \left(\frac{\Omega}{r'}\right)^4 l^a \tilde\nabla'_a l^b\\
&= l^b \frac{\partial}{\partial r'} \left(\frac{\Omega}{r'}\right)^2 + \left(\frac{\Omega}{r'}\right)^4 l^a \tilde\nabla_a l^b + \left(\frac{\Omega}{r'}\right)^4 l^a l^c \tilde{C}^b{}_{ac}\\
&= -2 l^b \left(\frac{\Omega}{r'}\right)^3 \frac{\partial}{\partial r'} \left(\frac{r'}{\Omega}\right) + 2 \left(\frac{\Omega}{r'}\right)^5 l^a l^c \delta^b{}_{(a} \tilde\nabla_{c)} \left(\frac{r'}{\Omega}\right) - \left(\frac{\Omega}{r'}\right)^5 \tilde{g}_{ac} l^a l^c \tilde{g}^{bd} \tilde\nabla_d \left(\frac{r'}{\Omega}\right)\\
&= -2 l^b \left(\frac{\Omega}{r'}\right)^3 \frac{\partial}{\partial r'} \left(\frac{r'}{\Omega}\right) + 2 \left(\frac{\Omega}{r'}\right)^5 l^a l^c \delta^b{}_a \tilde\nabla_c \left(\frac{r'}{\Omega}\right)\\
&= -2 l^b \left(\frac{\Omega}{r'}\right)^3 \frac{\partial}{\partial r'} \left(\frac{r'}{\Omega}\right) + 2 l^b \left(\frac{\Omega}{r'}\right)^3 \frac{\partial}{\partial r'} \left(\frac{r'}{\Omega}\right)\\
& }{}
 = 0 \quad ,
\end{align}
where we used that $l^a\tilde\nabla_a l^b=0$.
We can now show, as we did in the proof of Theorem \ref{thm:gaussian-null-coordinates}, that
\begin{align}
& l'{}^a \tilde\nabla'_a (\tilde{g}'_{bc} l'{}^b n'{}^c) = 0 \quad , \\
& l'{}^a \tilde\nabla'_a (\tilde{g}'_{bc} l'{}^b p'_A{}^c ) = 0 \quad .
\end{align}
This means that $\tilde{g}'_{ab}l'{}^an'{}^b$ and $\tilde{g}'_{ab}l'{}^ap'_A{}^b$ are preserved along the geodesics of $l'{}^a$.
Since conditions (\ref{eq:gnc-conditions-l}) are satisfied on $\scri^+$, it follows that they are true throughout $\tilde{M}$.
The metric can then be written in the desired form,
\begin{equation}
g_{ab} = r^{-2}\tilde{g}_{ab} = r{}^{-2}[2(dr_{(a} - \alpha du_{(a} - \beta_A dx{}^A_{(a})du_{b)} + \gamma_{AB} dx_a{}^A dx_b{}^B] \quad ,
\end{equation}
where $\alpha = \beta_A = 0$ at $\scri^+$, and we have omitted the $'$s for brevity.
\end{proof}

The upper case Roman indices on $\beta_A$ and $\gamma_{AB}$ are just labelling indices.
However, on each surface $\Sigma(r,u)$ of constant $(r,u)$, $\gamma_{AB}$ is the induced metric and such indices can be viewed as tensor indices.
It is helpful to introduce abstract index notation for these tensors: upper case Roman indices will be raised and lowered with $\gamma_{AB}$.

\begin{definition}
\label{def:induced-metric}
The tensor $\gamma_{AB}$ is the induced metric on surfaces $\Sigma(r,u)$ of constant $(r,u)$.
\end{definition}
\begin{remark}
If we view $\gamma_{AB}$ as a metric, we can define an associated derivative operator and Ricci tensor.
\end{remark}
\begin{definition}
\label{def:gamma-derivative}
On the surfaces $\Sigma(r,u)$, the derivative operator $D_A$ such that $D_A\gamma_{BC}=0$ acts on a covector $t_B$ on $\Sigma(r,u)$ as
\begin{equation}
\label{eq:gamma-derivative}
D_A t_B = \partial_A t_B - \tilde\Lambda^C{}_{AB}t_C \quad ,
\end{equation}
\label{eq:gamma-christoffel}
where the Christoffel symbols are given by
\begin{equation}
\tilde\Lambda^C{}_{AB} = \frac{1}{2}\gamma^{CD}(\partial_A\gamma_{BD} + \partial_B\gamma_{AD} - \partial_D\gamma_{AB}) \quad .
\end{equation}
\end{definition}
\begin{remark}
When the normal covariant derivative $\nabla_a$ acts on $\beta_A$ or $\gamma_{AB}$, they are treated as scalars.
Only the derivative $D_A$ acts on them as tensors.
\end{remark}

\begin{definition}
\label{def:gamma-ricci}
On the surfaces $\Sigma(r,u)$, the Ricci tensor of $\gamma_{AB}$ on the surfaces $\Sigma(r,u)$ is given by
\begin{equation}
\label{eq:gamma-ricci}
\mathcal{R}_{AB} = \partial_C\tilde\Lambda^C{}_{AB} - \partial_A\tilde\Lambda^C{}_{CB} + \tilde\Lambda^D{}_{AB}\tilde\Lambda^C{}_{CD} - \tilde\Lambda^D{}_{BC}\tilde\Lambda^C{}_{AD} \quad .
\end{equation}
\end{definition}

\subsection{Strong Asymptotic Flatness}
\label{sec:asymptotic-flatness-strong}
In this section, we introduce asymptotic expansions of the metric components.
We then specify certain conditions on the coefficients of these expansions in our definition of ``strong'' asymptotic flatness.
In general, to ensure that the Bondi mass is well defined, the definition of asymptotic flatness must be sufficiently strict; however, the definition must not be so strict that it excludes radiating spacetimes altogether.
In our case, this corresponds to choosing boundary conditions on the metric coefficients which provide just enough ``initial data'' to be able to obtain the metric expansion from the Einstein equations (see Section \ref{sec:einstein-equations}), without being too restrictive.

The asymptotic expansions of the metric components $\alpha$, $\beta_A$ and $\gamma_{AB}$, are given by
\begin{subequations}
\label{eq:metric-component-expansions}
\begin{align}
\label{eq:metric-component-expansions-alpha}
\alpha &\sim \sum_{n=0}^\infty r^n \alpha^{(n)} \quad , \\
\label{eq:metric-component-expansions-beta}
\beta_A &\sim \sum_{n=0}^\infty r^n \beta_A^{(n)} \quad , \\
\label{eq:metric-component-expansions-gamma}
\gamma_{AB} &\sim \sum_{n=0}^\infty r^n \gamma_{AB}^{(n)} \quad ,
\end{align}
\end{subequations}
where the coefficients $\alpha^{(n)}$, $\beta_A^{(n)}$ and $\gamma_{AB}^{(n)}$ ($n \geq 0$) are functions of $u,x^A$ on $\tilde{M}$.
The expansion coefficients can be calculated for example by
\begin{equation}
\alpha^{(n)} = \frac{1}{n!}\left[ \partial_r^n \alpha \right]_{r=0} \quad ,
\end{equation}
and in particular, $\alpha^{(0)} = \alpha |_{r=0}$.
The coefficients for the other expansions can be calculated in a similar manner.
Note that in fact the asymptotic expansions are only needed to finite order (we see later that no coefficients are used beyond $\alpha^{(d-1)}$), so we do not consider the convergence of the asymptotic series.

\noindent
The Ricci tensor and Christoffel symbols of $\gamma_{AB}$ can also be expanded,
\begin{subequations}
\begin{align}
& \mathcal{R}_{AB} = \sum_{n=0}^\infty r^n\mathcal{R}_{AB}^{(n)} \quad ,
&
& \tilde\Lambda^C{}_{AB} = \sum_{n=0}^\infty r^n \tilde\Lambda^C{}_{AB}^{(n)} \quad ,
\end{align}
\end{subequations}
where these expansion coefficients are given in terms of the expansion coefficients of $\gamma_{AB}$.
The asymptotic expansion of the Ricci scalar $\mathcal{R}=\gamma^{AB}\mathcal{R}_{AB}$ can be calculated from the expansions of $\gamma^{AB}$ and $\mathcal{R}_{AB}$.
We denote the derivative operator associated with the lowest order term $\gamma_{AB}^{(0)}$ by $\mathscr{D}_A$.
So $\mathscr{D}_A\gamma_{BC}^{(0)}=0$ and the Christoffel symbols of $\mathscr{D}_A$ are $\tilde\Lambda^C{}_{AB}^{(0)}$.

The following definition is useful since it formalises the notion of ``irrelevant'' terms (containing very high powers of $r$) which do not affect our calculations.
\begin{definition}
\label{def:order-tensor}
We say a tensor $T^{a\ldots}{}_{b\ldots}$ is \textit{of order $r^k$} for some integer $k$, and write
\begin{equation}
\label{eq:order-tensor-definition}
T^{a\ldots}{}_{b\ldots} = O(r^k) \quad ,
\end{equation}
if the tensor $r^{-k}T^{a\ldots}{}_{b\ldots}$ is smooth at $r=0$.
\end{definition}
\begin{remark}
If there are two tensors such that
\begin{align}
& S^{a\ldots}{}_{b\ldots}=O(r^m) \quad , \\
& T^{a\ldots}{}_{b\ldots}=O(r^n) \quad , \end{align}
then the product $(ST)^{a\ldots}{}_{b\ldots} = O(r^{m+n})$.
Moreover, $\left(l^c\nabla_c\right)^k T^{a\ldots}{}_{b\ldots}=O(r^{n-k})$ where $\left(l^c\nabla_c\right)^k$ denotes taking $k$ derivatives in the $r$ direction.
\end{remark}

\begin{definition}
\label{def:strong-asymptotic-flatness}
We call a weakly asymptotically flat spacetime $(M,g_{ab})$ asymptotically Einstein flat
if there exists a constant $\lambda>0$ and Einstein metric $s_{AB}$ on the surfaces $\Sigma(r,u)$ of constant $(r,u)$ which is independent\footnote{The metric $s_{AB}$ can be viewed as being defined on $\Sigma(0,0)$, parallel transported to the rest of $\mathscr{I}^+$ along the geodesics of $n^a$ and then parallel transported to the rest of $\tilde{M}$ by parallel transport along the geodesics of $l^a$. In coordinates, this of course means that $s_{AB}$ has no $r$ or $u$ dependence.} of $r,u$, such that
\begin{enumerate}[(i)]
\item $\alpha^{(1)} = 0$ \quad ,
\item $\alpha^{(2)} = \frac{\lambda}{2}$ \quad ,
\item $\mathcal{R}_{AB}^{(0)} = \lambda(d-3) s_{AB}$ \quad ,
\item \label{cond:strong-asymptotic-flatness-early-time} there exist $\delta>0, u_0\in\mathbb{R}$ such that $\gamma_{AB}\big|_{u=u_0} = s_{AB} + O(r^N)$ for $r<\delta$
\end{enumerate}
\end{definition}
\begin{remark}
If the Einstein metric $s_{AB}$ is such that $\lambda\leq0$, then the present formalism does not apply.
Indeed the coordinate vector $(\partial_u)^a$ has normal
\begin{equation}
\tilde{g}_{ab} (\partial_u)^a (\partial_u)^b = -2\alpha = -\lambda r^2 + O(r^3) \quad .
\end{equation}
It follows that, near $\mathscr{I}^+$, $(\partial_u)^a$ is spacelike for $\lambda<0$, null for $\lambda=0$ and timelike for $\lambda>0$.
In the first two cases, the metric would no longer have Lorentzian signature, meaning that the Gaussian null coordinates are not well defined in these cases.
Therefore, we require $\lambda>0$ to ensure that $u$ is a well defined timelike coordinate.
\end{remark}

%% file: EinsteinEquations.tex
\section{Einstein Equations}
\label{sec:einstein-equations}
The asymptotic expansions are substituted into the Einstein equations in order to investigate the asymptotic behaviour of the metric.

\subsection{Writing down the Einstein equations}
\label{sec:einstein-equations-calculation}
The next step is to solve for the metric expansion using the Einstein equations, but we must first express the physical vacuum Einstein equation in terms of unphysical tensors.

We can use the expression for the metric in Gaussian null coordinates to write down the Einstein equations. 
We first compute the unphysical Christoffel symbols,
\begin{equation}
\tilde\Gamma^c{}_{ab} = \frac{1}{2}\tilde{g}^{cd}(\partial_a\tilde{g}_{bd} + \partial_b\tilde{g}_{ad} - \partial_d\tilde{g}_{ab}) \quad ,
\end{equation}
so that we can calculate with the unphysical derivative operator,
\begin{equation}
\tilde\nabla_a\omega_b = \partial_a\omega_b - \tilde\Gamma^c{}_{ab}\omega_c \quad .
\end{equation}
Next we calculate the unphysical Ricci tensor,
\begin{equation}
\tilde{R}_{ab} = \partial_c\tilde\Gamma^c{}_{ab} - \partial_a\tilde\Gamma^c{}_{cb} + \tilde\Gamma^d{}_{ab}\tilde\Gamma^c{}_{cd} - \tilde\Gamma^d{}_{cb}\tilde\Gamma^c{}_{da} \quad ,
\end{equation}
and the unphysical Schouten tensor,
\begin{equation}
\label{eq:einstein-equations-calculation-schouten-tensor}
\tilde{S}_{ab} \defeq \frac{2}{d-2}\tilde{R}_{ab} - \frac{1}{(d-1)(d-2)}\tilde{R}\tilde{g}_{ab} \quad .
\end{equation}
We impose the physical vacuum Einstein equation,
\begin{equation}
\iftoggle{complete}{G_{ab} \equiv R_{ab} - \frac{1}{2}Rg_{ab} \iff }{}  R_{ab} = 0 \quad ,
\end{equation}
which we must express in terms of unphysical tensors. 
Since the physical metric is given by
\begin{equation}
g_{ab} = (r^{-1})^2\tilde{g}_{ab} \quad ,
\end{equation}
it follows \cite[pp.~445-6]{Wald} that,
\begin{align}
R_{ab} = \tilde{R}_{ab} 
& - (d-2)\tilde\nabla_a\tilde\nabla_b\ln(r^{-1}) - \tilde{g}_{ab}\tilde{g}^{cd}\tilde\nabla_c\tilde\nabla_d\ln(r^{-1}) \nonumber \\
& + (d-2)(\tilde\nabla_a\ln(r^{-1}))(\tilde\nabla_b\ln(r^{-1})) - (d-2)\tilde{g}_{ab}\tilde{g}^{cd}(\tilde\nabla_c\ln(r^{-1}))\tilde\nabla_d\ln(r^{-1}) \quad .
\end{align}
\iftoggle{complete}{
Now,
\begin{align*}
\tilde\nabla_a\ln(r^{-1}) &= r\tilde\nabla_a(r^{-1}) = -r^{-1}\tilde\nabla_a r\\
\tilde\nabla_a\tilde\nabla_b\ln(r^{-1}) &= \tilde\nabla_a(-r^{-1}\tilde\nabla_b r) = -r^{-1}\tilde\nabla_a\tilde\nabla_b r + r^{-2}(\tilde\nabla_a r)\tilde\nabla_b r
\end{align*}
}{}
Therefore, the vacuum Einstein equation can be written \cite[p.~278]{Wald}
\begin{equation}
\label{eq:einstein-equations-calculation-equation-ricci}
0 = \tilde{R}_{ab} + (d-2)r^{-1}\tilde\nabla_a\tilde\nabla_b r + \tilde{g}_{ab}\tilde{g}^{cd}[r^{-1}\tilde\nabla_c\tilde\nabla_d r - (d-1)r^{-2}(\tilde\nabla_c r)\tilde\nabla_d r] \quad .
\end{equation}
Taking the trace yields an expression for the unphysical scalar curvature
\begin{equation}
\tilde{R} = -2(d-1)r^{-1}\tilde{g}^{cd}\tilde\nabla_c\tilde\nabla_d r + d(d-1)r^{-2}\tilde{g}^{cd}(\tilde\nabla_c r)\tilde\nabla_d r \quad .
\end{equation}
\iftoggle{complete}{
Rearranging equation (\ref{eq:einstein-equations-calculation-schouten-tensor}) gives
\begin{equation*}
\tilde{R}_{ab} = \frac{d-2}{2}\tilde{S}_{ab} + \frac{1}{2(d-2)}\tilde{R}\tilde{g}_{ab}
\end{equation*}
}{}
In terms of the Schouten tensor, the vacuum Einstein equation thus becomes \cite{Hollands:2003ie}
\begin{equation}
0 = \tilde{S}_{ab} + 2r^{-1}\tilde\nabla_a\tilde\nabla_b r - r^{-2}\tilde{g}_{ab}\tilde{g}^{cd}(\tilde\nabla_c r)\tilde\nabla_d r \quad .
\end{equation}
For convenience, we define the tensor
\begin{equation}
\label{eq:einstein-tensor-short}
\tilde{E}_{ab} = \tilde{S}_{ab} + 2r^{-1}\tilde\nabla_a\tilde\nabla_b r - r^{-2}\tilde{g}_{ab}\tilde{g}^{cd}(\tilde\nabla_c r)\tilde\nabla_d r \quad ,
\end{equation}
so that the Einstein equations may be written concisely as
\begin{equation}
R_{ab} \equiv \tilde{E}_{ab} = 0 \quad .
\end{equation}
The complete expressions for the tensors $\tilde{R}_{ab}$ $\tilde{S}_{ab}$ and $\tilde{E}_{ab}$ and the scalar curvature $\tilde{R}$ can be found in Appendix \ref{app:grtensors}.

\subsection{Solving the Einstein Equations}
\label{sec:einstein-equations-solution}

Now that we have written down the Einstein equations, we can substitute in the expansions of the metric components to investigate the asymptotic behaviour of the metric.
The results of our analysis are given below in the following theorem.

\begin{theorem}
\label{thm:metric-expansion-low-order}
Let the integer $N$ be defined by
\begin{equation}
N \defeq \frac{d-2}{2} \quad .
\end{equation}
For an asymptotically Einstein flat metric, the expansion coefficients satisfy
\begin{subequations}
\label{eq:metric-components-low-order}
\begin{align}
\alpha &= \frac{\lambda}{2}r^2 + O(r^{N+2}) \quad , \\
\beta_A &= O(r^{N+1}) \quad , \\
\gamma_{AB} &= s_{AB} + O(r^N) \quad .
\end{align}
For $d>4$ we also have
\begin{equation}
\label{eq:metric-components-low-order-trace}
s^{AB}\gamma_{AB} = (d-2) + O(r^{2N}) \quad .
\end{equation}
\end{subequations}
\end{theorem}
\begin{remark}
Equations (\ref{eq:metric-components-low-order}a-c) are in agreement with existing definitions \cite{Hollands:2003ie} of asymptotic flatness in higher even dimensions.
Equation (\ref{eq:metric-components-low-order-trace}), however, is to our knowledge an entirely new result, which is extremely useful in proving the convergence of the Bondi mass formula in Section \ref{sec:bondi-mass}.
\end{remark}
\begin{remark}
The constant $N$ has been defined for more concise notation.
It is only used as a superscript to denote expansion coefficients; elsewhere constants are still written in terms of $d$.
\end{remark}

\begin{corollary}
\label{thm:background-metric}
The asymptotic form of the physical metric can now be written
\begin{equation}
\label{eq:metric-expansion-low-order}
g_{ab}  = \bar{g}_{ab} + O(r^\frac{d-2}{2}) du_a du_b + O(r^\frac{d-4}{2}) du_{(a} dx^A_{b)} + O(r^\frac{d-6}{2}) dx^A_{(a} dx^B_{b)} \quad ,
\end{equation}
where the ``background'' metric $\bar{g}_{ab}$ is given by
\begin{equation}
\label{eq:einstein-equations-background-metric}
\bar{g}_{ab} = r^{-2}[2dr_{(a}du_{b)} - r^2\lambda du_a du_b + s_{AB} dx^A_a dx^B_b] \quad .
\end{equation}
\end{corollary}
\begin{remark}
It is shown in Theorem \ref{thm:minkowski-metric} that the metric $\bar{g}_{ab}$ is simply the unphysical Minkowski metric.
For a comparison of the asymptotic behaviour shown here with other definitions of asymptotic flatness, see Note \ref{note:asymptotic-flatness} (Appendix \ref{app:notes}).
\end{remark}
\begin{proof}
By Theorem \ref{thm:gaussian-null-coordinates-gauge}, the physical metric can be written
\begin{equation}
g_{ab} = r^{-2}\tilde{g}_{ab} = r^{-2}[2(dr_{(a} - \alpha du_{(a} - \beta_A dx^A_{(a})du_{b)} + \gamma_{AB} dx^A_a dx^B_b] \quad .
\end{equation}
Substituting in the expressions for $\alpha$, $\beta_A$ and $\gamma_{AB}$ given in equations (\ref{eq:metric-components-low-order}) and using equation (\ref{eq:einstein-equations-background-metric}) to replace the lowest order terms with $\bar{g}_{ab}$ gives the desired expression (\ref{eq:metric-expansion-low-order}).
\end{proof}
\begin{corollary}
The asymptotic properties of the unphysical metric can be expressed in terms of the deviation $\tilde{X}_{ab}$ of the metric from the unphysical background $\tilde{\bar{g}}_{ab}$.
The deviation is given by
\begin{equation}
\tilde{X}_{ab} \defeq \tilde{g}_{ab} - \tilde{\bar{g}}_{ab} \quad ,
\end{equation}
and the asymptotic conditions can be written as
\begin{subequations}
\label{eq:metric-deviation-asymptotics}
\begin{align}
& \tilde{X}_{rr} = \tilde{X}_{ru} = \tilde{X}_{rA} = 0 \quad , \\
& \tilde{X}_{uu} = O(r^\frac{d+2}{2}) \quad , \\
& \tilde{X}_{uA} = O(r^\frac{d}{2}) \quad , \\
& \tilde{X}_{AB} = O(r^\frac{d-2}{2}) \quad , \\
& \tilde{X}^A{}_A = O(r^{d-2}) \quad .
\end{align}
\end{subequations}
\end{corollary}
\begin{proof}
This follows immediately by decomposing equation (\ref{eq:metric-expansion-low-order}) into components.
\end{proof}

\medskip\noindent
We now introduce lemmas required in the proof of Theorem \ref{thm:metric-expansion-low-order}.

\begin{lemma}
\label{thm:gamma-order-k}
For an integer $0\leq k<N$,
\begin{equation}
\label{eq:gamma-asymptotics-early-time}
\gamma_{AB}^{(k)}\big|_{u=u_0}=
\begin{cases}
s_{AB} & k=0 \\
0      & k>0
\end{cases} \quad ,
\end{equation}
where $u_0\in\mathbb{R}$ is the constant defined in condition (iv) of Definition \ref{def:strong-asymptotic-flatness}.
\end{lemma}
\begin{proof}
By the definition of asymptotic Einstein flatness, there exists a constant $\delta>0$ such that
\begin{equation}
\gamma_{AB} \big|_{u=u_0} = s_{AB} + O(r^N) \quad ,
\end{equation}
for $r<\delta$.
Setting $r=0$ gives $\gamma_{AB}^{(0)}\big|_{u=u_0}=s_{AB}$.
For $0<k<N$, taking the $r$ derivative $k$ times and setting $r=0$, we get $k!\gamma_{AB}^{(k)}\big|_{u=u_0} = 0$.
\end{proof}

\begin{corollary}
\label{thm:gamma-remove-u-derivative}
If, for some integer $0\leq k<N$, $\partial_u\gamma_{AB}^{(k)}=0$, then
\begin{equation}
\label{eq:gamma-k}
\gamma_{AB}^{(k)}=
\begin{cases}
s_{AB} & k=0 \\
0      & k>0
\end{cases} \quad ,
\end{equation}
throughout $\tilde{M}$.
\end{corollary}
\begin{remark}
It turns out that the Einstein equations only give us information about the $u$ derivative of the expansion coefficients of $\gamma_{AB}$.
Using this corollary, we can obtain conditions on the coefficients themselves from conditions on their $u$ derivatives.
\end{remark}
\begin{proof}
By Lemma \ref{thm:gamma-order-k}, $\gamma_{AB}^{(k)}$ is given by equation (\ref{eq:gamma-k}) at $u=u_0$.
Since $\gamma_{AB}^{(k)}=0$ does not depend on $u$, it is given by equation (\ref{eq:gamma-k}) for all $u$.
\end{proof}

The following lemma relates the asymptotic behaviour of the induced metric $\gamma_{AB}$ to that of its inverse, trace, Christoffel symbols and Ricci tensor.
\begin{lemma}
\label{thm:gamma-asymptotics}
If $\gamma_{AB}=s_{AB}+O(r^k)$ for some $k>1$, then
\begin{enumerate}[i)]
\item \label{eq:gamma-inverse-asymptotics} $\gamma^{AB} = s^{AB} + O(r^k)$\quad ,
\item \label{eq:gamma-trace-asymptotics} $s^{AB}\gamma_{AB}=(d-2)+O(r^{2k})$ \quad ,
\item \label{eq:gamma-christoffel-asymptotics} $\tilde\Lambda^C{}_{AB} = \tilde\Lambda^C{}_{AB}^{(0)} + O(r^k)$ \quad ,
\item \label{eq:gamma-ricci-asymptotics} $\mathcal{R}_{AB} = \lambda(d-3)s_{AB} + O(r^k)$ \quad .
\end{enumerate}
\end{lemma}
\proof{See Note \ref{note:lemma-proof} (Appendix \ref{app:notes}).}

\begin{proof}[Proof of Theorem \ref{thm:metric-expansion-low-order}]
To prove this theorem, we substitute the metric expansion into the Einstein equations.
We consider the $AB$, $rA$ and $ru$ equations (Appendix \ref{app:grtensors-einstein}) at each order of $r$ in turn to solve for the expansion coefficients of $\alpha$, $\beta_A$ and $\gamma_{AB}$ recursively.
We start with the ``initial data'' from our choice of gauge and coordinates, and from the definition of asymptotic Einstein flatness.
We can continue the recursion until a ``breakdown'' occurs in the equations: at a certain order, the leading terms cancel out and we cannot determine the expansion coefficients at the next order.
This breakdown is important since without it the metric expansion would be entirely determined by the initial data of Definition \ref{def:strong-asymptotic-flatness}, and radiating spacetimes would be excluded from consideration.

At order $r^{-2}$, the $AB$ equation gives
\begin{equation}
2(d-1)(d-2)\alpha^{(0)}\gamma_{AB}^{(0)} = 0 \quad ,
\end{equation}
which is satisfied by our choice of gauge.
At order $r^{-1}$, it gives
\begin{equation}
 0 = (d-1)(d-2)\partial_u\gamma_{AB}^{(0)} - 2(d-1)(d-2)\alpha^{(1)}\gamma_{AB}^{(0)} \quad ,
\end{equation}
and therefore
\begin{equation}
\partial_u\gamma_{AB}^{(0)} = 2\alpha^{(1)}\gamma_{AB}^{(0)} = 0 \quad .
\end{equation}
By Corollary \ref{thm:gamma-remove-u-derivative}, $\gamma_{AB}^{(0)}=s_{AB}$ throughout $\tilde{M}$.
We find from the $rA$ equation at order $r^{-1}$ that $\beta_A^{(1)}=0$.
In summary, we have that\footnote{Note that equations (\ref{eq:metric-expansion-halfway-alpha}-b) are trivial -- they simply follow from the choice of coordinates and the definition of asymptotic Einstein flatness.}
\begin{subequations}
\begin{align}
\label{eq:metric-expansion-halfway-alpha}
& \alpha^{(0)} = \alpha^{(1)} = 0 \quad , \\
& \alpha^{(2)} = \lambda/2 \quad , \\
& \beta_A^{(0)} = \beta_A^{(1)} = 0 \quad , \\
& \gamma_{AB}^{(0)} = s_{AB} \quad .
\end{align}
\end{subequations}
This completes the proof for $d=4$.
We now continue for $d>4$.
At order $1$, the $ru$ equation gives
\begin{align}
0 &= -\mathcal{R}^{(0)} + 2(d-2)(d-3)\alpha^{(2)} - (d-3)s^{AB}\partial_u\gamma_{AB}^{(1)} \nonumber \\
\iftoggle{complete}{
&= -\lambda(d-2)(d-3) + \lambda(d-2)(d-3) - (d-3)s^{AB}\partial_u\gamma_{AB}^{(1)} \nonumber \\
}{}
  &= - (d-3)s^{AB}\partial_u\gamma_{AB}^{(1)} \quad ,
\end{align}
where we used that $\alpha^{(2)}=\lambda/2$ and that $\mathcal{R}^{(0)}=\lambda(d-2)(d-3)$.
Therefore, $s^{AB}\partial_u\gamma_{AB}^{(1)}=0$.
At order $1$, the $AB$ equation gives
\begin{align}
0 &= (d-1)(d-4)\partial_u\gamma_{AB}^{(1)} + 2s_{AB}s^{CD}\partial_u\gamma_{CD}^{(1)} \nonumber \\
& + 2(d-1)\mathcal{R}_{AB}^{(0)} - s_{AB}\mathcal{R}^{(0)} - 2d(d-3)\alpha^{(2)}s_{AB} \nonumber \\
\iftoggle{complete}{
&= (d-1)(d-4)\partial_u\gamma_{AB}^{(1)} + \lambda d(d-3)s_{AB} - \lambda d(d-3)s_{AB} \nonumber \\
}{}
  &= (d-1)(d-4)\partial_u\gamma_{AB}^{(1)} \quad ,
\end{align}
where we used that $s^{CD}\partial_u\gamma_{CD}^{(1)}=0$, that $\mathcal{R}_{AB}^{(0)}=\lambda(d-3)s_{AB}$ and that $\alpha^{(2)}=\lambda/2$.
Hence, $\partial_u\gamma_{AB}^{(1)}=0$ and by Corollary \ref{thm:gamma-remove-u-derivative}, it follows that $\gamma_{AB}^{(1)}=0$ throughout $\tilde{M}$.
This means that $\gamma_{AB}=s_{AB}+O(r^2)$ and therefore, by Lemma \ref{thm:gamma-asymptotics}(\ref{eq:gamma-trace-asymptotics}), $s^{AB}\gamma_{AB}=(d-2)+O(r^4)$.

\noindent
At order $1$, the $rA$ equation gives
\begin{equation}
\beta_A^{(2)} = \frac{1}{d-3} s^{BC} \mathscr{D}^{\strut}_{[A}\gamma_{B]C}^{(1)} = 0 \quad .
\end{equation}
At order $r$, the $ru$ equation gives
\begin{equation}
0 = -\mathcal{R}^{(1)} + 4(d-2)(d-4)\alpha^{(3)} - 2(d-3)s^{AB}\partial_u\gamma_{AB}^{(2)} \quad ,
\end{equation}
so that
\begin{equation}
4(d-2)(d-4)\alpha^{(3)} = 0 \quad ,
\end{equation}
where we used that $s^{AB}\partial_u\gamma_{AB}^{(2)} = 0$ and that $\mathcal{R}^{(1)}=0$ by Lemma \ref{thm:gamma-asymptotics}(\ref{eq:gamma-ricci-asymptotics}).

In summary,
\begin{subequations}
\begin{align}
& \alpha^{(0)} = \alpha^{(1)} = \alpha^{(3)} = 0 \quad ,\\
& \alpha^{(2)} = \lambda/2 \quad , \\
& \beta_A^{(0)} = \beta_A^{(1)} = \beta_A^{(2)} = 0 \quad , \\
& \gamma_{AB}^{(0)}=s_{AB} \ , \quad \gamma_{AB}^{(1)} = 0 \quad , \\
& s^{AB}\gamma_{AB}^{(1)} = s^{AB}\gamma_{AB}^{(2)} = s^{AB}\gamma_{AB}^{(3)} = 0 \quad .
\end{align}
\end{subequations}
%
%

\medskip\noindent
Now we proceed by induction. Suppose that for some integer $k \geq 2$, we have that
\begin{subequations}
\label{eq:einstein-equations-solutions-induction-hypothesis}
\begin{align}
\label{eq:einstein-equations-solutions-induction-hypothesis-alpha}
& \alpha^{(0)} = \alpha^{(1)} = \alpha^{(3)} = \dots = \alpha^{(k+1)} = 0 \quad , \\
& \alpha^{(2)} = \lambda/2 \quad , \\
\label{eq:einstein-equations-solutions-induction-hypothesis-beta}
& \beta_A^{(0)} = \beta_A^{(1)} = \dots = \beta_A^{(k)} = 0 \quad , \\
\label{eq:einstein-equations-solutions-induction-hypothesis-gamma}
& \gamma_{AB}^{(0)} = s_{AB} \ , \quad \gamma_{AB}^{(1)} = \dots = \gamma_{AB}^{(k-1)} =  0\quad , \\
\label{eq:einstein-equations-solutions-induction-hypothesis-gamma-trace}
& s^{AB}\gamma_{AB}^{(1)} = s^{AB}\gamma_{AB}^{(2)} = \dots = s^{AB}\gamma_{AB}^{(2k-1)} = 0 \quad .
\end{align}
Note that by Lemma \ref{thm:gamma-asymptotics}(\ref{eq:gamma-ricci-asymptotics}), we also have
\begin{equation}
\mathcal{R}_{AB}^{(1)} = \ldots = \mathcal{R}_{AB}^{(k-1)} = 0 \quad .
\end{equation}
\end{subequations}
At order $r^{k-1}$, the $AB$ equation gives
\begin{align}
0 
\iftoggle{complete}{
= (d-1)(d-2-2k))\partial_u\gamma_{AB}^{(k)} - 2[(d-1)(d-2) - k(k+1)]\alpha^{(k+1)}s_{AB} \nonumber \\
}{}
& = (d-1)(d-2-2k)\partial_u\gamma_{AB}^{(k)} \quad .
\end{align}
We note that the coefficient of $\partial_u\gamma_{AB}^{(k)}$ becomes zero when $k=(d-2)/2$.
Provided $k<(d-2)/2$, it follows that $\partial_u\gamma_{AB}^{(k)}=0$ and thus by Corollary \ref{thm:gamma-remove-u-derivative}, $\gamma_{AB}^{(k)}=0$ throughout $\tilde{M}$.
Since $\gamma_{AB}=s_{AB}+O(r^{k+1})$, it follows from Lemma \ref{thm:gamma-asymptotics} that $s^{AB}\gamma_{AB}=(d-2)+O(r^{2(k+1)})$.
The ``breakdown'' when $k=(d-2)/2$ means that $\gamma_{AB}^{(d-2)/2}$ is undetermined, and therefore the inductive proof cannot continue beyond this order.

At order $r^{k-1}$, the $rA$ equation gives
\begin{align}
0 &= (d-2-k)(k+1)\beta_A^{(k+1)} - 2ks^{BC}\mathscr{D}_{[A}^{\strut}\gamma_{B]C}^{(k)} \nonumber \\
  &= (d-2-k)(k+1)\beta_A^{(k+1)} \quad .
\end{align}
So $\beta_A^{(k+1)}=0$.
(This equation also breaks down, when $k=d-2$. However, this has no effect on the proof since we cannot proceed beyond $k=(d-2)/2$ anyway.)
At order $r^k$, the $ru$ equation gives
\begin{align}
0 &= 2(k+1)(d-2)(d-3-k)\alpha^{(k+2)} - (d-3)(k+1)\mathscr{D}^A\beta_A^{(k+1)} \nonumber \\
  &= 2(k+1)(d-2)(d-3-k)\alpha^{(k+2)} \quad .
\end{align}
Thus $\alpha^{(k+2)}=0$.
In summary, we now have
\begin{subequations}
\begin{align}
& \alpha^{(0)} = \alpha^{(1)} = \alpha^{(3)} = \ldots = \alpha^{(k+2)} = 0 \quad , \\
& \alpha^{(2)} = \lambda/2 \quad , \\
& \beta_A^{(0)} = \beta_A^{(1)} = \ldots = \beta_A^{(k+1)} = 0 \quad , \\
& \gamma_{AB}^{(0)} = s_{AB} \ , \quad \gamma_{AB}^{(1)} = \ldots = \gamma_{AB}^{(k)} =  0\quad , \\
& s^{AB}\gamma_{AB}^{(1)} = s^{AB}\gamma_{AB}^{(2)} = \ldots = s^{AB}\gamma_{AB}^{(2k+1)} = 0 \quad .
\end{align}
\end{subequations}
This completes the inductive step.
Therefore equations (\ref{eq:einstein-equations-solutions-induction-hypothesis}) hold for all $k<(d-2)/2$ (since the inductive step fails at this point, as explained above).
Substituting $k=(d-4)/2 \equiv N-1$ into these equations gives the desired result.
\end{proof}
\begin{remark}
We saw in this proof that the coefficient of the leading order term in the $AB$ equation at order $r^N$ is zero, so that $\gamma_{AB}^{(N)}$ is undetermined by the recursion.
This caused the inductive proof to break down and we could not show that any more expansion coefficients were zero.
If we specifed $\gamma_{AB}^{(N)}$ then the recursion could continue, although the subsequent expansion coefficients would not be zero.
A similar breakdown occurs in the $ru$ equation at order $r^{2N-1}$ and in the $rA$ equation at order $r^{2N}$, meaning that $\alpha^{(2N+1)}$ and $\beta_A^{(2N+1)}$ are also undetermined.
If these coefficients are specified too then the recursion can continue indefinitely.
As a result, specifying the coefficients $\gamma_{AB}^{(N)}$, $\alpha^{(2N+1)}$ and $\beta_A^{(2N+1)}$ along with the conditions in Definition \ref{def:strong-asymptotic-flatness} fully determines a metric's asymptotic expansion.
\end{remark}

\subsection{Examples}
In this section, we write down expressions for Minkowski and Schwarzschild metrics in Gaussian null coordinates as examples of the general expansion  of Theorem \ref{thm:metric-expansion-low-order}.
First, however, we write down an expansion for a generalised Schwarzschild metric where the spherical metric on the cross sections $\Sigma(r,u)$ is replaced by an arbitrary Einstein metric.
The Minkowski and Schwarzschild metrics then follow as spcecial cases.

Due to the symmetry\footnote{Minkowski and Schwarzschild are of course spherically symmetric, but the condition we require is in fact more general: we simply require that $\gamma_{AB}=fs_{AB}$ for some Einstein metric $s_{AB}$ and scalar function $f$. Clearly, this includes Minkowski and Schwarzschild as the special cases where $s_{AB}$ is just the spherical metric.}
of Minkowski and Schwarzschild, the induced metric $\gamma_{AB}$ is described by a single scalar which can be written down explicitly. 
As a result, we can write down the full expansions in these cases.

\begin{theorem}
\label{thm:minkowski-metric}
In Gaussian null coordinates, the Minkowski metric takes the form
\begin{equation}
ds^2 = r^{-2}[2drdu - r^2du^2 + d\sigma^2] \quad ,
\end{equation}
where $d\sigma^2$ is the metric of the round sphere,
and the coordinate transformation between standard radial coordinates $(t,R,x^A)$ and Gaussian null coordinates $(r,u,x^A)$ is given by
\begin{equation}
\label{eq:minkowski-coordinate-transformation}
R = \frac{1}{r} \ , \quad
t = u + \frac{1}{r} \quad .
\end{equation}
Note that the ``angular'' coordinates $x^A$ are the same in both coordinate systems and that the angular line element $d\sigma^2$ is given by
\begin{equation}
d\sigma^2 = s_{AB} dx^A dx^B \quad ,
\end{equation}
where $s_{AB}$ is the spherical metric (which is of course Einsteinian as required, with $\lambda=1$).
\end{theorem}
\begin{proof}
Minkowski spacetime is the special case of Schwarzschild (Theorem \ref{thm:schwarzschild-metric}) when the mass, $M$, is zero.
Thus for $d>4$ the expression for the metric in Gaussian null coordinates follows from substituting $M=0$ in equation (\ref{eq:schwarzschild-metric-gnc}) of Theorem \ref{thm:schwarzschild-metric}.
In the case $d=4$, the line element can be obtained directly by substituting the coordinate transformation (\ref{eq:minkowski-coordinate-transformation}) into the standard Minkowski line element in radial coordinates,
\begin{equation}
ds^2 = -dt^2 + dR^2 + R^2d\sigma^2 \quad .
\end{equation}
\iftoggle{thesis}{
Indeed, we find that
\begin{equation}
dR = -\frac{1}{r^2}dr \ , \quad
dt = du - \frac{1}{r^2}dr \quad .
\end{equation}
Thus,
\begin{align}
ds^2
\iftoggle{complete}{
& = -dt^2 + dR^2 + R^2d\sigma^2 \\
& = -\left(du-\frac{1}{r^2}dr\right)^2 + \left(-\frac{1}{r^2}\right)^2dr^2 + r^{-2}d\sigma^2 \\
& = -du^2 + \frac{2}{r^2}dudr - \frac{1}{r^2}dr^2 + \frac{1}{r^4}dr^4 + r^{-2}d\sigma^2 \\
& = -du^2 + \frac{2}{r^2}dudr + r^{-2}d\sigma^2 \\
}{}
& = r^{-2}[2dudr - r^2du^2 + d\sigma^2] \quad ,
\end{align}
}{}
which is the required line element in Gaussian null coordinates.
\end{proof}

\begin{theorem}
\label{thm:schwarzschild-metric}
In Gaussian null coordinates, the Schwarzschild metric takes the form
\begin{equation}
\label{eq:schwarzschild-metric-gnc}
ds^2 = r^{-2}[2drdu - r^2du^2 + Mr^{d-1}du^2 + d\sigma^2] \quad ,
\end{equation}
where $d\sigma^2$ is the spherical metric in $(d-2)$ dimensions, and $M$ is a constant.
Moreover, in $4$ dimensions, the coordinate transformation between standard radial coordinates $(t,R,x^A)$ and $(r,u,x^A)$ is given by
\begin{equation}
\label{eq:schwarzschild-coordinate-transformation}
R=\frac{1}{r} \ , \quad t = u + \frac{1}{r} + M\log(1-Mr) - M\log r \quad .
\end{equation}
In higher dimensions, it becomes increasingly difficult to write down an explicit coordinate transformation.
\end{theorem}
\begin{proof}
For $d>4$, the theorem is merely a special case of Theorem \ref{thm:generalised-schwarzschild-metric} with $\lambda=1$ and $s_{AB}$ the $(d-2)$ dimensional spherical metric.
For $d=4$, the line element in Gaussian null coordinates can be obtained directly by substituting the coordinate transformation (\ref{eq:schwarzschild-coordinate-transformation}) into the 4 dimensional Schwarzschild line element in standard radial coordinates,
\begin{equation}
\label{eq:4d-schwarzschild-radial}
ds^2 = -\left(1-\frac{M}{R}\right)dt^2 + \left(1-\frac{M}{r}\right)^{-1}dR^2 + R^2d\sigma^2 \quad .
\end{equation}
\iftoggle{thesis}{
Indeed we find that
\begin{align}
dR &=-\frac{1}{r^2}dr \quad , \\
dt
\iftoggle{complete}{
& = du + \left( -\frac{1}{r^2} - \frac{M^2}{1-Mr} - \frac{M}{r} \right)dr \\
& = du + \left(\frac{-(1-Mr) - M^2r^2 - Mr(1-Mr)}{r^2(1-Mr)}\right)dr \\
& = du - \left(\frac{(1-Mr) + M^2r^2 + Mr(1-Mr)}{r^2(1-Mr)}\right)dr \\
& = du - \left(\frac{1-Mr + M^2r^2 + Mr-M^2r^2}{r^2(1-Mr)}\right)dr \\
}{}
& = du - \frac{1}{r^2(1-Mr)}dr \quad ,
\end{align}
and substituting into equation (\ref{eq:4d-schwarzschild-radial}), we indeed find that
\begin{align}
ds^2
\iftoggle{complete}{
& = -\left(1-\frac{M}{R}\right)dt^2 + \left(1-\frac{M}{R}\right)^{-1}dR^2 + R^2d\sigma^2 \\
& = -\left(1-Mr\right)\left(du - \frac{1}{r^2(1-Mr)}dr\right)^2 + \left(1-Mr\right)^{-1}\left(-\frac{1}{r^2}\right)^2 + \frac{1}{r^2}^2d\sigma^2 \\
& = -\left(1-Mr\right)\left(du^2 - \frac{2}{r^2(1-Mr)}drdu + \frac{1}{r^4(1-Mr)^2}dr^2\right) + \left(1-Mr\right)^{-1}\frac{1}{r^4}dr^2 + \frac{1}{r^2}d\sigma^2 \\
& = -(1-Mr)du^2 + \frac{2}{r^2}drdu - \frac{1}{r^4(1-Mr)}dr^2 + \frac{1}{r^4(1-Mr)}dr^2 + \frac{1}{r^2}d\sigma^2 \\
& = -(1-Mr)du^2 + \frac{2}{r^2}drdu + \frac{1}{r^2}d\sigma^2 \\
}{}
& = \frac{1}{r^2}\left[ 2drdu - du^2 + Mr^3du^2 + d\sigma^2 \right] \quad ,
\end{align}
as required.
}{}
\end{proof}


\begin{theorem}
\label{thm:generalised-schwarzschild-metric}
Consider a ``generalised Schwarzschild metric'' in dimension $d>4$ with line element
\begin{equation}
\label{eq:generalised-schwarzschild-radial}
ds^2 = -\left(1-\frac{C}{R^{d-3}}\right)dt^2 + \left(1-\frac{C}{R^{d-3}}\right)^{-1}dR^2 + R^2d\tau^2 \quad ,
\end{equation}
in standard radial coordinates $(t,R,x^A)$, where $C$ is a constant and $d\tau^2$ is the line element of a $(d-2)$ dimensional Einstein metric $s_{AB}$,
\begin{equation}
d\tau^2 = s_{AB} dx^A dx^B \quad .
\end{equation}
Then the line element in Gaussian null coordinates takes the form
\begin{equation}
\label{eq:generalised-schwarzschild-gnc}
ds^2 = r^{-2}[2drdu - \lambda r^2 du^2 + Mr^{d-1} + d\sigma^2] \quad ,
\end{equation}
where $M$ is a constant.
\end{theorem}
\begin{remark}
It is not explicitly shown that it is the same constant $M$ in equations (\ref{eq:generalised-schwarzschild-gnc}) and (\ref{eq:generalised-schwarzschild-radial}).
However, it follows from the Bondi mass formula in Section \ref{sec:bondi-mass} that the constant $M$ in equation (\ref{eq:generalised-schwarzschild-gnc}) corresponds to the mass of the Schwarzschild black hole.
\end{remark}

\begin{proof}
Because $\gamma_{AB}$ is described by a single scalar function, we can in this case use the $rr$ component of Einstein's equations to solve for it explicitly. This leads to significant simplification of the Einstein equations, allowing us to solve for the entire metric expansion.

Substituting $\gamma_{AB}=fs_{AB}$ into the $rr$ component of Einstein's equations, we obtain a differential equation
\begin{equation}
2f\frac{\partial^2f}{\partial r^2} = \left(\frac{\partial f}{\partial r}\right)^2 \quad ,
\end{equation}
which can be solved explicitly to give
\begin{equation}
\label{eq:generalised-schwarzschild-metric-proof-gamma}
f = (f_0 + rf_1)^2 \quad ,
\end{equation}
where $f_0$ and $f_1$ are functions of $u,x^A$.
By the general analysis of the previous section\footnote{This is where we require $d>4$, since for $d=4$, we do not have $\gamma_{AB}^{(1)}=0$.}
,
\begin{equation}
\gamma_{AB}^{(0)} = s_{AB} \ , \quad
\gamma_{AB}^{(1)} = 0 \quad ,
\end{equation}
which here means that $f_0=1$, $f_1=0$.
Substituting this into equation (\ref{eq:generalised-schwarzschild-metric-proof-gamma}), we find that $f=1$ and therefore $\gamma_{AB}=s_{AB}$.
As a result, $\partial_r\gamma_{AB}=\partial_u\gamma_{AB}=0$, and the Einstein equations simplify significantly.

The $rA$ component of Einstein's equations gives, at order $r^k$ for $k\geq0$,
\begin{equation}
(k+2)(d-3-k)\beta_A^{(k+2)} = 0 \quad ,
\end{equation}
and so $\beta_A^{(k)}=0$ for $k\geq2$, $k\neq d-1$, that is,
\begin{equation}
\beta_A = r^{d-1}\beta_A^{(d-1)} \quad .
\end{equation}
Next, we apply the $ru$ equation.
It is trivial at order $1$, and also at order $r^{d-3}$.
At order $r^k$ for $k>0$, $k\neq d-3,\ d-2,\ 2d-4$, we find that
\iftoggle{complete}{
\begin{equation}
0 = 2(d-2)(d-3-k)(k+1)\alpha^{(k+2)} \quad ,
\end{equation}
and so
}{}
$\alpha^{(k+2)}=0$.
At order $r^{d-2}$, we get
\iftoggle{complete}{
\begin{equation}
0 = 2(d-2)\alpha^{(d)} + (d-3)\mathscr{D}^A\beta_A^{(d-1)} \quad ,
\end{equation}
which gives
}{}
\begin{equation}
\alpha^{(d)} = -\frac{d-3}{2(d-2)}\mathscr{D}^A\beta_A^{(d-1)} \quad .
\end{equation}
At order $r^{2d-4}$, we find
\iftoggle{complete}{
\begin{equation}
2(d-2)\alpha^{(2d-2)} + \frac{1}{2}(d-3)\beta_A^{(d-1)}\beta^{A(d-1)} = 0 \quad ,
\end{equation}
and thus
}{}
\begin{equation}
\alpha^{(2d-2)} = -\frac{1}{4}\frac{d-3}{d-2}\beta_A^{(d-1)}\beta^{A(d-1)} \quad .
\end{equation}
Therefore, $\alpha$ takes the form
\begin{equation}
\alpha = \frac{\lambda}{2}r^2 + \alpha^{(d-1)}r^{d-1} - \frac{d-3}{2(d-2)}r^d\mathscr{D}^A\beta_A^{(d-1)} - \frac{d-3}{4(d-2)}\beta_A^{(d-1)}\beta^{A(d-1)}r^{2d-2} \quad .
\end{equation}
From the $uu$ equation at order $r^{4d-6}$ we find that
\begin{equation}
\beta_C^{(d-1)}\beta^{C(d-1)} = 0 \quad ,
\end{equation}
and then from the $AB$ equation at order $r^{2d-4}$, we find
\begin{equation}
\beta_A^{(d-1)}\beta_B^{(d-1)} = \frac{1}{d-2}s_{AB}\beta_C^{(d-1)}\beta^{C(d-1)} = 0 \quad ,
\end{equation}
so that $\beta_A^{(d-1)}=0$.
\iftoggle{complete}{
(to see this contract with any vector $t^A$,
\begin{equation}
0 = t^A\beta_A^{(d-1)} \beta_B^{(d-1)} \quad .
\end{equation}
If $t^A\beta_A^{(d-1)}\neq0$ then it follows that $\beta_B^{(d-1)}=0$.
If $t^A\beta_A^{(d-1)}=0$ for all $t^A$ then it also follows that $\beta_B^{(d-1)}=0$.)
}{}
Consequently, $\beta_A=0$ and $\alpha = (\lambda/2)r^2 + \alpha^{(d-1)}r^{d-1}$.
From the $uu$ and $uA$ equations at order $r^{d-2}$ we find that
\iftoggle{complete}{
\begin{equation}
0 = 2(d-1)(d-2)\partial_u\alpha^{(d-1)} = 2(d-1)(d-2)\partial_A\alpha^{(d-1)} \quad ,
\end{equation}
and so
}{}
$\partial_u\alpha^{(d-1)}=\partial_A\alpha^{(d-1)}=0$, that is, $\alpha^{(d-1)}$ is constant.
Defining $M\defeq -\alpha^{(d-1)}$, the metric expansion takes on the required form (\ref{eq:generalised-schwarzschild-gnc}).
\end{proof}

\subsection{Einstein Equations at Higher Orders}
\label{sec:einstein-equations-higher-order}
Now that we have written down the asymptotic behaviour of the metric (Theorem \ref{thm:metric-expansion-low-order}), we can write down greatly simplified expressions for the Einstein equations up to certain orders.
In turn, these equations lead to relations between the higher order metric expansion coefficients.
The main results of our analysis are given in the following theorem.

\begin{theorem}
\label{thm:metric-components-higher-order}
The expansion coefficients of $\alpha$ and $\mathcal{R}$ can be expressed directly in terms of $\gamma_{AB}$,
\begin{align}
\label{eq:alpha-higher-order-gamma}
 \alpha^{(k+2)} = - \frac{k-1}{2(k+1)(d-2-k)(d-3-k)} \mathscr{D}^A\mathscr{D}^B\gamma_{AB}^{(k)} \quad ,
\end{align}
\begin{align}
\label{eq:ricciscalar-higher-order-gamma}
 \mathcal{R}^{(k)} = \mathscr{D}^A\mathscr{D}^B\gamma_{AB}^{(k)} \quad ,
\end{align}
for $0<k<d-3$.
\end{theorem}
\begin{proof}
To prove this theorem, we need to write down the simplified Einstein equations (Lemma \ref{thm:einstein-equations-high-order}) and then solve them to obtain the expressions for the expansion coefficients (Lemma \ref{thm:metric-components-higher-order-halfway}).
The theorem is then proved by combining these expressions.

Indeed, substituting equation (\ref{eq:beta-higher-order-gamma}) into (\ref{eq:alpha-higher-order-beta}) gives (\ref{eq:alpha-higher-order-gamma}).
Substituting equations (\ref{eq:alpha-higher-order-gamma}) and (\ref{eq:beta-higher-order-gamma}) into (\ref{eq:ricciscalar-higher-order-alpha-beta}) gives (\ref{eq:ricciscalar-higher-order-gamma}).
\end{proof}

\begin{lemma}
\label{thm:einstein-equations-high-order}
The $rr$, $ru$, $rA$ and trace $AB$ components of the Einstein equations are given by
\begin{subequations}
\label{eq:einstein-equations-high-order}
\begin{align}
0 = &-\gamma^{AB}\partial_r^2\gamma_{AB} + \frac{1}{2}\gamma^{AB}\gamma^{CD}(\partial_r\gamma_{AC})\partial_r\gamma_{BD} \quad , \\
0 = &-\mathcal{R} - 2(d-2)\partial_r^2\alpha + 2(d-1)(d-2)r^{-1}(\partial_r\alpha - r^{-1}\alpha) - (d-3)\mathscr{D}^A\partial_r\beta_A\nonumber\\
    &-(d-3)\gamma^{AB}\partial_r\partial_u\gamma_{AB} + \left(\frac{d}{2}-2\right)\gamma^{AB}\gamma^{CD}(\partial_r\gamma_{AC})\partial_u\gamma_{BD} + O(r^{d-2}) \quad , \\
0 = &-\partial_r^2\beta_A + (d-2)r^{-1}\partial_r\beta_A - 2s^{BC}\mathscr{D}_{[A}\partial_r\gamma_{B]C} + O(r^{d-3}) \nonumber \\
  = &-\partial_r^2\beta_A + (d-2)r^{-1}\partial_r\beta_A + \mathscr{D}^{B}\partial_r\gamma_{AB} + O(r^{d-3}) \quad ,
\end{align}
\begin{align}
0 = &+d\mathcal{R} +2(d-2)\partial_r^2\alpha - 2(d-1)(d-2)^2 r^{-2}\alpha \nonumber \\
    &- 2\mathscr{D}^A\partial_r\beta_A + 2(d-1)(d-2)r^{-1}\mathscr{D}^A\beta_A \nonumber \\
    &- 2\gamma^{AB}\partial_r\partial_u\gamma_{AB} + (d-1)(d-2)r^{-1}\gamma^{AB}\partial_u\gamma_{AB} \nonumber \\
    &+ \left(\frac{d}{2}+1\right)\gamma^{AB}\gamma^{CD}(\partial_r\gamma_{AC})\partial_u\gamma_{CD} + O(r^{d-2}) \quad .
\end{align}
\end{subequations}
\end{lemma}
\begin{proof}
This follows from substitution of the metric component expansions (\ref{eq:metric-components-low-order}) into the Einstein equations and from simplifying.

In particular, for the $rA$ equation we used
\begin{align}
2s^{BC}\mathscr{D}_{[A}\partial_r\gamma_{B]C}
& = s^{BC}\mathscr{D}_A\partial_r\gamma_{BC} - s^{BC}\mathscr{D}_B\partial_r\gamma_{AC} \nonumber \\
& = -\mathscr{D}^C\partial_r\gamma_{AC} + O(r^{d-3}) \quad ,
\end{align}
where we used that
\begin{equation}
s^{BC}\mathscr{D}_A\partial_r\gamma_{BC} = \mathscr{D}_A(s^{BC}\partial_r\gamma_{BC}) = O(r^{d-3}) \quad .
\end{equation}
\end{proof}

\begin{lemma}
\label{thm:metric-components-higher-order-halfway}
Solving the Einstein equations (\ref{eq:einstein-equations-high-order}) gives the following expressions for the expansion coefficients of $\alpha$, $\beta_A$ and $\mathcal{R}$,
\begin{align}
\label{eq:alpha-higher-order-beta}
& \alpha^{(k+2)}=\frac{k-1}{2k(d-3-k)}\mathscr{D}^A\beta_A^{(k+1)} \quad , \\
\label{eq:beta-higher-order-gamma}
& \beta_A^{(k+2)} = -\frac{k+1}{(k+2)(d-3-k)}\mathscr{D}^B\gamma_{AB}^{(k+1)} \quad ,
\end{align}
\begin{equation}
\label{eq:ricciscalar-higher-order-alpha-beta}
\mathcal{R}^{(k)} = 2(d-2)(d-3-k)(k+1)\alpha^{(k+2)} - (d-3)(k+1)\mathscr{D}^A\beta_A^{(k+1)} \quad ,
\end{equation}
for all $0<k<d-3$.
\end{lemma}
\begin{proof}
To obtain these expressions, we apply the Einstein equations (\ref{eq:einstein-equations-high-order}).
At order $r^k$ for some integer $k<d-3$, the $ru$ and trace $AB$ equations give
\begin{align}
0 &= -\mathcal{R}^{(k)} + 2(d-2)(d-3-k)(k+1)\alpha^{(k+2)} - (d-3)(k+1)\mathscr{D}^A\beta_A^{(k+1)} \quad , \\
0 &= d\mathscr{R}^{(k)} - 2(d-2)(d+k)(d-3-k)\alpha^{(k+2)} + 2(d^2-3d-k+1)\mathscr{D}^A\beta_A^{(k+1)} \quad .
\end{align}
Rearranging the first equation gives the desired expression for $\mathcal{R}^{(k)}$.
Substituting this expression into the second equation and rearranging gives the expression for $\alpha^{(k+2)}$.
At order $r^k$, the $rA$ equation is
\begin{equation}
0 = (d-3-k)(k+2)\beta_A^{(k+2)} + (k+1)\mathscr{D}^C\gamma_{AC}^{(k+1)} \quad ,
\end{equation}
which can be rearranged to give the expression for $\beta_A^{(k)}$.
\end{proof}

\iftoggle{complete}{
\begin{theorem}
When $d>4$, $\beta_A^{(2N)}$ is divergence free, $\mathscr{D}^A\beta_A^{(2N)}=0$.
\end{theorem}
\begin{proof}
At order $r^{d-4}$, the $rr$ equation becomes
\begin{equation*}
0 = (d-3)s^{AB}\gamma_{AB}^{(d-2)} + \frac{3d-10}{8}\gamma^{AB\left(\frac{d-2}{2}\right)}\gamma_{AB}^{\left(\frac{d-2}{2}\right)} \quad ,
\end{equation*}
and taking the $u$ derivative of both sides we find that
\begin{equation*}
0 = (d-3)s^{AB}\partial_u\gamma_{AB}^{(d-2)} + \frac{3d-10}{4}\gamma^{AB\left(\frac{d-2}{2}\right)}\partial_u\gamma_{AB}^{\left(\frac{d-2}{2}\right)} \quad .
\end{equation*}
At order $r^{d-3}$, the $ru$ and trace $AB$ components become
\begin{align*}
0 = &-(d-2)(d-3)\mathscr{D}^A\beta_A^{(d-2)} - \mathcal{R}^{(d-3)}\\
    &-(d-2)(d-3)s^{AB}\partial_u\gamma_{AB}^{(d-2)} - \frac{1}{4}(d-2)(3d-10)\gamma^{AB\left(\frac{d-2}{2}\right)}\partial_u\gamma_{AB}^{\left(\frac{d-2}{2}\right)}\\
  = &-(d-2)(d-3)\mathscr{D}^A\beta_A^{(d-2)} - \mathcal{R}^{(d-3)}\\
0 = &+2(d-2)^2\mathscr{D}^A\beta_A^{(d-2)} + d\mathcal{R}^{(d-3)}\\
    &+(d-2)(d-3)s^{AB}\partial_u\gamma_{AB}^{(d-20} + \frac{1}{4}(d-2)(3d-10)\gamma^{AB\left(\frac{d-2}{2}\right)}\partial_u\gamma_{AB}^{\left(\frac{d-2}{2}\right)}\\
  = &+2(d-2)^2\mathscr{D}^A\beta_A^{(d-2)} + d\mathcal{R}^{(d-3)}\\
\end{align*}
Combining these equations to eliminate $\mathcal{R}^{(d-3)}$, we find that
\begin{equation*}
0 = (d-1)(d-2)(d-4)\mathscr{D}^A\beta_A^{(d-2)} \quad ,
\end{equation*}
\end{proof}
\begin{remark}
As was seen in this proof, the highest order coefficients of $\alpha$ and $\gamma^{AB}\partial_u\gamma_{AB}$ cancelled out when combining the $ru$ and trace $AB$ components.
As a result, we do not learn anything about the expansion coefficients of these terms at this order ($\alpha^{(2N+1)}$, $s^{AB}\gamma_{AB}^{(2N)}$ and $\gamma^{AB(N)}\partial_u\gamma_{AB}^{(N)}$).
\end{remark}
}{}

%% file: AsymptoticSymmetries.tex
\section{Asymptotic Symmetries}
\label{sec:asymptotic-symmetries}
A definition is given of asymptotic symmetries and an asymptotic expansion is made of the infinitesimal generator of such a symmetry.

\subsection{Asymptotic Killing Equation}
\label{sec:asymptotic-symmetries-killing-equation}

\begin{definition}
\label{def:asymptotic-symmetry}
An asymptotic symmetry $\phi$ is a diffeomorphism which preserves asymptotic Einstein flatness \cite{Hollands:2003ie}.
That is, whenever a metric $g_{ab}$ is asymptotically Einstein flat, $\phi^\ast g_{ab}$ must also be asymptotically Einstein flat.
\end{definition}

\medskip\noindent
The generator $\xi^a$ of an isometry group is a solution of Killing's equation
\begin{equation}
\Lie_\xi g_{ab} \equiv 2\nabla_{(a}\xi_{b)} = 0 \quad .
\end{equation}
To obtain the generator of asymptotic symmetries, we must first write down the equivalent of this equation: an ``asymptotic Killing equation''.

\begin{theorem}
The vector $\xi^a$ is an infinitesimal generator of an asymptotic symmetry of $g_{ab}$ if the tensor
\begin{equation}
\label{eq:unphysical-killing-tensor}
\tilde\chi_{ab} = r^2\Lie_\xi g_{ab}
 = 2\tilde\nabla_{(a}\xi_{b)} - 2r^{-1}\tilde{g}_{ab}(\nabla^c r)\tilde\xi_c
\end{equation}
satisfies the following asymptotic conditions,
\begin{subequations}
\label{eq:asymptotic-killing-equation-components}
\begin{align}
& \tilde\chi_{rr} = \tilde\chi_{ru} = \tilde\chi_{rA} = 0 \quad , \\
& \tilde\chi_{uu} = O(r^\frac{d+2}{2}) \quad , \\
& \tilde\chi_{uA} = O(r^\frac{d}{2}) \quad , \\
& \tilde\chi_{AB} = O(r^\frac{d-2}{2}) \quad , \\
& \gamma^{AB}\tilde\chi_{AB} = O(r^{d-2}) \quad .
\end{align}
\end{subequations}
\end{theorem}
\begin{proof}
Under an infinitesimal transformation $\phi$ with generator $\xi^a$, an asymptotically Einstein flat metric $g_{ab}$ transforms as
\begin{equation}
g_{ab} \rightarrow g_{ab} + \Lie_\xi g_{ab} \quad .
\end{equation}
For $\phi$ to be an asymptotic symmetry, the right hand side of this equation must also be asymptotically Einstein flat.
Thus, by Corollary \ref{thm:background-metric}, we can write
\begin{align}
g_{ab} &= \bar{g}_{ab} + O(r^\frac{d-2}{2}) du_a du_b + O(r^\frac{d-4}{2}) du_{(a} dx^A_{b)} + O(r^\frac{d-6}{2}) dx^A_{(a} dx^B_{b)} \quad , \nonumber \\
g_{ab} + \Lie_\xi g_{ab} &= \bar{g}_{ab} + O(r^\frac{d-2}{2}) du_a du_b + O(r^\frac{d-4}{2}) du_{(a} dx^A_{b)} + O(r^\frac{d-6}{2}) dx^A_{(a} dx^B_{b)} \quad .
\end{align}
Subtracting the first equation from the second yields a ``physical asymptotic Killing equation''
\begin{equation}
\label{eq:asymptotic-symmetry-combined-condition}
\Lie_\xi g_{ab} = O(r^\frac{d-2}{2}) du_a du_b + O(r^\frac{d-4}{2}) du_{(a} dx^A_{b)} + O(r^\frac{d-6}{2}) dx^A_{(a} dx^B_{b)} \quad .
\end{equation}
Defining the deviation tensor
\begin{equation}
\chi_{ab} = \Lie_\xi g_{ab} \quad ,
\end{equation}
and splitting equation (\ref{eq:asymptotic-symmetry-combined-condition}) into components, we get
\begin{subequations}
\label{eq:physical-asymptotic-symmetry-conditions}
\begin{align}
& \chi_{rr} = \chi_{ru} = \chi_{rA} = 0 \quad , \\
& \chi_{uu} = O(r^\frac{d-2}{2}) \quad , \\
& \chi_{uA} = O(r^\frac{d-4}{2}) \quad , \\
& \chi_{AB} = O(r^\frac{d-6}{2}) \quad , \\
& \chi^A{}_A = O(r^{d-4}) \quad .
\end{align}
\end{subequations}
Rewriting $\chi_{ab}$ in terms of the unphysical derivative, we find, 
\begin{align}
\chi_{ab}
& = \Lie_\xi g_{ab} \nonumber \\
& = 2\nabla_{(a}\xi_{b)} \nonumber \\
& = 2\tilde\nabla_{(a}\xi_{b)} - 2\tilde{C}^c{}_{ab}\xi_c \quad ,
\end{align}
where the ``conformal connection coefficients'' are
\begin{align}
\tilde{C}^c{}_{ab} 
\iftoggle{complete}{
& = \frac{1}{2} g^{cd} (\tilde\nabla_a g_{bd} + \tilde\nabla_b g_{ad} - \tilde\nabla_d g_{ab})\\
& = \frac{1}{2} r^2 \tilde{g}^{cd} (\tilde\nabla_a r^{-2} \tilde{g}_{bd} + \tilde\nabla_b r^{-2} \tilde{g}_{ad} - \tilde\nabla_d r^{-2} \tilde{g}_{ab})\\
& = \frac{1}{2} r^2 \tilde{g}^{cd} (-2r^{-3}) (\tilde\nabla_a r \tilde{g}_{bd} + \tilde\nabla_b r \tilde{g}_{ad} - \tilde\nabla_d r \tilde{g}_{ab})\\
& = r^{-1} (\tilde\nabla_a r \delta^c{}_b + \tilde\nabla_b r \delta^c{}_a - (\tilde\nabla^c r)\tilde{g}_{ab})\\
}{}
& = r^{-1} (2\tilde\nabla_{(a} r \delta^c{}_{b)} - (\tilde\nabla^c r)\tilde{g}_{ab}) \quad ,
\end{align}
so that
\begin{align}
\Lie_\xi g_{ab} 
\iftoggle{complete}{
& = 2\tilde\nabla_{(a}\xi_{b)} + 4r^{-1} (\tilde\nabla_{(a} r \delta^c{}_{b)}) \xi_c - 2r^{-1} \tilde{g}_{ab} (\tilde\nabla^c r) \xi_c \nonumber \\
}{}
& = 2\tilde\nabla_{(a}\xi_{b)} + 4 r^{-1} \xi_{(a} \tilde\nabla_{b)}r - 2r^{-1} \tilde{g}_{ab} (\tilde\nabla^c r) \xi_c \quad .
\end{align}
Now define an unphysical vector $\tilde\xi^a \defeq \xi^a$ so that
\begin{equation}
\tilde\xi_a = \tilde{g}_{ab}\tilde\xi^b = r^2 g_{ab}\xi^b = r^2\xi_a \quad ,
\end{equation}
and
\begin{equation}
\tilde\nabla_{(a}\xi_{b)} = \tilde\nabla_{(a}(r^{-2}\tilde\xi_{b)}) = r^{-2}[\tilde\nabla_{(a}\tilde\xi_{b)} - 2r^{-1}\tilde\xi_{(a}\tilde\nabla_{b)}r] \quad .
\end{equation}
We can therefore write the physical deviation as
\begin{align}
\iftoggle{complete}{
\Lie_\xi g_{ab} \equiv
}{}
\chi_{ab}
\iftoggle{complete}{
& = 2r^{-2} [\tilde\nabla_{(a}\tilde\xi_{b)} - 2r^{-1}\tilde\xi_{(a}\tilde\nabla_{b)}r + 2r^{-1}\xi_{(a}\tilde\nabla_{b)}r - r^{-1}\tilde{g}_{ab}(\tilde\nabla^c r)\tilde\xi_c] \nonumber \\
}{}
& = 2r^{-2} [\tilde\nabla_{(a}\tilde\xi_{b)} - r^{-1}\tilde{g}_{ab}(\tilde\nabla^c r)\tilde\xi_c] \quad .
\end{align}
If we define $\tilde\chi_{ab} \defeq r^2\chi_{ab} = r^2\Lie_\xi g_{ab}$, then we get
\begin{equation}
\tilde\chi_{ab} = 2\tilde\nabla_{(a}\tilde\xi_{b)} - 2r^{-1}\tilde{g}_{ab}(\tilde\nabla^c r)\tilde\xi_c \quad ,
\end{equation}
and writing the conditions (\ref{eq:physical-asymptotic-symmetry-conditions}) for the asymptotic symmetry in terms of $\tilde\chi_{ab}$ gives the desired conditions (\ref{eq:asymptotic-killing-equation-components}).
\end{proof}

\subsection{Solving the Asymptotic Killing equation}
\label{sec:asymptotic-symmetries-solution}
In this section, we obtain the expansion of a vector $\xi^a$ which generates the asymptotic symmetries, using conditions (\ref{eq:asymptotic-killing-equation-components}).
The tensor $\tilde\chi_{ab}$ is written in terms of the covector $\tilde\xi_a$, but we are seeking the expansion of the vector $\xi^a \equiv \tilde\xi^a$.
We could proceed either by solving for $\tilde\xi_a$ and raising the index afterwards with $\tilde{g}^{ab}$,
or we could write $\tilde\xi_a = \tilde{g}_{ab}\tilde\xi^b$ and solve directly for $\tilde\xi^a$.
To prove Theorem $\ref{thm:asymptotic-symmetries}$, we use the second method, since it turns out to be simpler.

\begin{theorem}
\label{thm:asymptotic-symmetries}
In dimension $d>4$, the vector $\xi^a$ given by
\begin{subequations}
\label{eq:asymptotic-symmetries-components}
\begin{align}
& \xi^r = r\partial_uf + r^2(g_0+\lambda f) - \sum_{k=N+2}^\infty \frac{1}{k-2} r^k \beta^{A(k-1)} \ \mathscr{D}_Af \quad , \\
& \xi^u = f \quad , \\
& \xi^A = h^A - r\mathscr{D}^Af - \sum_{k=N+1}^\infty \frac{1}{k} r^k \gamma^{AB(k-1)} \ \mathscr{D}_Bf
\end{align}
\end{subequations}
is a solution of the asymptotic Killing equation (\ref{eq:asymptotic-killing-equation-components}),
where $g_0$ is a constant, $h^A$ are functions of $x^A$ and $f=f_0+uf_1$ for functions $f_0$, $f_1$ of $x^A$, such that
\begin{subequations}
\label{eq:asymptotic-symmetries-conditions-low-order}
\begin{align}
\label{eq:asymptotic-symmetries-conditions-low-order-hA}
& \mathscr{D}_{(A}h_{B)} - f_1 s_{AB} = 0 \quad , \\
\label{eq:asymptotic-symmetries-conditions-low-order-f0}
& \mathscr{D}_{(A}\mathscr{D}_{B)}f_0 + (\lambda f_0 + g_0)s_{AB} = 0 \quad , \\
\label{eq:asymptotic-symmetries-conditions-low-order-f1}
& \mathscr{D}_{(A}\mathscr{D}_{B)}f_1 + \lambda f_1 s_{AB} = 0 \quad .
\end{align}
\end{subequations}
\end{theorem}
\begin{remark}
Since $h^A$ is a lowest order term, we raise and lower its index with the lowest order term of $\gamma_{AB}$.
That is, its index is raised and lowered with $s_{AB}$ rather than with $\gamma_{AB}$ itself.
Note that $\mathscr{D}_A$ is also a ``lowest order term'' and so we raise and lower its index with $s_{AB}$.
The equations (\ref{eq:asymptotic-symmetries-conditions-low-order}) are not analysed in detail in this paper.
It is expected that they have qualitatively different behaviour for $\lambda=0,+1,-1$.
\end{remark}

The following Lemma will be useful in proving this theorem.
\begin{lemma}
\label{thm:gamma-christoffel-higher-order}
For $0<k<2N$,
\begin{equation}
\tilde\Lambda^A{}_{AB}^{(k)} = 0 \quad .
\end{equation}
\end{lemma}
\begin{proof}
We calculate that, for $0<k<2N$,
\begin{align}
\tilde\Lambda^A{}_{AB}^{(k)}
& = \frac{1}{2}\gamma^{AC(k)}(\partial_A s_{BC} + \partial_B s_{AC} - \partial_C s_{AB}) + \frac{1}{2}s^{AC}(\partial_A^{\strut}\gamma_{BC}^{(k)} + \partial_B^{\strut}\gamma_{AC}^{(k)} - \partial_C^{\strut}\gamma_{AB}^{(k)}) \nonumber \\
& = \frac{1}{2}\gamma^{AC(k)}\partial_B s_{AC} + \frac{1}{2}s^{AC}\partial_B^{\strut}\gamma_{AC}^{(k)} \quad .
\end{align}
In addition, we find that
\begin{equation}
\frac{1}{2}\gamma^{AC(k)}\partial_B s_{AC} = -s^{AC}\tilde\Lambda^D{}_{AB}^{(0)}\gamma_{CD}^{(k)} \quad .
\end{equation}
Therefore we have that
\begin{equation}
\tilde\Lambda^A{}_{AB}^{(k)} = \frac{1}{2}s^{AC}\mathscr{D}^{\strut}_B\gamma_{AC}^{(k)} = 0 \quad ,
\end{equation}
since $s^{AC}$ can be taken inside the derivative.
\end{proof}

\begin{proof}[Proof of Theorem \ref{thm:asymptotic-symmetries}]
First note that the unphysical vector $\tilde\xi^a$ is by definition equal to the physical generator $\xi^a$.
Therefore we need only solve (\ref{eq:asymptotic-killing-equation-components}) for the unphysical vector $\tilde\xi^a \equiv \xi^a$.

We use the ``exact'' $rr$, $ru$ and $rA$ equations (i.e. the components of $\tilde\chi_{ab}$ which are exactly $0$) to obtain the basic asymptotic expansions of the vector $\tilde\xi^a$.
We then derive the conditions satisfied by the lowest order coefficients of the expansions using the ``inexact'' $uu$, $uA$ and $AB$ components.
This is enough to fully specify the vector $\tilde\xi^a$.
The trace $AB$ equation then serves as a consistency check, not only on our expressions for the asymptotic symmetries,
but also on our expression (\ref{eq:beta-higher-order-gamma}) for the higher order metric expansion coefficients $\beta_A^{(k)}$.

Rewriting equation (\ref{eq:unphysical-killing-tensor}) in terms of $\tilde\xi^a$ (rather than $\tilde\xi_a$ with its index down),
\begin{equation}
\tilde\chi_{ab} = 2\tilde{g}_{c(a}\tilde\nabla_{b)}\tilde\xi^c - 2r^{-1}\tilde{g}_{ab}\tilde\xi^r \quad .
\end{equation}
We can now solve for $\tilde\xi^a$ using conditions (\ref{eq:asymptotic-killing-equation-components}).
The $rr$ component gives
\begin{align}
0 = \tilde\chi_{rr} 
 = 2\tilde\partial_r\tilde\xi^u \quad ,
\end{align}
and integration with respect to $r$ yields
\begin{equation}
\tilde\xi^u = f \quad ,
\end{equation}
where $f$ is a scalar function of $u,x^A$.
Next, the $rA$ equation gives
\begin{equation}
0 = \tilde\chi_{rA} = \mathscr{D}_A f + \gamma_{AB}\partial_r\tilde\xi^B \quad ,
\end{equation}
and contracting into $\gamma^{AC}$ gives,
\begin{equation}
0 = \partial_r\tilde\xi^A + \gamma^{AB}\mathscr{D}_Bf \quad .
\end{equation}
Since $\tilde\xi^{A(0)}$ is not determined by this equation, we denote it by $h^A\defeq\tilde\xi^{A(0)}$.
At order $1$, we find that $\tilde\xi^{A(1)} = -\mathscr{D}^Af$ (where the index of $\mathscr{D}^A$ has been raised with $s^{AB}$ as previously noted).
At order $r^k$ for $0<k<N$, we find that $\tilde\xi^{A(k+1)}=0$.
At order $r^k$ for $k\geq N$, we find that
\begin{align}
(k+1)\tilde\xi^{A(k+1)}
& = -\gamma^{AB(k)}\mathscr{D}_Bf \quad . 
\end{align}
Therefore
\begin{equation}
\tilde\xi^A = h^A -r\mathscr{D}^Af - \sum_{k=N+1}^\infty \frac{1}{k} r^k \gamma^{AB(k-1)}\mathscr{D}_Bf \quad .
\end{equation}
Given this expression for $\tilde\xi^A$, we see that the $ru$ equation gives
\begin{align}
0 = \tilde\chi_{ru}
& = \partial_u f + \partial_r\tilde\xi^r - 2r^{-1}\tilde\xi^r + \beta^A \mathscr{D}_A f \quad .
\end{align}
Multiplying by $r$ and setting $r=0$ gives $\tilde\xi^{r(0)}=0$.
At order $1$, we find $\tilde\xi^{r(1)}=\partial_uf$. 
At order $r$ the equation is trivial; since $\tilde\xi^{r(2)}$ is undetermined by this equation, we denote it by $g\defeq\tilde\xi^{r(2)}$.
At order $r^k$ for $1<k<N+1$, we find that $\tilde\xi^{r(k+1)}=0$.
Finally, at order $r^k$ for $k\geq N+1$, we find that
\begin{equation}
\tilde\xi^{r(k+1)} = -\frac{1}{k-1}\beta^{A(k)}\mathscr{D}_Af \quad .
\end{equation}
Therefore,
\begin{equation}
\tilde\xi^r = r\partial_u f + r^2g - \sum_{k=N+2}^\infty \frac{1}{k-2} r^k \beta^{A(k-1)}\mathscr{D}_Af \quad .
\end{equation}

\noindent
Thus all that remains to obtain equations (\ref{eq:asymptotic-symmetries-components}) is to derive the conditions
on the low order expansion terms.

\medskip\noindent
The $uA$ component is
\begin{equation}
\tilde\chi_{uA} = \tilde\nabla_A\tilde\xi^r - 2\alpha\tilde\nabla_A\tilde\xi^u - \beta_B\tilde\nabla_A\tilde\xi^B + \gamma_{AB}\tilde\nabla_u\tilde\xi^B + 2r^{-1}\beta_A\tilde\xi^r \quad ,
\end{equation}
and so the $uA$ equation gives
\begin{align}
O(r^{N+1}) = \tilde\chi_{uA}
& = r^2\mathscr{D}_A(g-\lambda f) + \gamma_{AB}\partial_uh^B \quad .
\end{align}
At order $r^2$, since $d>4$, we find that $\partial_A(g-\lambda f)=0$. At order $1$, we find that $\partial_u h_A = \partial_u h^A = 0$.
The $uu$ component is
\begin{equation}
\frac{1}{2}\tilde\chi_{uu} = \tilde\nabla_u\tilde\xi^r - 2\alpha\tilde\nabla_u\tilde\xi^u - \beta_A\tilde\nabla_u\tilde\xi^A + 2r^{-1}\alpha\tilde\xi^r \quad ,
\end{equation}
so that the $uu$ equation gives
\begin{align}
O(r^{N+2}) = \tilde\chi_{uu}
& = r\partial_u^2f + r^2\partial_u(g-\lambda f) - r^{N+1}\beta_A^{(N+1)}\partial_uh^A \quad .
\end{align}
At order $r$, we find that $\partial_u^2f=0$, i.e. $f=f_0+f_1$ where $f_0$ and $f_1$ are functions of $x^A$ alone.
At order $r^2$ we find that $\partial_u(g-\lambda f)=0$.
Since we have already found that $\partial_A(g-\lambda f)=0$,  we see that $g-\lambda f=g_0$ where $g_0$ is a constant.
Therefore $g=\lambda f + g_0$.
At order $r^{N+1}$, the equation is automatically satisfied since $\partial_uh^A = 0$.

The $AB$ component is
\begin{equation}
\frac{1}{2}\tilde\chi_{AB} = -\beta_{(A}\tilde\nabla_{B)}\tilde\xi^u + \gamma_{C(A}\tilde\nabla_{B)}\tilde\xi^C - r^{-1}\gamma_{AB}\tilde\xi^r \quad ,
\end{equation}
so that the $AB$ equation gives
\begin{equation}
O(r^N) = s_{C(A}\mathscr{D}_{B)}h^C - rs_{C(A}\mathscr{D}_{B)}\mathscr{D}^Cf - [(\partial_uf) - rg]s_{AB} \quad . 
\end{equation}
Separating this equation into each order of $r$ and $u$ gives us precisely equations (\ref{eq:asymptotic-symmetries-conditions-low-order}).

Finally, we check that our solution is consistent with the trace $AB$ equation, $s^{AB}\tilde\chi_{AB}=O(r^{2N})$.
We find that
\begin{align}
\frac{1}{2}s^{AB}\tilde\chi_{AB}
\iftoggle{complete}{
& = s^{AB}(\tilde{g}_{c(A}\tilde\nabla_{B)}\tilde\xi^c - r^{-1}\gamma_{AB}\tilde\xi^r) \nonumber \\
}{}
& = s^{AB}(-\beta_A\tilde\nabla_B\tilde\xi^u + \gamma_{AC}\tilde\nabla_B\tilde\xi^C) - r^{-1}s^{AB}\gamma_{AB}\tilde\xi^r \nonumber \\
\iftoggle{complete}{
& = s^{AB}(-\beta_A\partial_B\tilde\xi^u - \beta_A\tilde\Gamma^u{}_{Bb}\tilde\xi^b + \gamma_{AC}\partial_B\tilde\xi^C + \gamma_{AC}\tilde\Gamma^C{}_{Bb}\tilde\xi^b) - (d-2)r^{-1}\tilde\xi^r +O(r^{2N}) \nonumber \\
& = -\beta^A\partial_A\tilde\xi^u + s^{AB}\gamma_{AC}\partial_B\tilde\xi^C - \beta^A(\tilde\Gamma^u{}_{rA}\tilde\xi^r + \tilde\Gamma^u{}_{uA}\tilde\xi^u + \tilde\Gamma^u{}_{AC}\tilde\xi^C) \nonumber \\
& + s^{AB}\gamma_{AC}(\tilde\Gamma^C{}_{rB}\tilde\xi^r + \tilde\Gamma^C{}_{uB}\tilde\xi^u + \tilde\Gamma^C{}_{BD}\tilde\xi^D) - (d-2)r^{-1}\tilde\xi^r + O(r^{2N}) \nonumber \\
& = -\beta^A\partial_A\tilde\xi^u + s^{AB}\gamma_{AC}\partial_B\tilde\xi^C - \beta^A\left(\frac{1}{2}\partial_r\beta_A\tilde\xi^u - \frac{1}{2}(\partial_r\gamma_{AB})\tilde\xi^B\right) \nonumber \\
& + s^{AB}\gamma_{AC}\frac{1}{2}\gamma^{CD}(\partial_r\gamma_{BD})\tilde\xi^r \nonumber \\
& + s^{AB}\gamma_{AC}\left(\frac{1}{2}\gamma^{CD}\partial_u\gamma_{BD} - \gamma^{CD}(D_{[B}\beta_{D]}) + \frac{1}{2}\beta^C\partial_r\beta_B\right)\tilde\xi^u \nonumber \\
& + s^{AB}\gamma_{AC}\left(\tilde\Lambda^C{}_{BD} - \frac{1}{2}\beta^C(\partial_r\gamma_{BD})\tilde\xi^D\right) - (d-2)r^{-1}\tilde\xi^r + O(r^{2N})\nonumber \\
& = -\beta^A\partial_A\tilde\xi^u + s^{AB}\gamma_{AC}\partial_B\tilde\xi^C + s^{AB}\gamma_{AC}\tilde\Lambda^C{}_{BD}\tilde\xi^D - (d-2)r^{-1}\tilde\xi^r \nonumber \\
& + \frac{1}{2}s^{AB}\delta^D{}_A(\partial_r\gamma_{BD})\tilde\xi^r + s^{AB}\delta^D{}_A\left(\frac{1}{2}\partial_u\gamma_{BD} - D_{[B}\beta_{D]}\right)\tilde\xi^u + O(r^{2N}) \nonumber \\
}{}
& = -\beta^A\partial_A\tilde\xi^u + s^{AB}\gamma_{AC}D_B\tilde\xi^C - r^{-1}s^{AB}\gamma_{AB}\tilde\xi^r \nonumber \\
& + \frac{1}{2}s^{AB}(\partial_r\gamma_{AB})\tilde\xi^r + \frac{1}{2}s^{AB}\partial_u\gamma_{AB}\tilde\xi^u + O(r^{2N}) \quad .
\end{align}
Now, since $\tilde\xi^r=O(r)$ and $s^{AB}\gamma_{AB}=(d-2)+O(r^{2N})$, both terms in the last line are of order $r^{2N}$.
We can thus write the trace $AB$ condition as
\begin{equation}
0 = -\beta^A\partial_A\tilde\xi^u + s^{AB}\gamma_{AC}D_B\tilde\xi^C - (d-2)r^{-1}\tilde\xi^r + O(r^{2N}) \quad .
\end{equation}
We now substitute the expansions for $\tilde\xi^r$, $\tilde\xi^u$ and $\tilde\xi^A$ as well as $\beta_A$ and $\gamma_{AB}$.
We find
\begin{align}
0
\iftoggle{complete}{
& = -\sum_{i=N+1}^\infty r^i \beta_A^{(i)}\mathscr{D}^Af + s^{AB}s_{AC}\mathscr{D}_B \left(h^C - r\mathscr{D}^Cf + \sum_{i=N+1}^\infty \frac{1}{i} r^i s^{CD}\gamma_{DE}^{(i-1)}\mathscr{D}^Ef\right) \\
& + s^{AB}\sum_{i=N}^\infty r^i\gamma_{AC}^{(i)}\mathscr{D}_B(h^C - r\mathscr{D}^Cf) + s^{AB}s_{AC}\sum_{i=N}^\infty r^i \tilde\Lambda^C{}_{BD}^{(i)} (h^D - r\mathscr{D}^Df) \\
& - (d-2)\left(\partial_uf + r(g_0+\lambda f) - \sum_{i=N+2}^\infty \frac{1}{i-2} r^{i-1} \beta_A^{(i-1)} \mathscr{D}^Af\right) + O(r^{2N}) \\
& = -\sum_{i=N+1}^\infty r^i \beta_A^{(i)}\mathscr{D}^Af + \delta^B{}_C\mathscr{D}_B \left(h^C - r\mathscr{D}^Cf + \sum_{i=N+1}^\infty \frac{1}{i} r^i s^{CD}\gamma_{DE}^{(i-1)}\mathscr{D}^Ef\right) \\
& + \sum_{i=N}^\infty r^i\gamma_{AC}^{(i)}(\mathscr{D}^A h^C - r\mathscr{D}^A\mathscr{D}^Cf) + \delta^B{}_C\sum_{i=N}^\infty r^i \tilde\Lambda^C{}_{BD}^{(i)} (h^D - r\mathscr{D}^Df) \\
& - (d-2)\left(\partial_uf + r(g_0+\lambda f) - \sum_{i=N+1}^\infty \frac{1}{i-1} r^i \beta_A^{(i)} \mathscr{D}^Af\right) + O(r^{2N}) \\
}{}
& = \mathscr{D}_Bh^B - r\mathscr{D}_B\mathscr{D}^Bf - (d-2)(\partial_uf + r(g_0+\lambda f)) \\
& + \sum_{k=N+1}^{2N-1} \frac{d-1-k}{k-1} r^k \beta^{A(k)}\mathscr{D}_Af - \sum_{k=N+1}^{2N-1} \frac{1}{k} r^k (\mathscr{D}_B\gamma^{AB(k-1)})\mathscr{D}_Af \\
& + \sum_{k=N}^{2N-1} r^k \tilde\Lambda^B{}_{BD}^{(k)} (h^D - r\mathscr{D}^Df) \nonumber \\
& + \sum_{k=N}^{2N-1} r^k\gamma_{AB}^{(k)}\mathscr{D}^A h^B - \sum_{k=N}^{2N-2} \frac{k+1}{k} r^{k+1} \gamma_{AB}^{(k)}\mathscr{D}^A\mathscr{D}^Bf + O(r^{2N}) \quad ,
\end{align}
where we note that as usual the indices of $h^A$ and $\mathscr{D}_A$ are raised and lowered with $s_{AB}$.
The terms in the first line cancel due to the conditions on the lowest order expansion coefficients (\ref{eq:asymptotic-symmetries-conditions-low-order}).
\iftoggle{complete}{
\begin{align*}
 \mathscr{D}_A h^A - (d-2)(\partial_uf)
&= s^{AB}(\mathscr{D}_A h_B - s_{AB}f_1) = 0 \\
 \mathscr{D}_A\mathscr{D}^Af - (d-2)(\lambda f + g_0) 
&= s^{AB}\mathscr{D}_A\mathscr{D}_B(f_0 + uf_1) \\
&= s^{AB}(\lambda f_0 + g_0 + u\lambda f_1)s_{AB} - (d-2)(\lambda f + g_0) = 0
\end{align*}
}{}
In addition, these conditions imply $\mathscr{D}^Ah^B$ and $\mathscr{D}^A\mathscr{D}^Bf$ are proportional to $s^{AB}$.
Therefore, since $s^{AB}\gamma_{AB}^{(k)}=0$ for $0<k<2N$, it follows that the last two terms are zero.
\iftoggle{complete}{
\begin{align*}
\mathscr{D}^Ah^B &= f_1s^{AB} \\
\mathscr{D}^A\mathscr{D}^B f
& = \mathscr{D}^A\mathscr{D}^B (f_0 + uf_1) \\
& = -(\lambda f_0 + g_0 + u\lambda f_1)s^{AB}
\end{align*}
}{}
By Lemma \ref{thm:gamma-christoffel-higher-order}, the second to last line is also zero.
Therefore the trace $AB$ condition reduces to
\begin{align}
0
& = \sum_{k=N+1}^{2N-1} \frac{d-1-k}{k-1} r^k \beta^{A(k)}\mathscr{D}_Af - \sum_{k=N+1}^{2N-1} \frac{1}{k} r^k (\mathscr{D}_B\gamma^{AB(k-1)})\mathscr{D}_Af + O(r^{2N}) \nonumber \\
& = \sum_{k=N+1}^{2N-1} r^k \left( \frac{d-1-k}{k-1}\beta_A^{(k)} + \frac{1}{k}\mathscr{D}^B\gamma_{AB}^{(k-1)} \right)\mathscr{D}^Bf + O(r^{2N}) \quad ,
\end{align}
where in the second line we used that $\gamma^{AB(k)} = -s^{AC}s^{BD}\gamma_{CD}^{(k)}$ for $0<k<2N$.
Applying equation (\ref{eq:beta-higher-order-gamma}) to express $\beta_A^{(k)}$ in terms of $\gamma_{AB}^{(k-1)}$, we see that this is indeed satisfied.
\iftoggle{complete}{
\begin{equation*}
\beta_A^{(i)} = -\frac{(i-1)}{i(d-1-i)}\mathscr{D}^B\gamma_{AB}^{(i-1)}
\end{equation*}
}{}
The fact that this condition the trace is satisfied is a consistency check on equation (\ref{eq:beta-higher-order-gamma}) and on our expansions for $\xi^a$.
\end{proof}

%% file: BondiMass.tex
\iftoggle{thesis}{
}{}
\section{Bondi Mass}
\label{sec:bondi-mass}
\setcounter{section}{1}
\setcounter{theorem}{0}
A definition of the Bondi mass at ``time'' u is given, motivated by the Bondi mass formula given in \cite{Hollands:2003ie}.
An expression is obtained for the energy in terms of metric expansion coefficients.

All integrals over surfaces $\Sigma(r,u)$ of constant $r,u$ are calculated with respect to the volume element
\begin{equation}
\label{eq:cross-section-volume-element}
{}^{(d-2)}\tilde\epsilon_{a_1\ldots a_{d-2}} = \sqrt\gamma d^{d-2}x_{a_1\ldots a_{d-2}} \equiv \sqrt\gamma \ (dx^1 \wedge \cdots \wedge dx^{d-2})_{a_1\ldots a_{d-2}} \quad ,
\end{equation}
which is omitted from the integrals for simpler notation.

\begin{definition}
\label{def:bondi-news}
We define the ``News tensor'' $N_{AB}$ on $\scri^+$ to be
\begin{equation}
\label{eq:bondi-news-definition}
N_{AB} \defeq \left[ \frac{1}{r^{N-1}} \left(\tilde{S}_{ab} - \lambda\tilde{g}_{ab}\right) (\partial_A)^a (\partial_B)^b \right]_{r=0} \quad .
\end{equation}
\end{definition}
\begin{remark}
A definition of the News tensor is given by equation (61) of \cite{Hollands:2003ie}.
The definition given above is analogous to the one in that paper, but adapted to our gauge.
In particular, the definition in \cite{Hollands:2003ie} is given in a gauge where $\tilde{S}_{ab}$ is zero up to order $r^N$ and trace free up to order $r^{N+1}$, which is why the low order terms have been subtracted from $\tilde{S}_{ab}$ in the present definition.
Setting $r=0$ is equivalent to the pull back to null infinity in \cite{Hollands:2003ie}.
\end{remark}

\begin{definition}
\label{def:bondi-mass-density}
We define the Bondi mass density to be
\begin{equation}
\label{eq:bondi-mass-density-definition}
\mu
 = \frac{A}{8(d-3)\pi G} \frac{1}{r^{2N-2}} \left[ \frac{1}{2} \left(\tilde{S}^{ab} - \lambda\tilde{g}^{ab}\right) \tilde\nabla_a\tilde\nabla_b u
 - \frac{1}{r}\tilde{C}^{abcd}(\tilde\nabla_a u)(\tilde\nabla_b r)(\tilde\nabla_c u)\tilde\nabla_d r \right] \quad ,
\end{equation}
where $A>0$ is a normalisation constant and $G$ is the universal gravitational constant.
\end{definition}
\begin{definition}
\label{def:bondi-mass}
We define the Bondi mass on a cross section $\Sigma(0,u)$ of $\scri^+$ to be
\begin{equation}
\label{eq:bondi-mass-definition}
m(u) = \lim_{r\rightarrow0} \int_{\Sigma(r,u)} \mu \quad .
\end{equation}
\end{definition}
\begin{remark}
A formula for the Bondi mass in higher (even) dimensions is derived in \cite{Hollands:2003ie} and given by equation (99) of that paper.
We propose that equation (\ref{eq:bondi-mass-definition}) is the equivalent of this definition, adapted to our gauge.
To prove this, we would need to verify that the derivation of the formula in \cite{Hollands:2003ie} can be adapted to our gauge or explicitly obtain the transformation formula between the two gauges.
\end{remark}
\begin{definition}
\label{def:bondi-flux}
We define the Bondi flux on a cross section $\Sigma(0,u)$ of $\scri^+$ to be
\begin{equation}
\label{eq:bondi-flux-definition}
f(u) = \partial_u m(u) \quad .
\end{equation}
\end{definition}

\begin{theorem}
\label{thm:bondi-news}
The News tensor is given by
\begin{equation}
\label{eq:bondi-news-formula}
N_{AB} = -\partial_u\gamma_{AB}^{(N)} \quad .
\end{equation}
\end{theorem}
\begin{proof}
Rewriting equation (\ref{eq:bondi-news-definition}) in coordinate form, we find
\begin{equation}
\label{eq:bondi-news-coordinate}
N_{AB} = \left[\frac{1}{r^{N-1}}(\tilde{S}_{AB}-\lambda\gamma_{AB})\right]_{r=0} \quad .
\end{equation}
In Section \ref{sec:einstein-equations-calculation}, we wrote down an expression (\ref{eq:einstein-tensor-short}) for the Einstein equations in terms of $\tilde{S}_{ab}$.
Rearranging this equation yields an expression for $\tilde{S}_{ab}$.
In particular, the AB component, equation (\ref{eq:grtensors-einstein-compact-AB}), gives
\begin{align}
\label{eq:schouten-AB-low-order}
\tilde{S}_{AB}
\iftoggle{complete}{
& = -r^{-1}\partial_u\gamma_{AB} + 2r^{-2}\alpha\gamma_{AB} + O(r^{N}) \nonumber \\
}{}
& = - r^{N-1}\partial_u\gamma_{AB}^{(N)} + \lambda s_{AB} + O(r^N) \quad ,
\end{align}
and by Theorem \ref{thm:metric-expansion-low-order}
\begin{equation}
\label{eq:gamma-low-order}
\gamma_{AB} = s_{AB} + O(r^N) \quad .
\end{equation}
Substituting back into equation (\ref{eq:bondi-news-coordinate}) gives the desired expression for the News tensor, we find
\begin{align}
N_{AB}
\iftoggle{complete}{
& = \left[\frac{1}{r^{N-1}}\left(-r^{N-1}\partial_u\gamma_{AB}^{(N)} + \lambda s_{AB} - \lambda s_{AB} + O(r^N)\right)\right]_{r=0} \nonumber\\
}{}
& = \left[-\partial_u\gamma_{AB}^{(N)} + O(r)\right]_{r=0}
  = -\partial_u\gamma_{AB}^{(N)} \quad .
\end{align}
\end{proof}

\begin{theorem}
\label{thm:bondi-mass}
The Bondi mass on a cross section $\Sigma(0,u)$ of $\scri^+$ is given by
\begin{equation}
\label{eq:bondi-mass-formula}
m(u) = \frac{A}{8(d-3)\pi G} \int_{\Sigma(0,u)} \left(\frac{d-2}{8}\gamma^{AB(N)}\partial_u\gamma_{AB}^{(N)} - (d-2)(d-3)\alpha^{(2N+1)}\right) \quad .
\end{equation}
\end{theorem}
\noindent
To prove the Bondi mass formula, we proceed as follows.
First we use the definition, equation (\ref{eq:bondi-mass-definition}) to obtain an integral of the form
\begin{equation}
\label{eq:bondi-mass-formula-example}
\lim_{r\rightarrow0} \int_{\Sigma(r,u)} \left( \sum_k \text{terms} \ O(r^{-k}) + \text{terms} \ O(1) + \text{terms} \ O(r) \right) \quad .
\end{equation}
The terms $O(1)$ turn out to be precisely those appearing in the integral in equation (\ref{eq:bondi-mass-formula}).
The terms $O(r)$ do not affect the limit at $\scri^+$ (i.e. the limit as $r\rightarrow0$).
The terms in the sum superficially appear to be singular at $\scri^+$, so we must show that they do not contribute to the integral before we take the limit at $\scri^+$.
To do this, we show that they are equal to a total divergence $\mathscr{D}_A v^A$ on $\Sigma(r,u)$, for some vector $v^A$ on $\Sigma(r,u)$.
Then they integrate to zero by Gauss' theorem, since the boundary of $\Sigma(r,u)$ is the empty set.
To formalise this proof, we first calculate the Bondi mass density in the following lemma.

\begin{lemma}
\label{thm:bondi-mass-density}
The energy density $\mu$ can be written
\begin{align}
\label{eq:bondi-mass-density-formula}
\mu = \frac{A}{8(d-3)\pi G}\bigg[ &\frac{d-2}{8}\gamma^{AB(N)}\partial_u\gamma_{AB}^{(N)} - (d-2)(d-3)\alpha^{(d-1)} \nonumber \\
& -\sum_{k=0}^{2N-2} k(k+1) r^{k+1-2N}\alpha^{(k+2)} + O(r) \bigg] \quad .
\end{align}
\end{lemma}
\begin{proof}
Let us first write the energy density $\mu$ as
\begin{equation}
\label{eq:bondi-mass-density-split}
\mu = \frac{A}{8(d-3)\pi G}(\mu_1 + \mu_2) \quad ,
\end{equation}
where we have denoted the two terms of the density by
\begin{align}
\label{eq:bondi-mass-term-1-definition}
& \mu_1 = \frac{1}{2}\frac{1}{r^{2N-2}}\left(\tilde{S}^{ab}-\lambda\tilde{g}^{ab}\right) \tilde\nabla_a\tilde\nabla_b u \quad , \\
\label{eq:bondi-mass-term-2-definition}
& \mu_2 = -\frac{1}{r^{2N-1}}\tilde{C}^{abcd} (\tilde\nabla_a u)(\tilde\nabla_b r)(\tilde\nabla_c u)\tilde\nabla_d r \quad .
\end{align}
We calculate that
\begin{align}
\tilde\nabla_a\tilde\nabla_b u
\iftoggle{complete}{
& = \tilde\nabla_a\partial_b u \nonumber \\
& = \partial_a\partial_b u - \tilde\Gamma^c_{ab}\partial_c u \nonumber \\
}{}
& = -\tilde\Gamma^c{}_{ab}\tilde\nabla_c u
\iftoggle{complete}{ \nonumber \\
& = -\frac{1}{2}\tilde{g}^{uc}(\partial_a\tilde{g}_{bc}+\partial_b\tilde{g}_{ac} - \partial_c\tilde{g}_{ab}) \nonumber \\
& = -\frac{1}{2}\tilde{g}^{ur}(\partial_a\tilde{g}_{br}+\partial_b\tilde{g}_{ar} - \partial_r\tilde{g}_{ab}) \nonumber  \\
& }{}
 = \frac{1}{2}\partial_r\tilde{g}_{ab} \quad ,
\end{align}
and since
\begin{equation}
\partial_r\tilde{g}^{ab} = -\tilde{g}^{ac}\tilde{g}^{bd}\partial_r\tilde{g}_{cd} \quad ,
\end{equation}
we can express the first term of the energy density as
\begin{align}
\label{eq:bondi-mass-term-1-formula-halfway}
\mu_1
\iftoggle{complete}{
& = \frac{1}{4}\frac{1}{r^{2N-2}}\left(\tilde{S}^{ab}-\lambda\tilde{g}^{ab}\right)\partial_r\tilde{g}_{ab} \nonumber \\
}{}
& = -\frac{1}{4}\frac{1}{r^{2N-2}}\left(\tilde{S}_{ab}-\lambda\tilde{g}_{ab}\right)\partial_r\tilde{g}^{ab} \quad .
\end{align}
From equations (\ref{eq:grtensors-einstein-compact-rr}) and (\ref{eq:grtensors-einstein-compact-rA}), we find that
\begin{align}
& \tilde{S}_{rr} = 0 \quad , \\
& \tilde{S}_{rA} = -r^{-1}\partial_r\beta_A + r^{-1}\beta^B\partial_r\gamma_{AB} = O(r^{N-1}) \quad ,
\end{align}
and from equations (\ref{eq:schouten-AB-low-order}) and (\ref{eq:gamma-low-order}) that
\begin{align}
\tilde{S}_{AB} - \lambda\gamma_{AB}
\iftoggle{complete}{
& = -r^{N-1}\partial_u\gamma_{AB}^{(N)} + \lambda s_{AB} - \lambda \gamma_{AB} + O(r^N) \nonumber \\
}{}
& = -r^{N-1}\partial_u\gamma_{AB}^{(N)} + O(r^N) \quad .
\end{align}
Summing over the contracted indices in equation (\ref{eq:bondi-mass-term-1-formula-halfway}) we find that
\begin{align}
\label{eq:bondi-mass-term-1-formula}
\mu_1
\iftoggle{complete}{
& = -\frac{1}{4}\frac{1}{r^{2N-2}}\bigg[(\partial_r\tilde{g}^{rr})(\tilde{S}_{rr} - \lambda\tilde{g}_{rr}) + (\partial_r\tilde{g}^{ru})(\tilde{S}_{ru} - \lambda\tilde{g}_{ru}) \nonumber \\
& + (\partial_r\tilde{g}^{rA}(\tilde{S}_{rA} - \lambda\tilde{g}_{rA})  + (\partial_r\tilde{g}^{AB})(\tilde{S}_{AB} - \lambda\tilde{g}_{AB})\bigg] \nonumber \\
}{}
& = -\frac{1}{4}\frac{1}{r^{2N-2}}\bigg[(\partial_r\beta^A)\tilde{S}_{rA} + (\partial_r\gamma^{AB})(\tilde{S}_{AB} - \lambda\tilde{g}_{AB})\bigg] \nonumber \\
\iftoggle{complete}{
& = -\frac{1}{4}\frac{1}{r^{2N-2}}\bigg[r^{N-1}N\gamma^{AB(N)}(-r^{N-1}\partial_u\gamma_{AB}^{(N)}) + O\left(r^{2N-1}\right)\bigg] \nonumber \\
& = \frac{1}{4}\frac{1}{r^{2N-2}}\bigg[r^{2N-2}N\gamma^{AB(N)}\partial_u\gamma_{AB}^{(N)} + O\left(r^{2N-1}\right)\bigg] \nonumber \\
}{}
& = \frac{d-2}{8}\gamma^{AB(N)}\partial_u\gamma_{AB}^{(N)} + O(r) \quad .
\end{align}
This is the first term of equation (\ref{eq:bondi-mass-density-formula}) as required.

\medskip\noindent
It remains to derive an expression for $\mu_2$.
In coordinates, we have
\begin{equation}
\mu_2 = - \frac{1}{r^{2N-1}}\tilde{C}^{abcd}(\tilde\nabla_a u)(\tilde\nabla_b r)(\tilde\nabla_c u)\tilde\nabla_d r = -\frac{1}{r^{2N-1}}\tilde{C}^{urur} \quad .
\end{equation}
Raising the indices of $\tilde{C}_{abcd}$ with $\tilde{g}^{ab}$, we find
\begin{align}
\tilde{C}^{urur}
\label{eq:weyl-upper-1}
& = \tilde{g}^{au}\tilde{g}^{br}\tilde{g}^{cu}\tilde{g}^{dr}\tilde{C}_{abcd} \nonumber \\
\iftoggle{complete}{
& = \tilde{g}^{br}\tilde{g}^{dr}\tilde{C}_{rbrd} \nonumber \\
& = (2\alpha+\beta_C\beta^C)^2\tilde{C}_{rrrr} + 2(2\alpha+\beta_C\beta^C)\tilde{C}_{rrru} + 2(2\alpha+\beta_C\beta^C)\beta^A\tilde{C}_{rrrA} \nonumber \\
& + \tilde{C}_{ruru} + 2\beta^A\tilde{C}_{rurA} + \beta^A\beta^B\tilde{C}_{rArB} \nonumber \\
}{}
& = \tilde{C}_{ruru} + 2\beta^A\tilde{C}_{rurA} + O(r^{2N+2})\quad ,
\end{align}
where we have used that $\tilde{C}_{abcd}=-\tilde{C}_{bacd}$.
The relevant components of the Weyl tensor are given in Appendix \ref{app:grtensors-weyl}.
Substituting these into equation (\ref{eq:weyl-upper-1}) we find that
\begin{equation}
\tilde{C}^{urur} = \partial_r^2\alpha - 2r^{-1}(\partial_r\alpha - r^{-1}\alpha) + O(r^{2N}) \quad .
\end{equation}
Expanding this in powers of $r$ yields
\begin{equation}
\tilde{C}^{urur} = \sum_{k=0}^{2N-1} k(k+1) r^k \alpha^{(k+2)} + O(r^{2N}) \quad ,
\end{equation}
and substituting this expansion back into equation (\ref{eq:bondi-mass-term-2-definition}) gives
\begin{equation}
\label{eq:bondi-mass-term-2-formula}
\mu_2 = -\sum_{k=0}^{2N-2} k(k+1) r^{k+1-2N} \alpha^{(k+2)} - (d-2)(d-3)\alpha^{(2N+1)} + O(r) \quad .
\end{equation}
Substituting equations (\ref{eq:bondi-mass-term-1-formula}) and (\ref{eq:bondi-mass-term-2-formula}) into (\ref{eq:bondi-mass-density-split}) gives the desired expression (\ref{eq:bondi-mass-density-formula}) for $\mu$.
\end{proof}

\begin{proof}[Proof of Theorem \ref{thm:bondi-mass}]
Substituting the formula (\ref{eq:bondi-mass-density-formula}) into the definition of the Bondi mass, equation (\ref{eq:bondi-mass-definition}), we obtain an integral of the form (\ref{eq:bondi-mass-formula-example}).
It remains to show that the singular terms in the sum are equal to total divergences.
These terms take the form
\begin{equation}
k(k+1)\alpha^{(k+2)} r^{k+1-2N}
\end{equation}
for $0\leq k<2N-1$ and therefore we can apply Theorem \ref{thm:metric-components-higher-order-halfway} to rewrite them as
\begin{equation}
k(k+1)\alpha^{(k+2)} = \frac{k^2-1}{2(d-3-k)}\mathscr{D}^A\beta_A^{(k+1)} \quad ,
\end{equation}
which is precisely the form we wanted.
These coefficients are total divergences\footnote{In fact, total divergences on surfaces $\Sigma(r,u)$ for $r\neq0$ should be with respect to the derivative operator $D_A$. However, the correction terms are of order $O(r^N)$, and so we can write $D^A\beta_A^{(k)}=\mathscr{D}^A\beta_A^{(k)}$ for $0\leq k<2N+1$ for $k<2N+1$, since $\beta_A=O(r^{N+1})$.}
on the surfaces $\Sigma(r,u)$.
The integral therefore becomes
\begin{align}
\int_{\Sigma(r,u)} \mu
\iftoggle{complete}{
& = \frac{1}{d-3} \int_{\Sigma(r,u)} \left(\frac{d-2}{8}\gamma^{AB(N)}\partial_u\gamma_{AB}^{(N)} + (d-2)(d-3)\alpha^{(2N+1)} + \frac{1}{r^{2N-1}}\mathscr{D}_A v^A + O(r)\right) \\
}{}
 = \frac{A}{8(d-3)\pi G} \int_{\Sigma(r,u)} \bigg( \frac{d-2}{8}&\gamma^{AB(N)}\partial_u\gamma_{AB}^{(N)} \nonumber \\
& + (d-2)(d-3)\alpha^{(2N+1)}\bigg) + O(r) \quad ,
\end{align}
where the total divergence terms have dropped out of the integral by Gauss' law.
The integrand is now $O(1)$ and we can take the limit at $\scri^+$ without issue.
Doing this, the terms $O(r)$ drop out and we obtain the Bondi mass formula (\ref{eq:bondi-mass-formula}).
\end{proof}

\begin{theorem}
\label{thm:bondi-flux}
The flux of energy at time $u$ (i.e. the energy radiated through the cross section $\Sigma(0,u)$ of $\scri^+$) is
\begin{equation}
\label{eq:bondi-flux-formula}
f(u) = -\frac{A}{32\pi G} \int_{\Sigma(0,u)} N^{AB} N_{AB} \quad ,
\end{equation}
where the indices of $N_{AB}$ are raised with $s^{AB}$.
\end{theorem}
\begin{remark}
Since $s_{AB}$ is a Riemannian metric, it follows from equation (\ref{eq:bondi-flux-formula}) that the flux is negative or zero.
This means that a positive amount of energy is radiated away at $\scri^+$, as we would expect.
\iftoggle{complete}{

Note: we know the flux is negative because the Riemannian metric $s_{AB}$ can be written as the identity matrix in suitable coordinates.
The quantity $N^{AB}N_{AB}$ is a scalar, and can therefore coordinate independent.
If it is positive in an orthonormal basis then it is positive in all coordinate systems.
}{}
\end{remark}
\begin{proof}
We apply the definition, equation (\ref{eq:bondi-flux-definition}).
We take a $u$ derivative of the Bondi mass (\ref{eq:bondi-mass-formula}) and then apply Einstein's equations to obtain equation (\ref{eq:bondi-flux-formula}).
The volume element ${}^{(d-2)}\tilde\epsilon_{a_1\ldots a_{d-2}}$ is independent of $u$ on $\scri^+$ because $s_{AB}$ is.
As a result, the $u$ derivative can be taken inside the integral.
\iftoggle{complete}{
The derivative of the volume element is zero, $\partial_u \sqrt{s}d^{d-2}x = 0$,
and the geometry and size of the cross sections does not change, so the ``derivative of the integral'' is also zero.
}{}

\medskip\noindent
Differentiating equation (\ref{eq:bondi-mass-formula}) with respect to $u$ thus gives
\begin{align}
\label{eq:bondi-flux-formula-halfway}
\partial_u m(u) = \frac{A}{8(d-3)\pi G} \int_{\Sigma(0,u)} \bigg( \frac{d-2}{8}(&\partial_u\gamma^{AB(N)})\partial_u\gamma_{AB}^{(N)} + \frac{d-2}{8}\gamma^{AB(N)}\partial_u^2\gamma_{AB}^{(N)} \nonumber \\
& - (d-2)(d-3)\partial_u\alpha^{(2N+1)} \bigg) \quad .
\end{align}
Then the $uu$ component of Einstein's equations at order $r^{2N}$ gives
\begin{align}
\label{eq:einstein-uu-2N}
\iftoggle{complete}{
0
& = -(d-1)s^{AB}\partial_u^2\gamma_{AB}^{(2N)} - (d-1)\gamma^{AB(N)}\partial_u^2\gamma_{AB}^{(N)} - \frac{d-1}{2}\left(\partial_u\gamma^{AB(N)}\right)\partial_u\gamma_{AB} \nonumber \\
& + 2(d-1)(d-2)\partial_u\alpha^{(d-1)} + 2(d-1)s^{AB}\mathscr{D}_A\mathscr{D}_B \alpha^{(d-2)} + \lambda\mathcal{R}^{(d-4)} \nonumber \\
0
& = -s^{AB}\partial_u^2\gamma_{AB}^{(2N)} - \gamma^{AB(N)}\partial_u^2\gamma_{AB}^{(N)} - \frac{1}{2}\left(\partial_u\gamma^{AB(N)}\right)\partial_u\gamma_{AB} \nonumber \\
& + 2(d-2)\partial_u\alpha^{(d-1)} + 2s^{AB}\mathscr{D}_A\mathscr{D}_B \alpha^{(d-2)} + \frac{\lambda}{d-1}\mathcal{R}^{(d-4)} \nonumber \\
}{}
0
& = -s^{AB}\partial_u^2\gamma_{AB}^{(2N)} - \gamma^{AB(N)}\partial_u^2\gamma_{AB}^{(N)} - \frac{1}{2}\left(\partial_u\gamma^{AB(N)}\right)\partial_u\gamma_{AB} \nonumber \\
& + 2(d-2)\partial_u\alpha^{(d-1)} + \mathscr{D}_A w^A \quad ,
\end{align}
where we have used Lemma \ref{thm:metric-components-higher-order-halfway} to express terms as total divergences.
We can simplify the first line using the $rr$ equation.
At order $r^{2N-2}$, it gives
\begin{equation}
s^{AB}\gamma_{AB}^{(2N)} = -\frac{3d-10}{8(d-3)}\gamma^{AB(N)}\gamma_{AB}^{(N)} \quad .
\end{equation}
Taking two $u$ derivatives, we find
\begin{equation}
s^{AB}\partial_u^2\gamma_{AB}^{(d-2)} = -\frac{3d-10}{4(d-3)}(\partial_u\gamma^{AB(N)})\partial_u\gamma_{AB}^{(N)} -\frac{3d-10}{4(d-3)}\gamma^{AB(N)}\partial_u^2\gamma_{AB}^{(N)} \quad ,
\end{equation}
and substituting this back into equation (\ref{eq:einstein-uu-2N}) gives
\begin{align}
0
\iftoggle{complete}{
& = \frac{3d-10}{4(d-3)}\partial_u\gamma^{AB(N)}\partial_u\gamma_{AB}^{(N)} +\frac{3d-10}{4(d-3)}\gamma^{AB(N)}\partial_u^2\gamma_{AB}^{(N)} \nonumber \\
& - \gamma^{AB(N)}\partial_u^2\gamma_{AB}^{(N)} - \frac{1}{2}\left(\partial_u\gamma^{AB(N)}\right)\partial_u\gamma_{AB}^{(N)} \nonumber \\
& + 2(d-2)\partial_u\alpha^{(d-1)} + 2s^{AB}\mathscr{D}_A\mathscr{D}_B \alpha^{(d-2)} + \frac{\lambda}{d-1}\mathcal{R}^{(d-4)} \nonumber \\
& = \frac{3d-10-4(d-3)}{4(d-3)}\gamma^{AB(N)}\partial_u^2\gamma_{AB}^{(N)} + \frac{3d-10-2(d-3)}{4(d-3)}\left(\partial_u\gamma^{AB(N)}\right)\partial_u\gamma_{AB}^{(N)} \nonumber \\
& + 2(d-2)\partial_u\alpha^{(d-1)} + 2s^{AB}\mathscr{D}_A\mathscr{D}_B \alpha^{(d-2)} + \frac{\lambda}{d-1}\mathcal{R}^{(d-4)} \nonumber \\
}{}
& = -\frac{d-2}{4(d-3)}\gamma^{AB(N)}\partial_u^2\gamma_{AB}^{(N)} + \frac{d-4}{4(d-3)}\left(\partial_u\gamma^{AB(N)}\right)\partial_u\gamma_{AB}^{(N)} \nonumber \\
& + 2(d-2)\partial_u\alpha^{(d-1)} + \mathscr{D}_A w^A \quad .
\end{align}
Multiplying throughout by $(d-3)/2$, this can be rewritten as
\begin{align}
\frac{d-2}{8}\gamma^{AB(N)}\partial_u^2\gamma_{AB}^{(N)} &+ \frac{d-2}{8}\left(\partial_u\gamma^{AB(N)}\right)\partial_u\gamma_{AB} - (d-2)(d-3)\partial_u\alpha^{(d-1)} \nonumber \\
& = \frac{d-3}{4}\left(\partial_u\gamma^{AB(N)}\right)\partial_u\gamma_{AB}^{(N)} + \mathscr{D}_A w^A \quad .
\end{align}
This allows us to rewrite the Bondi flux formula (\ref{eq:bondi-flux-formula-halfway}) as
\begin{align}
\label{eq:bondi-flux-formula-halfway-2}
\partial_u m(u) = \frac{A}{32\pi G} \int_{\Sigma(0,u)} \left(\left\{\partial_u\gamma^{AB(N)}\right\}\partial_u\gamma_{AB}^{(N)} + \mathscr{D}_A w^A\right) \quad .
\end{align}
The total divergence term integrates to zero by Gauss' law.
We know from Theorem \ref{thm:bondi-news} that $\partial_u\gamma_{AB}^{(N)}=-N_{AB}$ and so it follows that
\begin{align}
N^{AB}
\iftoggle{complete}{
& = s^{AC}s^{BD}N_{CD} \nonumber\\
}{}
& = - s^{AC}s^{BD} \partial_u\gamma_{CD}^{(N)} \nonumber\\
& = \partial_u\gamma^{AB(N)} \quad .
\end{align}
Thus, equation (\ref{eq:bondi-flux-formula-halfway-2}) becomes
\begin{equation}
\partial_u m(u) = -\frac{A}{32\pi G} \int_{\Sigma(0,u)} N^{AB}N_{AB} \quad ,
\end{equation}
which is the desired formula (\ref{eq:bondi-flux-formula}).
\end{proof}

\numberwithin{equation}{section}

%% file: Spinors.tex
\section{Spinors in Gaussian Null Coordinates}
\label{sec:spinors}
Basic spinor notation is introduced, which is used in Section \ref{sec:dirac-spinor} when solving the Dirac equation.
A basis is fixed, and thus the spinor connection coefficients can be calculated.
A representation is also selected for the flat space gamma matrices.

We assume the spacetime admits a spin structure (in a spacetime without a spin structure, a spinorial positivity proof can of course not be performed).
In even dimension $d$, spinors have $2^{d/2}$ components.
A spinor $\psi$ is viewed as a column vector and its Hermitian transpose (conjugate transpose) $\psi^\ast$ is a row vector, as is its Dirac conjugate $\bar\psi$ (see equation (\ref{eq:dirac-conjugate})).
Spinor indices are suppressed for simpler notation.

\subsection{Basis}
\label{sec:spinors-basis}

\begin{definition}
We define a ``Gaussian null basis'' to be a basis $e^{\mu a}$ of the physical spacetime $(M,g_{ab})$ such that
\begin{equation}
g_{ab} e^{\mu a} e^{\nu b} = \lambda^{\mu\nu}
\end{equation}
where
\begin{equation}
\lambda_{\mu\nu} = \left(
\begin{matrix}
0 & 1 & 0\\
1 & 0 & 0\\
0 & 0 & \delta_{\hat{A}\hat{B}}
\end{matrix}
\right) \quad ,
\end{equation}
and
\begin{equation}
\lambda^{\mu\nu} \equiv (\lambda^{-1})^{\mu\nu} = \lambda_{\mu\nu} \quad .
\end{equation}
Lower case Greek indices are labelling indices which run through $+,-,1,\ldots, d-~2$ (the choice of ``+'' and ``--'' as labelling indices is purely notational).
Upper case Roman letters with a ``hat'' are labelling indices running from $1$ to $d-2$.
In this basis, the metric is given by
\begin{equation}
\label{eq:metric-spinor-basis}
g_{ab} = \lambda_{\mu\nu} e^\mu_a e^\nu_b \quad .
\end{equation}
\iftoggle{complete}{
It can then be shown that $g_{\mu\nu} = \lambda_{\mu\nu}$.
\begin{align*}
\lambda_{\mu\nu} 
& = g_{ab} e^a_\mu e^b_\nu \\
& = g_{\rho\sigma} e^\rho_a e^\sigma_b e^a_\mu e^b_\nu \\
& = g_{ac} e^a_\mu e^{\rho c} g_{bd} e^b_\nu e^{\sigma d} g_{\rho\sigma} \\
& = \lambda^\rho{}_\mu \lambda^\sigma{}_\nu g_{\rho\sigma} \\
& = g_{\mu\nu} \quad .
\end{align*}
and therefore,
\begin{equation}
\label{eq:spinor-basis-metric-components}
g_{ab} = \lambda_{\mu\nu} e^\mu_a e^\nu_b \quad .
\end{equation}
}{}
\end{definition}
\noindent
A Gaussian null basis $\tilde{e}^{\mu a}$ can be constructed in a similar manner for the unphysical spacetime, so that
\begin{equation}
\label{eq:unphysical-metric-basis}
\tilde{g}_{ab} = \lambda_{\mu\nu}\tilde{e}^\mu_a\tilde{e}^\nu_b \quad .
\end{equation}
Since $g^{ab}=r^2\tilde{g}^{ab}$, it follows that the physical and unphysical basis vectors are related according to
\begin{equation}
\label{eq:spinor-basis-conformal-relation}
e^{\mu a} = r\tilde{e}^{\mu a} \quad .
\end{equation}
%
We choose the following (nonunique) basis vectors for the unphysical spacetime,
\begin{subequations}
\label{eq:spinor-basis}
\begin{align}
\label{eq:spinor-basis-+}
\tilde{e}^{+a} &= \partial_r^a \quad , \\
\label{eq:spinor-basis--}
\tilde{e}^{-a} &= \alpha\partial_r^a + \partial_u^a \quad , \\
\label{eq:spinor-basis-A}
\tilde{e}^{\hat{A}a} &= \beta_A\hat{e}^{\hat{A}A}\partial_r^a + \hat{\tilde{e}}^{\hat{A}A}\partial_A^a \quad ,
\end{align}
\end{subequations}
where the vectors $\hat{\tilde{e}}^{\hat{A}A}$ form an orthonormal basis on the cross sections $\Sigma(r,u)$.
That is,
\begin{equation}
\gamma_{AB} \hat{\tilde{e}}^{\hat{A}A}\hat{\tilde{e}}^{\hat{B}B} = \delta^{\hat{A}\hat{B}} \quad ,
\end{equation}
and in this basis the induced metric $\gamma_{AB}$ is given by
\begin{equation}
\gamma_{AB} = \delta_{\hat{A}\hat{B}} \hat{\tilde{e}}^{\hat{A}}_A \hat{\tilde{e}}^{\hat{B}}_B \quad .
\end{equation}
It can be checked that this basis does indeed satisfy equation (\ref{eq:unphysical-metric-basis}),
\begin{align}
\lambda_{\mu\nu} \tilde{e}^\mu_a \tilde{e}^\nu_b
\iftoggle{complete}{
& = 2 \tilde{e}^+_{(a} \tilde{e}^-_{b)} + \delta_{\hat{A}\hat{B}} \hat{\tilde{e}}^{\hat{A}}_a \hat{\tilde{e}}^{\hat{B}}_b \nonumber \\
}{}
& = 2(dr_{(a} - \alpha du_{(a} - \beta_A dx^A_{(a})du_{b)} + \gamma_{AB} dx^A_a dx^B_b = \tilde{g}_{ab} \quad .
\end{align}

\subsection{Gamma Matrices}
\label{sec:spinors-gamma-matrices}

\begin{definition}
\label{def:gamma-matrices-flat}
The flat space gamma matrices, $\Gamma_\mu$, are $2^{d/2}\times2^{d/2}$ matrices satisfying the defining relation
\begin{equation}
\label{eq:gamma-matrices-flat-definition}
\{\Gamma_\mu, \Gamma_\nu\} = 2\lambda_{\mu\nu} I \quad ,
\end{equation}
where $\{\cdot, \cdot\}$ denotes the anticommutator
\begin{equation}
\{A, B\} = AB + BA \quad ,
\end{equation}
and $I$ is the $2^{d/2}\times2^{d/2}$ identity matrix.
The matrices satisfying (\ref{eq:gamma-matrices-flat-definition}) form an equivalence class and an explicit choice of such matrices is called a representation.
We specify our representation later in this section.
\end{definition}
\noindent
Equation (\ref{eq:gamma-matrices-flat-definition}) can be split into components to give
\begin{subequations}
\begin{align}
& \Gamma_+^2 = \Gamma_-^2 = 0 \quad , \\
& \{\Gamma_+, \Gamma_-\} = 2 I \quad , \\
\label{eq:gamma-matrices-flat-definition-xA}
& \{\Gamma_+, \Gamma_{\hat{A}}\} = \{\Gamma_-, \Gamma_{\hat{A}}\} = 0 \quad ,\\
& \{\Gamma_{\hat{A}}, \Gamma_{\hat{B}}\} = 2\delta_{\hat{A}\hat{B}} I \quad .
\end{align}
\end{subequations}
It follows from (\ref{eq:gamma-matrices-flat-definition-xA}) that
\begin{subequations}
\begin{align}
& [\Gamma_+,\Gamma_{\hat{A}}] = 2\Gamma_+\Gamma_{\hat{A}} \quad , \\
& [\Gamma_-,\Gamma_{\hat{A}}] = 2\Gamma_-\Gamma_{\hat{A}} \quad .
\end{align}
\end{subequations}

\begin{definition}
The curved space gamma matrices, $\Gamma_a$, are given by
\begin{equation}
\Gamma_a = \Gamma_\mu e^\mu_a \quad .
\end{equation}
It follows from equations (\ref{eq:metric-spinor-basis}) and (\ref{eq:gamma-matrices-flat-definition}) that
\begin{align}
\label{eq:spinor-gamma-matrices-condition-curved}
\{\Gamma_a, \Gamma_b\}
& = \{\Gamma_\mu, \Gamma_\nu\} e^\mu_a e^\nu_b \nonumber \\
& = 2\lambda_{\mu\nu} I e^\mu_a e^\nu_b \nonumber \\
& = 2g_{ab} I \quad .
\end{align}
\end{definition}

\begin{definition}
The unphysical curved space gamma matrices $\tilde\Gamma_a$ are defined by
\begin{equation}
\tilde\Gamma_a = \Gamma_\mu \tilde{e}^\mu_a \quad .
\end{equation}
\noindent
It follows from equation (\ref{eq:spinor-basis-conformal-relation}) that the physical and unphysical curved space gamma matrices are related by
\begin{equation}
\Gamma^a = r\tilde\Gamma^a \quad .
\end{equation}
\end{definition}
\noindent
In the basis (\ref{eq:spinor-basis}), the unphysical gamma matrices take the form
\begin{equation}
\label{eq:curved-gamma-matrices}
\tilde\Gamma^a = (\Gamma_+ + \alpha\Gamma_- + \beta_A\tilde\Lambda^A)\partial_r^a + \Gamma_-\partial_u^a + \tilde\Lambda^A\partial_A^a \quad ,
\end{equation}
where $\tilde\Lambda^A = \hat{\tilde{e}}^{\hat{A}A}\Gamma_{\hat{A}}$ so that
\begin{align}
\{\tilde\Lambda_A, \tilde\Lambda_B\}
& = \{\Gamma_{\hat{A}},\Gamma_{\hat{B}}\} \hat{\tilde{e}}^{\hat{A}}_A \hat{\tilde{e}}^{\hat{B}}_B \nonumber \\
& = 2\delta^{\hat{A}\hat{B}} I \hat{\tilde{e}}^{\hat{A}}_A \hat{\tilde{e}}^{\hat{B}}_B \nonumber \\
& = 2\gamma_{AB} I \quad .
\end{align}
Note that setting $r=0$ above gives
\begin{equation}
\{\tilde\Lambda_A^{(0)},\tilde\Lambda_B^{(0)}\} = s_{AB} I \quad .
\end{equation}
\noindent
We now choose a representation; that is, we make a particular choice of gamma matrices from the equivalence class of matrices satisfying the defining relation (\ref{eq:gamma-matrices-flat-definition}).
Such a representation is given by
\begin{align}
\label{eq:gamma-matrices-representation}
& \Gamma_+ = \sqrt{2}\left(
\begin{matrix}
0 & I_{1/2} \\
0 & 0
\end{matrix}
\right) \quad ,
&
& \Gamma_- = \sqrt{2}\left(
\begin{matrix}
0 & 0 \\
I_{1/2} & 0
\end{matrix}
\right) \quad ,
&
& \Gamma_{\hat{A}} = \left(
\begin{matrix}
\sigma_{\hat{A}} & 0 \\
0 & -\sigma_{\hat{A}}
\end{matrix}
\right) \quad ,
\end{align}
where $I_{1/2}$ denotes the $2^{d/2-1}\times2^{d/2-1}$ identity matrix, and the matrices $\sigma_{\hat{A}}$ satisfy
\begin{equation}
\{\sigma_{\hat{A}},\sigma_{\hat{B}}\} = 2\delta_{\hat{A}\hat{B}} I_{1/2} \quad .
\end{equation}
We also require the matrices $\sigma_{\hat{A}}$ to be Hermitian, that is,
\begin{equation}
\sigma^{\strut}_{\hat{A}} = \sigma_{\hat{A}}^\ast \quad ,
\end{equation}
where $\ast$ denotes the Hermitian conjugate (conjugate transpose).
If we define $\tilde\sigma^{\strut}_A = \sigma_{\hat{A}}\hat{\tilde{e}}_A^{\hat{A}}$, then we have
\begin{equation}
\{\tilde\sigma_A,\tilde\sigma_B\} = 2\gamma_{AB}I_{1/2} \quad ,
\end{equation}
and moreover we can write
\begin{equation}
\label{eq:gamma-matrices-explicit-curved}
\tilde\Lambda^A = \left(
\begin{matrix}
\tilde\sigma^A & 0 \\
0 & -\tilde\sigma^A\
\end{matrix}
\right) \quad .
\end{equation}

\begin{definition}
We define the ``projectors''
\begin{align}
& P_+ = (1/2)\Gamma_+\Gamma_- \quad , \\
& P_- = (1/2)\Gamma_-\Gamma_+ \quad ,
\end{align}
which in our representation take the form
\begin{align}
&
&P_+ = \left(
\begin{matrix}
I_{1/2} & 0 \\
0 & 0
\end{matrix}
\right) \quad ,
&
& P_- = \left(
\begin{matrix}
0 & 0 \\
0 & I_{1/2}
\end{matrix}
\right) \quad .
\end{align}
Note that $P_\pm^2=P_\pm$.
These projectors can be applied to a spinor $\psi$,
\begin{equation}
\Psi_+ \equiv P_+\psi \ , \quad
\Psi_- \equiv P_-\psi \quad .
\end{equation}
Since $P_+ + P_- = I$, it follows that $\psi = \Psi_+ + \Psi_-$.
It is helpful to introduce $2^{d/2-1}$ component spinors $\psi_\pm$ such that, in our representation,
\begin{align}
& \Psi_+ = \left(
\begin{matrix}
\psi_+ \\ 0
\end{matrix}
\right) \quad , &
& \Psi_- = \left(
\begin{matrix}
0 \\ \psi_-
\end{matrix}
\right) \quad .
\end{align}
When solving the Dirac equation, we split it into ``plus'' and ``minus'' components and solve for $\psi_+$ and $\psi_-$ separately.
To see how this decomposition works, each equation should be written out explicitly in matrix form using the representation (\ref{eq:gamma-matrices-representation}).
Then the upper and lower components of the resulting matrix equation can be considered separately:  they are the plus and minus components to which we refer.
\end{definition}

\begin{definition}
\label{def:dirac-conjugate}
The Dirac conjugate $\bar\psi$ is given by
\begin{equation}
\label{eq:dirac-conjugate}
\bar\psi = -\frac{1}{\sqrt{2}}\psi^\ast(\Gamma_+ - \Gamma_-) \equiv (-\psi_-, \psi_+) \quad .
\end{equation}
\end{definition}
\begin{remark}
This can be used to define a vector $\bar\psi\Gamma^a\psi$ from a spinor.
It can easily be shown that this quantity is real.
If we introduce an ``unphysical spinor'', $\tilde\psi=r^{1/2}\psi$, it follows that $\bar{\tilde\psi}=r^{1/2}\bar\psi$, and therefore
\begin{equation}
\bar{\tilde\psi}\tilde\Gamma^a\tilde\psi = \bar\psi\Gamma^a\psi \quad ,
\end{equation}
since $\tilde\Gamma^a=r^{-1}\Gamma^a$.
\end{remark}

\subsection{Spinor Connection}
\label{sec:spinors-connection}
In order to calculate with covariant derivatives acting on spinors, we must calculate the connection coefficients, which are analogous to the Christoffel symbols for tensors. They govern the action of a covariant derivative on a spinor $\psi$ according to the equation
\begin{equation}
\label{eq:spinor-derivative}
\nabla_a\psi = \partial_a\psi + \omega_a\psi \quad .
\end{equation}
To generalise the natural derivative operator such that $\nabla_cg_{ab}=0$ to act on spinors, we must choose connection coefficients $\omega_a$ such that
\begin{equation}
\label{eq:spinor-connection-definition}
0 = \nabla_a\Gamma_b \equiv \partial_a\Gamma_b - \Gamma^c{}_{ab}\Gamma_c + [\omega_a, \Gamma_b] \quad .
\end{equation}
A suitable choice of $\omega_a = \omega_a^{\mu\nu}[\Gamma_\mu,\Gamma_\nu]$ is given by
\begin{align}
\label{eq:spinor-connection-coefficients}
\omega^{\mu\nu}_a
& = \frac{1}{8} e^{\mu b} \nabla_a e^\nu_b
\iftoggle{complete}{\nonumber \\
& = \frac{1}{8} e^{\mu b} \partial_a e^\nu_b - \frac{1}{8} e^{\mu b} e^\nu_c \Gamma^c{}_{ab}
}{} \quad .
\end{align}
\iftoggle{complete}{[[TODO: DERIVE THIS FORMULA FROM EQUATION (\ref{eq:spinor-connection-definition})]]}{}
It can be verified that the connection coefficients (\ref{eq:spinor-connection-coefficients}) satisfy equation (\ref{eq:spinor-connection-definition}).
The connection coefficients for the unphysical covariant derivative $\tilde\nabla_c\tilde{g}_{ab}=0$ are defined in a similar manner.
The full expressions for the unphysical connection coefficients in Gaussian null coordinates are given in Appendix \ref{app:grspinors-spinor-connection}.

The physical and unphysical connection coefficients are related by \cite{Hollands:2005wt}
\begin{equation}
\label{eq:conformal-spinor-connection}
\omega_a = \tilde\omega_a + \frac{1}{2}r^{-1}(\tilde\nabla_a r - \tilde\Gamma_a\tilde\Gamma_b\tilde\nabla^b r) \quad .
\end{equation}
Thus for a spinor $\psi$,
\begin{equation}
\nabla_a\psi = \tilde\nabla_a\psi + \frac{1}{2}r^{-1}(\tilde\nabla_a r - \tilde\Gamma_a\tilde\Gamma_b\tilde\nabla^b r)\psi \quad ,
\end{equation}
and, introducing the ``unphysical'' spinor
\begin{equation}
\label{eq:unphysical-spinor}
\tilde\psi = r^{1/2}\psi \quad ,
\end{equation}
this gives the final relationship between physical and unphysical derivatives,
\begin{equation}
\label{eq:conformal-spinor-derivative}
\nabla_a\psi = r^{-1/2}[\tilde\nabla_a\tilde\psi - \frac{1}{2}r^{-1}\tilde\Gamma_a\tilde\Gamma_b(\tilde\nabla^b r)\tilde\psi] \quad .
\end{equation}
The calculations leading to equations (\ref{eq:conformal-spinor-connection}) and (\ref{eq:conformal-spinor-derivative}) can be found in Note \ref{note:conformal-spinor-connection}.

Now we have defined the spinor derivative, we can calculate the covariant derivatives of not only spinors, but also the gamma matrices.
The derivatives of the gamma matrices $\Gamma_\pm$ and $\tilde\Lambda^A$ are given in Appendix \ref{app:grspinors-gamma-matrices-derivatives} and will be used when solving the Dirac equation.

\begin{definition}
The action of $D_A$ on the projected spinors is given by
\begin{equation}
D_A\tilde\psi_\pm = \partial_A\tilde\psi_\pm + \tilde\Omega_A\tilde\psi_\pm \quad ,
\end{equation}
where the spin connection coefficients $\tilde\Omega_A$ are given by
\begin{equation}
\tilde\Omega_A = \frac{1}{8}\hat{\tilde{e}}^{\hat{A}B}\tilde{D}_A\hat{\tilde{e}}^{\hat{B}}_B [\sigma_{\hat{A}}, \sigma_{\hat{B}}] \quad .
\end{equation}
Note that $D_A\tilde\psi_\pm = P_\pm D_A\tilde\psi$.
\end{definition}
\noindent
When solving the Dirac equation it is necessary to exchange $\tilde\nabla_r$ and $\tilde\nabla_A$ derivatives, using the general formula for commutation of spinor derivatives,
\begin{equation}
\label{eq:spinor-derivative-commutation}
\nabla_{[a}\nabla_{b]}\tilde\psi = \frac{1}{16}R_{abcd}[\Gamma^c,\Gamma^d]\tilde\psi \quad .
\end{equation}
Explicit expressions in coordinates are given in Appendix \ref{app:grspinors-derivative-commutation}.

\subsection{Expansions of Spinors}
We make asymptotic expansions of spinors analogous to those we made of tensors.
The main difference is that instead of requiring the expansion coefficients to be independent of $r$, we require that they are covariantly constant in the $r$ direction.
This turns out to make spinorial calculations more simple.
For simpler notation, we denote the covariant derivative in the $r$ direction by
\begin{equation}
\tilde\nabla_r \equiv l^a\tilde\nabla_a \quad .
\end{equation}
A spinor $\tilde\psi$ is expanded as
\begin{equation}
\tilde\psi = \sum_{k=0}^\infty r^k \tilde\psi^{(k)} \quad ,
\end{equation}
where the expansion coefficients are defined on $\scri^+$ by
\begin{equation}
\tilde\psi^{(k)} = \left[ \tilde\nabla_r^k \tilde\psi \right]_{r=0} \quad ,
\end{equation}
where $\tilde\nabla_r^k$ denotes taking $k$ derivatives.
These expansion coefficients are extended to $\tilde{M}$ by requiring
\begin{equation}
\label{eq:spinor-expansion-coefficients-constancy-condition}
l^a\tilde\nabla_a \tilde\psi^{(k)} \equiv \tilde\nabla_r \tilde\psi^{(k)} = 0 \quad .
\end{equation}
That is, they are parallel transported along the geodesics of $l^a$.
It is important to note that the expansion coefficients $\psi^{(k)}$ may in general depend on $r$.

The gamma matrix $\tilde\Lambda^A$ is expanded in a similar manner.
In particular the leading term $\tilde\Lambda^{A(0)}$ is defined on $\scri^+$ by
\begin{equation}
\tilde\Lambda^{A(0)} = \tilde\Lambda^A\big|_{r=0} \quad .
\end{equation}
\noindent
When expanding tensors, we used the notion of ``order'' to allow us to disregard terms containing very high powers of $r$.
Since the spinorial expansion coefficients defined above can in general depend on $r$, we cannot use the same definition of ``order'' in the spinor case.
The expansion coefficients are covariantly constant, so we make an alternative definition using covariant $r$ derivatives.
\begin{definition}
\label{def:order-spinor}
We say a smooth spinor $\psi$ is \textit{of order $r^k$} for some integer $k$ if
\begin{equation}
\left[ \nabla_r^j \psi \right]_{r=0} = 0 \ , \quad \text{for all} \ 0\leq j<k .
\end{equation}
\end{definition}
\begin{remark}
This definition has the same properties as mentioned above for the tensor definition and could also be used for tensors.
Conversely, using the tensor definition on our spinor expansions would lead to confusion since the spinor expansion coefficients have $r$ dependence.
Indeed, if we used the tensor definition, ``hidden'' powers of $r$ in spinor expansion coefficients could easily be forgotten.
However, the spinor expansion coefficients are covariantly constant in the $r$ direction by definition.
We can thus avoid this problem by using covariant derivatives instead (as in the definition of order given above for spinors).
\end{remark}

%% file: DiracEquation.tex
\section{Dirac Equation}
\label{sec:dirac-spinor}

An asymptotic expansion is made of a spinor satisfying the Dirac equation,
\begin{equation}
\label{eq:dirac-equation}
0 = \Gamma^a\nabla_a\psi \quad .
\end{equation}
\noindent
An expression is obtained for a covariantly constant spinor in Minkowski spacetime, which is used to motivate boundary conditions on the solution of the Dirac equation.
The Dirac equation on a spacelike hypersurface is decomposed into a ``timelike'' equation and the ``spatial'' Witten equation.
The desired asymptotic expansion can then be made in analogy to the metric expansion in Section \ref{sec:einstein-equations-solution}.

\subsection{Minkowski Spinor}
We want to find the asymptotic expansion of a covariantly constant spinor $\psi_0$,
\begin{equation}
\label{eq:constant-spinor-condition}
0 = \nabla_a\psi_0 \quad ,
\end{equation}
in Minkowski spacetime such that the vector $\bar\psi_0\Gamma^a\psi_0$ satisfies
\begin{equation}
\bar\psi_0\Gamma^a\psi_0 = (\partial_t)^a
\end{equation}
on $\scri^+$.
This condition means that the vector $\bar\psi_0\Gamma^a\psi_0$ is a time translation symmetry at $\scri^+$.
Since the spinor is covariantly constant,
\begin{equation}
\nabla_a(\bar\psi_0\Gamma^b\psi_0)=0 \quad ,
\end{equation}
the vector must equal $(\partial_t)^a$ throughout $\tilde{M}$.
By the Minkowski coordinate transformation (\ref{eq:minkowski-coordinate-transformation}),
\begin{equation}
(\partial_t)^a \equiv (\partial_u)^a \quad ,
\end{equation}
so that the condition becomes
\begin{equation}
\label{eq:minkowski-spinor-normalisation-overall}
\bar\psi_0\Gamma^a\psi_0 = (\partial_u)^a \quad .
\end{equation}

\begin{theorem}
\label{thm:minkowski-spinor}
A covariantly constant spinor $\psi_0$ in Minkowski space, satisfying condition (\ref{eq:minkowski-spinor-normalisation-overall}), takes the form
\begin{equation}
\label{eq:minkowski-spinor}
\tilde\psi_0 = \left(
\begin{matrix}
\tilde\psi_{0+}^{(0)} \\
r\tilde\psi_{0-}^{(1)}
\end{matrix}
\right) \quad ,
\end{equation}
where $\tilde\psi_{0+}^{(0)}$ and $\tilde\psi_{0-}^{(1)}$ are $2^{d/2-1}$ spinors independent of $r,u$ such that
\begin{subequations}
\label{eq:minkowski-spinor-relation}
\begin{align}
\label{eq:minkowski-spinor-relation-+}
& \mathscr{D}_A\tilde\psi_{0+}^{(0)} = \frac{1}{\sqrt{2}}\tilde\sigma_A^{(0)}\tilde\psi_{0-}^{(1)} \quad , \\
\label{eq:minkowski-spinor-relation--}
& \mathscr{D}_A\tilde\psi_{0-}^{(1)} = -\frac{\lambda}{2\sqrt{2}}\tilde\sigma_A^{(0)}\tilde\psi_{0+}^{(0)} \quad .
\end{align}
\end{subequations}
In addition, they satisfy the normalisation properties
\begin{subequations}
\label{eq:minkowski-spinor-normalisation}
\begin{align}
\label{eq:minkowski-spinor-normalisation-+}
& \tilde\psi_{0+}^{(0)\ast}\tilde\psi_{0+}^{(0)} = \frac{1}{\sqrt2} \quad , \\
\label{eq:minkowski-spinor-normalisation--}
& \tilde\psi_{0-}^{(1)\ast}\tilde\psi_{0-}^{(1)} = \frac{\lambda}{2\sqrt2} \quad , \\
\label{eq:minkowski-spinor-normalisation-A}
& \Re\left(\tilde\psi_{0+}^{(0)\ast}\tilde\sigma^{A(0)}\tilde\psi_{0-}^{(1)}\right) = 0 \quad ,
\end{align}
\end{subequations}
where $\Re$ denotes taking the real part.
Note that the physical Minkowski spinor is given by $\psi_0=r^{-1/2}\tilde\psi_0$.
\end{theorem}

\begin{proof}
First of all, we note that by Theorem \ref{thm:minkowski-metric}, the metric components in Minkowski space are given by
\begin{subequations}
\label{eq:dirac-metric-components}
\begin{align}
\alpha &= \frac{\lambda}{2}r^2  \quad ,\\
\beta_A &= 0 \quad , \\
\gamma_{AB} &= s_{AB} \quad .
\end{align}
\end{subequations}
Since $\gamma_{AB}$ is independent of $r,u$ it follows that $\tilde\Lambda^A\equiv\tilde\Lambda^{A(0)}$ is too, that is,
\begin{equation}
\partial_r\tilde\Lambda^A = \partial_u\tilde\Lambda^A = 0 \quad .
\end{equation}
\noindent
We rewrite condition (\ref{eq:constant-spinor-condition}) in terms of the ``unphysical'' spinor,
\begin{equation}
\tilde\psi_0 = r^{1/2}\psi_0 \quad ,
\end{equation}
and from equation (\ref{eq:conformal-spinor-derivative}) we get
\begin{equation}
0 = \nabla_a\psi_0 = r^{-1/2}\left[\tilde\nabla_a\tilde\psi_0 - \frac{1}{2}r^{-1}\tilde\Gamma_a\tilde\Gamma_b(\tilde\nabla^b r)\tilde\psi_0\right] \quad .
\end{equation}
Multiplying by $r^{1/2}$ and splitting this into components, we get
\begin{subequations}
\label{eq:constant-spinor-equations}
\begin{align}
\label{eq:constant-spinor-equation-r}
0
\iftoggle{complete}{
& = \tilde\nabla_r\tilde\psi_0 - \frac{1}{2}r^{-1}\tilde\Gamma_r\tilde\Gamma^r\tilde\psi_0 \nonumber \\
}{}
& = \tilde\nabla_r\tilde\psi_0 - r^{-1}P_-\tilde\psi_0 \quad , \\
\label{eq:constant-spinor-equation-u}
0
\iftoggle{complete}{
& = \tilde\nabla_u\tilde\psi_0 - \frac{1}{2}r^{-1}\tilde\Gamma_u\tilde\Gamma^r\tilde\psi_0 \nonumber \\
}{}
& = \tilde\nabla_u\tilde\psi_0 - \frac{\lambda}{2} r (P_+ - P_-) \tilde\psi_0 \quad , \\
\label{eq:constant-spinor-equation-A}
0
\iftoggle{complete}{
& = \tilde\nabla_A\tilde\psi_0 - \frac{1}{2}r^{-1}\tilde\Gamma_A\tilde\Gamma^r\tilde\psi_0 \nonumber \\
}{}
& = \tilde\nabla_A\tilde\psi_0 - \frac{1}{2}r^{-1}\tilde\Lambda_A\Gamma_+\tilde\psi_0 - \frac{\lambda}{4}r\tilde\Lambda_A\Gamma_-\tilde\psi_0 \quad .
\end{align}
\end{subequations}
\noindent
Substituting the expressions (\ref{eq:dirac-metric-components}) for the metric components into the spinor connection coefficients (\ref{eq:grspinors-spinor-connection-coefficients}), we find
\begin{subequations}
\begin{align}
\tilde\omega_r &= 0 \quad , \\
\tilde\omega_u &= \frac{\lambda}{2} r(P_+ - P_-) \quad , \\
\tilde\omega_A &= \tilde\Omega_A^{(0)} \quad .
\end{align}
\end{subequations}
Note that since $\tilde\omega_r=0$,
\begin{equation}
0 = \tilde\nabla_r\tilde\psi_0^{(k)} = \partial_r\tilde\psi_0^{(k)} \quad ,
\end{equation}
and so in the case of Minkowski space, the expansion coefficients $\tilde\psi_0^{(k)}$ are independent of $r$ as well as covariantly constant (which is not true for a general spinor expansion in curved space).
\medskip\noindent
Next we apply equations (\ref{eq:constant-spinor-equations}) to obtain the expansion of the spinor $\tilde\psi_0$.
Equation (\ref{eq:constant-spinor-equation-r}) gives
\begin{equation}
0 = \partial_r\tilde\psi_0 - r^{-1}P_-\tilde\psi_0 \quad .
\end{equation}
Writing this explicitly as a matrix equation in our representation, we have
\begin{equation}
0 = \left(
\begin{matrix}
\partial_r\tilde\psi_{0+} \\
\partial_r\tilde\psi_{0-} - r^{-1}\tilde\psi_{0-}
\end{matrix}
\right) \quad .
\end{equation}
The ``plus'' component gives $\tilde\psi_{0+} = \tilde\psi_{0+}^{(0)}$.
The ``minus'' component gives $\tilde\psi_{0-} = r\tilde\psi_{0-}^{(1)}$.
Therefore the spinor takes the desired form (\ref{eq:minkowski-spinor}).
\noindent
Next we apply equation (\ref{eq:constant-spinor-equation-u}), and find that
\begin{equation}
0 = \partial_u\tilde\psi_0 \quad ,
\end{equation}
so that $\tilde\psi_{0+}^{(0)}$ and $\tilde\psi_{0-}^{(1)}$ are independent of $u$.
Equation (\ref{eq:constant-spinor-equation-A}) gives
\begin{equation}
0 = \mathscr{D}_A\tilde\psi - \frac{1}{2}r^{-1}\tilde\Lambda_A^{(0)}\Gamma_+\tilde\psi - \frac{\lambda}{4}r\tilde\Lambda_A^{(0)}\Gamma_-\tilde\psi \quad ,
\end{equation}
which in matrix form is given by
\begin{equation}
0 = \left (
\begin{matrix}
\mathscr{D}_A\tilde\psi_{0+}^{(0)} - (1/\sqrt2)\tilde\sigma_A^{(0)}\tilde\psi_{-0}^{(1)} \\
r\mathscr{D}_A\tilde\psi_{0-}^{(1)} + r(\lambda/2\sqrt2)\tilde\sigma_A^{(0)}\tilde\psi_{0+}^{(0)}
\end{matrix}
\right) \quad .
\end{equation}
The ``plus'' and ``minus'' components give equations (\ref{eq:minkowski-spinor-relation}).
\noindent
Finally, we impose the normalisation condition (\ref{eq:minkowski-spinor-normalisation-overall}).
Splitting it into components,
\begin{subequations}
\begin{align}
& \bar{\tilde\psi}_0\tilde\Gamma^r\tilde\psi_0 = 0 \quad , \\
& \bar{\tilde\psi}_0\tilde\Gamma^u\tilde\psi_0 = 1 \quad , \\
& \bar{\tilde\psi}_0\tilde\Gamma^A\tilde\psi_0 = 0 \quad .
\end{align}
\end{subequations}
Note that these are scalar equations, not spinor equations.
To evaluate their left hand sides, they can be explicitly written out in matrix notation.
The $u$ component is
\begin{equation}
1 = \bar{\tilde\psi}_0\Gamma_-\tilde\psi_0 = \sqrt{2}\tilde\psi_{0+}^{(0)\ast}\tilde\psi_{0+}^{(0)} \quad ,
\end{equation}
which gives the first normalisation condition, equation (\ref{eq:minkowski-spinor-normalisation-+}).
Then from the $r$ component we find
\begin{equation}
0 = \bar{\tilde{\psi_0}}(\Gamma_+ + \frac{\lambda}{2}r^2\Gamma_-)\tilde\psi_0 \quad ,
\end{equation}
and thus,
\begin{equation}
0 = r^2\tilde\psi_{0-}^{(1)\ast}\tilde\psi_{0-}^{(1)} - \frac{\lambda}{2}r^2\tilde\psi_{0+}^{(0)\ast}\tilde\psi_{0+}^{(0)} \quad ,
\end{equation}
which gives the second normalisation condition (\ref{eq:minkowski-spinor-normalisation--}).
Finally, from the $A$ component, we find
\begin{align}
0 = \bar{\tilde\psi}_0\tilde\Lambda^A\tilde\psi_0
& = -r\tilde\psi_{0-}^{(1)\ast}\tilde\sigma^{A(0)}\tilde\psi_{0+}^{(0)} - r\tilde\psi_{0+}^{(0)\ast}\tilde\sigma^{A(0)}\tilde\psi_{0-}^{(1)} \nonumber \\
& = -2r\Re(\tilde\psi_{0+}^{(0)\ast}\tilde\sigma^A\tilde\psi_{0-}^{(1)}) \quad ,
\end{align}
which is the final normalisation condition (\ref{eq:minkowski-spinor-normalisation-A}).
\end{proof}

\subsection{Writing down the Dirac equation}
Next we write down the Dirac equation (\ref{eq:dirac-equation}) on a spinor $\psi$
in Gaussian null coordinates using the curved space gamma matrices introduced in Section \ref{sec:spinors}.

Suppose there is a spacelike, asymptotically null hypersurface $\mathscr{S}$, given asymptotically by $\{u=r/2\}$.
On $\mathscr{S}$, the induced Riemannian metric $h_{ab}$ is given by the standard formula for a $(d-1)+1$ decomposition,
\begin{equation}
h_{ab} = -cN_aN_b + g_{ab} \quad ,
\end{equation}
where $N^a$ is the normal to $\mathscr{S}$ and $c = (g_{ab}N^aN^b)^{-1} < 0$ is a normalisation constant.
Suppose also that the spinor $\psi$ satisfies \cite{Horowitz82}
\begin{equation}
\label{eq:dirac-equation-time}
N^a\nabla_a\psi = 0 \quad .
\end{equation}
On $\mathscr{S}$, the Dirac equation (\ref{eq:dirac-equation}) combined with equation (\ref{eq:dirac-equation-time}) is equivalent \cite{Horowitz82} to the the ``Witten equation'' (used in the positivity proof for the ADM mass in \cite{Witten81}), given in general by
\begin{align}
\label{eq:dirac-equation-space}
0
& = g^{ab}\Gamma_a\nabla_b\psi \nonumber \\
& = cN^aN^b\Gamma_a\nabla_b\psi + h^{ab}\Gamma_a\nabla_b\psi \nonumber \\
& = h^{ab}\Gamma_a\nabla_b\psi
\end{align}
on $\mathscr{S}$, where in the last line we used equation (\ref{eq:dirac-equation-time}).
Intuitively, this is a decomposition of the Dirac equation (\ref{eq:dirac-equation}) into a ``timelike'' equation (\ref{eq:dirac-equation-time}) and a ``spatial'' equation (\ref{eq:dirac-equation-space}).

In order to expand a spinor solution to this equation, however, we must first write it in coordinates.

\begin{theorem}
In Gaussian null coordinates and the spinor basis (\ref{eq:spinor-basis}), the Witten equation on $\mathscr{S}$ takes the form
\begin{align}
\label{eq:witten-equation}
0
 =\ &r\left\{\Gamma_+ + \left(2 - \alpha - \beta_A\beta^A\right)\Gamma_- + \beta_A\tilde\Lambda^A\right\}\tilde\nabla_r\tilde\psi + r\left(-\beta_A\Gamma_- + \tilde\Lambda^A\right)\tilde\nabla_A\tilde\psi \nonumber \\
& - \left\{\frac{d}{2}\Gamma_+ + \frac{d-2}{2}\alpha\Gamma_- + \beta_A\tilde\Lambda^A\right\}\tilde\psi \quad .
\end{align}
\end{theorem}
\begin{proof}
The first step is to rewrite the Dirac equation (\ref{eq:dirac-equation}) in terms of the unphysical Gamma matrices and derivative operator.
Since $\Gamma^a = r\tilde\Gamma^a$, it becomes, by equation (\ref{eq:conformal-spinor-derivative}),
\begin{equation}
0 = r^{1/2}\left[\tilde\Gamma^a\tilde\nabla_a\tilde\psi - \frac{d}{2}r^{-1}\tilde\Gamma^b(\tilde\nabla_br)\tilde\psi\right] \quad ,
\end{equation}
where we used that
\begin{align}
\tilde\Gamma_a\tilde\Gamma^a
\iftoggle{complete}{
& = \tilde{g}_{ab}\tilde\Gamma^a\tilde\Gamma^b \nonumber \\
& = \tilde{g}_{ab}\tilde{g}^{ab} \nonumber \\
}{}
& = d \quad .
\end{align}
Summing over contracted indices and substituting the components of $\tilde\Gamma^a$, equation (\ref{eq:curved-gamma-matrices}), we find
\begin{equation}
\label{eq:dirac-equation-long}
0 = (\Gamma_+ + \alpha\Gamma_- + \beta_A\tilde\Lambda^A)\tilde\nabla_r\tilde\psi + \Gamma_-\tilde\nabla_u\tilde\psi + \tilde\Lambda^A\tilde\nabla_A\tilde\psi - \frac{d}{2}r^{-1}(\Gamma_+ + \alpha\Gamma_- + \beta_A\tilde\Lambda^A)\tilde\psi \quad .
\end{equation}
Next we use equation (\ref{eq:dirac-equation-time}) to eliminate the $\tilde\nabla_u\tilde\psi$ term.
Since $\mathscr{S}$ is a surface of constant
\begin{equation}
f \equiv r/2-u
\end{equation}
near $\scri^+$, its normal vector $N^a$ is given asymptotically by
\begin{equation}
\label{eq:asymptotic-surface-normal-halfway}
N^a = g^{ab}\tilde\nabla_b f = r^2\tilde{g}^{ab}\left[\frac{1}{2}dr_b - du_b\right] \quad .
\end{equation}
Define the ``unphysical normal'' by
\begin{equation}
\tilde{N}^a = r^{-2}N^a
\end{equation}
so that, substituting in the metric components, it becomes
\begin{equation}
\label{eq:asymptotic-surface-normal}
\tilde{N}^a = \left(-1 + \alpha+\frac{1}{2}\beta_A\beta^A\right)(\partial_r)^a + \frac{1}{2}(\partial_u)^a + \frac{1}{2}\beta^A(\partial_A)^a \quad .
\end{equation}
Rewriting the ``timelike'' equation (\ref{eq:dirac-equation-time}) in terms of the unphysical spinor and derivative operator using equation (\ref{eq:conformal-spinor-derivative}), we have
\begin{align}
\label{eq:dirac-equation-time-unphysical}
0
& = r^{3/2}\left[\tilde{N}^a\tilde\nabla_a\tilde\psi - \frac{1}{2}r^{-1}\tilde{N}^a\tilde\Gamma_a\tilde\Gamma_b\left(\tilde\nabla^br\right)\tilde\psi\right] \quad . 
\end{align}
Substituting the components of $N^a$ and $\tilde\Gamma^a$, we find
\begin{align}
0
& = (-1+\alpha+\frac{1}{2}\beta_A\beta^A)\tilde\nabla_r\tilde\psi + \frac{1}{2}\tilde\nabla_u\tilde\psi + \frac{1}{2}\beta^A\tilde\nabla_A\tilde\psi \nonumber \\
& - \frac{1}{2}r^{-1}\left(\alpha+\frac{1}{2}\beta_A\beta^A\right)\tilde\psi + \frac{1}{2}r^{-1}(\Gamma_-\Gamma_+ + \beta_A\Gamma_-\tilde\Lambda^A)\tilde\psi \quad ,
\end{align}
and finally, left multiplying by $\Gamma_-$ and using that $\Gamma_-^{\ 2}=0$, we find
\begin{equation}
\label{eq:spinor-time-derivative}
\Gamma_-\tilde\nabla_u\tilde\psi = (2-2\alpha-\beta_A\beta^A)\tilde\nabla_r\tilde\psi - \beta^A\Gamma_-\tilde\nabla_A\tilde\psi + \left(\alpha+\frac{1}{2}\beta_A\beta^A\right)\Gamma_-\tilde\psi \quad .
\end{equation}
We can now substitute this back into the Dirac equation (\ref{eq:dirac-equation-long}) to eliminate the $u$ derivative.
The spatial equation then takes the form
\begin{align}
0
& = \left\{\Gamma_+ + \left(2 - \alpha - \beta_A\beta^A\right)\Gamma_- + \beta_A\tilde\Lambda^A\right\}\tilde\nabla_r\tilde\psi + \left(-\beta_A\Gamma_- + \tilde\Lambda^A\right)\tilde\nabla_A\tilde\psi \nonumber \\
& - r^{-1}\left[\frac{d}{2}\Gamma_+ + \frac{d-2}{2}\alpha\Gamma_- + \beta_A\tilde\Lambda^A\right]\tilde\psi \quad .
\end{align}
Multiplying by $r$ (to remove the negative powers of $r$ and make later calculations easier), we get the final expression (\ref{eq:witten-equation}) for the Witten equation on $\mathscr{S}$, as required.
\end{proof}

\subsection{Solving the Witten Equation in four dimensions}
\label{sec:dirac-spinor-solution}
In this section, we make an asymptotic expansion of a spinor $\psi$ satisfying equation (\ref{eq:witten-equation}) on $\mathscr{S}$, in 4 dimensions.
The main result is given in the following theorem.
We only obtain as much information from the Dirac equation as is needed to demonstrate the positivity of the Bondi mass in Section \ref{sec:bondi-positivity-4d}.

\begin{theorem}
\label{thm:witten-solution-4d}
Suppose there exists a spinor $\psi$ satisfying the Dirac equation (\ref{eq:dirac-equation}) such that the unphysical spinor $\tilde\psi=r^{1/2}\psi$ is smooth at $\scri^+$.
In addition, suppose that the spinor satisfies the following boundary conditions on $\Sigma(0,0)$ (the intersection of $\scri^+$ and $\mathscr{S}$),
\begin{align}
\label{eq:dirac-bcs}
& \tilde\psi - \tilde\psi_0 = 0 \quad , &
& \nabla_r(\tilde\psi - \tilde\psi_0) = 0 \quad ,
\end{align}
where $\tilde\psi_0$ is the unphysical Minkowski spinor given in Theorem \ref{thm:minkowski-spinor}.
Then the spinor $\psi$ takes the form
\begin{equation}
\psi = r^{-1/2}\tilde\psi = \left(
\begin{matrix}
\tilde\psi_+ \\ \tilde\psi_-
\end{matrix}
\right) \quad ,
\end{equation}
where the expansions of the ``plus'' and ``minus'' components are given by
\begin{subequations}
\label{eq:half-spinor-expansions}
\begin{align}
& \tilde\psi_+ = \tilde\psi_+^{(0)} + r^3\tilde\psi_+^{(3)} + O(r^4) \quad , \\
& \tilde\psi_- = r\tilde\psi_-^{(1)} + r^2\tilde\psi_-^{(2)} + O(r^3) \quad ,
\end{align}
\end{subequations}
where $\tilde\psi_+^{(0)}$ and $r\tilde\psi_-^{(1)}$ are the plus and minus components of the Minkowski spinor, and the spinors $\tilde\psi_+^{(3)}$ and $\tilde\psi_-^{(2)}$ satisfy
\begin{equation}
\label{eq:psi+-3-condition}
\Re\left[12\sqrt{2}\tilde\psi^{(0)}_+{}^\ast\tilde\psi^{(3)}_+ - 2\tilde\psi^{(0)}_+{}^\ast\tilde\sigma^{A(0)}\mathscr{D}_A\tilde\psi^{(2)}_- - 2\alpha^{(3)} + \frac{1}{2}\mathscr{D}^A\beta_A^{(2)} + \frac{1}{4}\gamma^{AB(1)}\partial_u\gamma_{AB}^{(1)}\right] = 0 \quad .
\end{equation}
\end{theorem}
\noindent
The following lemma gives a standard formula for commuting covariant $r$ and $A$ derivatives on spinors, which will be required when making an expansion of the Witten spinor.
\begin{lemma}
\label{thm:spinor-derivative-commutation}
For $n\geq1$,
\begin{equation}
\label{eq:spinor-rA-commutation}
\tilde\nabla_r^n\tilde\nabla_A\tilde\psi = \tilde\nabla_A\tilde\nabla_r^n\tilde\psi + \sum_{k=1}^n \frac{n!}{k!(n-k)!}(\tilde\nabla_r^{k-1}\tilde{Q}_A)\tilde\nabla_r^{n-k}\tilde\psi \quad ,
\end{equation}
where for brevity we have defined
\begin{equation}
\tilde{Q}_A = \frac{1}{8}\tilde{R}_{rAcd}[\tilde\Gamma^c,\tilde\Gamma^d] \quad .
\end{equation}
Expressions for $\tilde{Q}_A$ and its derivatives can be found in Appendix \ref{app:grspinors-derivative-commutation}.
Note that this Lemma is true in arbitrary dimension.
\end{lemma}
\begin{proof}
By equation (\ref{eq:spinor-derivative-commutation}), we have
\begin{align}
\label{eq:spinor-rA-commutation-1}
\tilde\nabla_r\tilde\nabla_A\tilde\psi
& = \tilde\nabla_A\tilde\nabla_r\tilde\psi + \frac{1}{8}\tilde{R}_{rAcd}[\tilde\Gamma^c,\tilde\Gamma^d]\tilde\psi \nonumber \\
& = \tilde\nabla_A\tilde\nabla_r\tilde\psi + \tilde{Q}_A\tilde\psi \quad .
\end{align}
Taking a second derivative, we find
\begin{align}
\tilde\nabla_r^2\tilde\nabla_A\tilde\psi 
& = \tilde\nabla_r(\tilde\nabla_A\tilde\nabla_r\tilde\psi + \tilde{Q}_A\tilde\psi) \nonumber \\
& = \tilde\nabla_A\tilde\nabla_r^2\tilde\psi + 2\tilde{Q}_A\tilde\nabla_r\tilde\psi + (\tilde\nabla_r\tilde{Q}_A)\tilde\psi \quad .
\end{align}
We have thus verified equation (\ref{eq:spinor-rA-commutation}) for $n=1,2$.
It can be proved for all $n\geq1$ by induction.
Indeed, if we assume it is true for some $m\geq1$, we find that
\begin{align}
\tilde\nabla_r^{m+1}\tilde\nabla_A\tilde\psi
& = \tilde\nabla_r\left(\tilde\nabla_A\tilde\nabla_r^m\tilde\psi + \sum_{k=1}^m \frac{m!}{k!(m-k)!}(\tilde\nabla_r^{k-1}\tilde{Q}_A)\tilde\nabla_r^{m-k}\tilde\psi\right) \quad ,
\end{align}
where we have applied (\ref{eq:spinor-rA-commutation}) for $n=m$.
Now equation (\ref{eq:spinor-rA-commutation-1}) is also true if we replace $\tilde\psi$ with $\tilde\nabla_r^m\tilde\psi$.
Thus,
\begin{align}
\tilde\nabla_r^{m+1}\tilde\nabla_A\tilde\psi
& = \tilde\nabla_A\tilde\nabla_r^{n+1}\tilde\psi + \tilde{Q}_A\tilde\nabla_r^n\tilde\psi \nonumber \\
& + \sum_{k=1}^n\frac{n!}{k!(n-k)!}(\tilde\nabla_r^k\tilde{Q}_A)\tilde\nabla_r^{n-k}\tilde\psi + \sum_{k=1}^n\frac{n!}{k!(n-k)!}(\tilde\nabla_r^{k-1}\tilde{Q}_A)\tilde\nabla_r^{n+1-k}\tilde\psi \nonumber \\
\iftoggle{complete}{
& = \tilde\nabla_A\tilde\nabla_r^{n+1}\tilde\psi + \tilde{Q}_A\tilde\nabla_r^n\tilde\psi \nonumber \\
& + \sum_{k=2}^{n+1}\frac{n!}{(k-1)!(n+1-k)!}(\tilde\nabla_r^{k-1}\tilde{Q}_A)\tilde\nabla_r^{n+1-k}\tilde\psi \nonumber \\
& + \sum_{k=1}^n\frac{n!}{k!(n-k)!}(\tilde\nabla_r^{k-1}\tilde{Q}_A)\tilde\nabla_r^{n+1-k}\tilde\psi \nonumber \\
}{}
& = \tilde\nabla_A\tilde\nabla_r^{n+1}\tilde\psi + \sum_{k=1}^{n+1}\frac{(n+1)!}{k!(n+1-k)!}(\tilde\nabla_r^{k-1}\tilde{Q}_A)\tilde\nabla_r^{n+1-k}\tilde\psi \quad ,
\end{align}
which is precisely equation (\ref{eq:spinor-rA-commutation}) for $n=m+1$.
This completes the proof by induction.
\end{proof}

\begin{proof}[Proof of Theorem \ref{thm:witten-solution-4d}]
The boundary conditions (\ref{eq:dirac-bcs}) fix the two lowest order expansion coefficients $\tilde\psi^{(0)}$ and $\tilde\psi^{(1)}$ to be the same as those of the Minkowski spinor (\ref{eq:minkowski-spinor}) at $\scri^+$.
Since the expansion coefficients are not necessarily independent of $r$, we cannot directly consider the Witten equation (\ref{eq:witten-equation}) at each order $r^k$.
However, the expansion coefficients are covariantly constant, so instead we will take successive covariant derivatives of the Witten equation and set $r=0$ for each one.
This procedure is equivalent to the analysis at each order performed for the Einstein equations.
\medskip\noindent
The Witten equation and its first three covariant $r$ derivatives are given by,
\begin{subequations}
\label{eq:witten-equation-4d}
\begin{align}
\label{eq:witten-equation-4d-0}
& 0 = r\{\Gamma_+ + (2-\alpha)\Gamma_- + \beta_A\tilde\Lambda^A\}\tilde\nabla_r\tilde\psi + r(-\beta_A\Gamma_- + \tilde\Lambda^A)\tilde\nabla_A\tilde\psi \nonumber \\
& + (-2\Gamma_+ - \alpha\Gamma_- - 2\beta_A\tilde\Lambda^A)\tilde\psi + O(r^4) \quad ,
\end{align}
\begin{align}
\label{eq:witten-equation-4d-1}
& 0 = r\{\Gamma_+ + 2\Gamma_-\}\tilde\nabla_r^2\tilde\psi + r\{\tilde\nabla_r\Gamma_+ - (\partial_r\alpha)\Gamma_- + (\partial_r\beta_A)\tilde\Lambda^A\}\tilde\nabla_r\tilde\psi \nonumber \\
& + \{-\Gamma_+ + (2-2\alpha)\Gamma_- - \beta_A\tilde\Lambda^A)\tilde\nabla_r\tilde\psi \nonumber \\
& + (-2\tilde\nabla_r\Gamma_+ - (\partial_r\alpha)\Gamma_- - 2(\partial_r\beta_A)\tilde\Lambda^A - 2\beta_A\tilde\nabla_r\tilde\Lambda^A)\tilde\psi \nonumber \\
& + r\tilde\Lambda^A\tilde\nabla_r\tilde\nabla_A\tilde\psi + r(-\partial_r\beta^A \Gamma_- + \tilde\nabla_r\tilde\Lambda^A)\tilde\nabla_A\tilde\psi + (-\beta^A\Gamma_- + \tilde\Lambda^A)\tilde\nabla_A\tilde\psi + O(r^3) \quad ,
\end{align}
\begin{align}
\label{eq:witten-equation-4d-2}
& 0 = r(\Gamma_+ + 2\Gamma_-)\tilde\nabla_r^3\tilde\psi + r\{\tilde\nabla_r^2\Gamma_+ - (\partial_r^2\alpha)\Gamma_- + (\partial_r^2\beta_A)\tilde\Lambda^A\}\tilde\nabla_r\tilde\psi \nonumber \\
& + 4\Gamma_-\tilde\nabla_r^2\tilde\psi + \{-2\tilde\nabla_r\Gamma_+ - 4(\partial_r\alpha)\Gamma_- - 2(\partial_r\beta_A)\tilde\Lambda^A\}\tilde\nabla_r\tilde\psi \nonumber \\
& + \{-2\tilde\nabla_r^2\Gamma_+ - (\partial_r^2\alpha)\Gamma_- - 2(\partial_r^2\beta_A)\tilde\Lambda^A - 4(\partial_r\beta_A)\tilde\nabla_r\tilde\Lambda^A)\tilde\psi \nonumber \\
& + r\tilde\Lambda^A\tilde\nabla_r^2\tilde\nabla_A\tilde\psi + 2r(\tilde\nabla_r\tilde\Lambda^A)\tilde\nabla_r\tilde\nabla_A\tilde\psi + 2\tilde\Lambda^A\tilde\nabla_r\tilde\nabla_A\tilde\psi \nonumber \\
& + r(-\partial_r^2\beta_A \Gamma_- + \tilde\nabla_r^2\tilde\Lambda^A)\tilde\nabla_A\tilde\psi + 2(-\partial_r\beta^A \Gamma_- + \tilde\nabla_r\tilde\Lambda^A)\tilde\nabla_A\tilde\psi + O(r^2) \quad ,
\end{align}
\begin{align}
\label{eq:witten-equation-4d-3}
& 0 = (\Gamma_+ + 6\Gamma_-)\tilde\nabla_r^3\tilde\psi - 3\{\tilde\nabla_r^2\Gamma_+ + 2(\partial_2^\alpha)\Gamma_- + (\partial_r^2\beta_A)\tilde\Lambda^A\}\tilde\nabla_r\tilde\psi \nonumber \\
& + \{ -2\tilde\nabla_r^3\Gamma_+ - (\partial_r^3\alpha)\Gamma_- - 2(\partial_r^3\beta_A)\tilde\Lambda^A - 6(\partial_r^2\beta_A)\tilde\nabla_r\tilde\Lambda^A\}\tilde\psi \nonumber \\
& + 3\tilde\Lambda^A\tilde\nabla_r^2\tilde\nabla_A\tilde\psi + 6(\tilde\nabla_r\tilde\Lambda^A)\tilde\nabla_r\tilde\nabla_A\tilde\psi + 3(-\partial_r^2\beta^A \Gamma_- + \tilde\nabla_r^2\tilde\Lambda^A)\tilde\nabla_A\tilde\psi + O(r) \quad .
\end{align}
\end{subequations}
To solve these equations for the spinor expansion, we set $r=0$ in each one.
Using equation (\ref{eq:grspinors-spinor-connection-coefficient-A}), we find that
\begin{align}
\left[\tilde\nabla_A\tilde\psi\right]_{r=0}
& = \mathscr{D}_A\tilde\psi^{(0)} - \frac{1}{4}\gamma_{AB}^{(1)}\tilde\Lambda^{B(0)}\Gamma_+\tilde\psi^{(0)} \nonumber \\
& = \mathscr{D}_A\tilde\psi^{(0)} \quad ,
\end{align}
where in the second line we used that $\Gamma_+\tilde\psi^{(0)}=0$ (because $\tilde\psi_-^{(0)}=0$).
In addition, we make use of the expressions given in Appendix \ref{app:grspinors-gamma-matrices-derivatives} for the covariant $r$ derivatives of the gamma matrices.
In particular, we use that
\begin{subequations}
\begin{align}
\left[\tilde\nabla_r\tilde\Lambda^A\right]_{r=0} &= -\frac{1}{2}s^{AB}\gamma_{BC}^{(1)}\tilde\Lambda^{C(0)} \quad , \\
\left[\tilde\nabla_r^2\tilde\Lambda^A\right]_{r=0} &= \beta^{A(2)}\Gamma_- - \frac{3}{4}\gamma^{AB(1)}\gamma_{BC}^{(1)}\tilde\Lambda^{C(0)} - s^{AB}\gamma_{BC}^{(2)}\tilde\Lambda^{C(0)} \quad ,
\end{align}
\begin{align}
\left[\tilde\nabla_r\Gamma_+\right]_{r=0} &= 0 \quad , \\
\left[\tilde\nabla_r^2\Gamma_+\right]_{r=0} &= -\beta_A^{(2)}\tilde\Lambda^{A(0)} \quad , \\
\left[\tilde\nabla_r^3\Gamma_+\right]_{r=0} &= -3\beta_A^{(3)}\tilde\Lambda^{A(0)} + \beta^{A(2)}\gamma_{AB}^{(1)}\tilde\Lambda^{B(0)} \quad .
\end{align}
\end{subequations}
Equations (\ref{eq:witten-equation-4d}) also contain terms of the form $\tilde\nabla_r^k\tilde\nabla_A\tilde\psi$.
To evaluate these at $r=1$, we commute the derivatives using Lemma \ref{thm:spinor-derivative-commutation}.
In particular we find that
\begin{subequations}
\begin{align}
\left[\tilde\nabla_r\tilde\nabla_A\tilde\psi\right]_{r=0}
 = \mathscr{D}_A\tilde\psi^{(1)} &+ \frac{1}{4}\gamma_{AB}^{(1)}\tilde\Lambda^{B(0)}\Gamma_+\tilde\psi^{(1)} + \frac{1}{2}\beta_A^{(2)}\tilde\psi^{(0)} \nonumber \\
& - \frac{1}{4}\mathscr{D}_{[B}\gamma_{D]A}^{(1)}\tilde\Lambda^{B(0)}\tilde\Lambda^{D(0)}\tilde\psi^{(0)} \quad ,
\end{align}
\begin{equation}
\left[\tilde\Lambda^A\tilde\nabla_r\tilde\nabla_A\tilde\psi\right]_{r=0} = \tilde\Lambda^{A(0)}\mathscr{D}_A\tilde\psi^{(1)} + \beta_A^{(2)}\tilde\Lambda^{A(0)}\tilde\psi^{(0)} + \frac{1}{4}s^{AB}\gamma_{AB}^{(1)}\Gamma_+\tilde\psi^{(1)} \quad ,
\end{equation}
and in addition that
\begin{align}
& \left[\bar{\tilde\psi}\tilde\Lambda^A\tilde\nabla_r\tilde\nabla_A\tilde\psi\right]_{r=0} = \alpha^{(2)} = \frac{\lambda}{2} \quad , \\
& \left[\bar{\tilde\psi}\tilde\Lambda^A\tilde\nabla_r^2\tilde\nabla_A\tilde\psi\right]_{r=0} = -2\tilde\psi_+^{(0)}{}^\ast\tilde\sigma^{A(0)}\mathscr{D}_A\tilde\psi_-^{(2)} + \frac{1}{2}\mathscr{D}^A\beta_A^{(2)} + \frac{\lambda}{4}s^{AB}\gamma{AB}^{(1)} \quad .
\end{align}
\end{subequations}

\noindent
Now we have presented these basic results, we are ready to solve the Witten equation at each order.
Setting $r=0$, we see that equations (\ref{eq:witten-equation-4d-0}) and (\ref{eq:witten-equation-4d-1}) are trivial due to the boundary conditions and the properties of the Minkowski spinor.
\iftoggle{complete}{
\begin{align*}
& -2\Gamma_+\tilde\psi^{(0)} = 0 \ , \quad \tilde\psi^{(0)}_- = 0 \quad , \\
& (-\sqrt{2}\tilde\psi^{(1)}_- + \tilde\sigma^{A(0)}\mathscr{D}_A\tilde\psi^{(0)}_+ = 0 \quad , \\
& 2\sqrt{2}\tilde\psi^{(1)}_+ - \tilde\sigma^{A(0)}\mathscr{D}_A\tilde\psi^{(0)}_- = 0 \quad .
\end{align*}
}{}
Setting $r=0$ in equation (\ref{eq:witten-equation-4d-2}), we find
\begin{equation}
0 = 8\Gamma_-\tilde\psi^{(2)} - \lambda\Gamma_-\tilde\psi^{(0)} + 2\tilde\Lambda^{A(0)}\mathscr{D}_A\tilde\psi^{(1)} + \frac{1}{2}s^{AB}\gamma_{AB}^{(1)}\Gamma_+\tilde\psi^{(1)} - s^{AB}\gamma_{BC}^{(0)}\tilde\Lambda^{C(0)}\mathscr{D}_A\tilde\psi^{(0)} \quad .
\end{equation}
The plus component turns out to be trivial.
Indeed it gives,
\begin{align}
0
& = \iftoggle{complete}{2\tilde\sigma^{A(0)}\mathscr{D}_A\tilde\psi_+^{(1)} +}{} \frac{1}{\sqrt2}s^{AB}\gamma_{AB}^{(1)}\tilde\psi^{(1)}_- - s^{AB}\gamma_{BC}^{(1)}\sigma^{C(0)}\mathscr{D}_A\tilde\psi^{(0)}_+ \nonumber \\
& = \frac{1}{\sqrt2}s^{AB}\gamma_{AB}^{(1)}\tilde\psi^{(1)}_- - \frac{1}{\sqrt2}s^{AB}\gamma_{BC}^{(1)}\tilde\sigma^{C(0)}\tilde\sigma_A^{(0)}\tilde\psi^{(1)}_- \nonumber \\
\iftoggle{complete}{
& = \frac{1}{\sqrt2}s^{AB}\gamma_{AB}^{(1)}\tilde\psi^{(1)}_- - \frac{1}{\sqrt2}\gamma_{BC}^{(1)}\tilde\sigma^{C(0)}\tilde\sigma^{B(0)}\tilde\psi^{(1)}_- \nonumber \\
& = \frac{1}{\sqrt2}s^{AB}\gamma_{AB}^{(1)}\tilde\psi^{(1)}_- - \frac{1}{\sqrt2}\gamma_{BC}^{(1)}s^{BC}\tilde\psi^{(1)}_- \nonumber \\
}{}
& = 0 \quad ,
\end{align}
where we have rewritten $\mathscr{D}_A\tilde\psi_+^{(0)}$ in terms of $\tilde\psi_-^{(1)}$ using the conditions (\ref{eq:minkowski-spinor-relation}) on the Minkowski spinor.
Applying these conditions to the minus component, we get
\begin{align}
0
\iftoggle{complete}{
& = 8\sqrt{2}\tilde\psi^{(2)}_+ - \sqrt{2}\lambda\tilde\psi^{(0)}_+ - 2\tilde\sigma^{A(0)}\mathscr{D}_A\tilde\psi^{(1)}_- + s^{AB}\gamma_{BC}^{(1)}\tilde\sigma^{C(0)}\mathscr{D}_A\tilde\psi^{(0)}_- \nonumber \\
}{}
& = 8\sqrt{2}\tilde\psi^{(2)}_+ - \sqrt{2}\lambda\tilde\psi^{(0)}_+ + \frac{\lambda}{\sqrt2}\tilde\sigma^{A(0)}\tilde\sigma_A^{(0)}\tilde\psi^{(0)}_+ \nonumber \\
\iftoggle{complete}{
& = 8\sqrt{2}\tilde\psi^{(2)}_+ - \sqrt{2}\lambda\tilde\psi^{(0)}_+ + \sqrt{2}\lambda\tilde\psi^{(0)}_+ \nonumber \\
}{}
& = 8\sqrt{2}\tilde\psi^{(2)}_+ \quad ,
\end{align}
where we used that $\tilde\sigma^{A(0)}\tilde\sigma_A^{(0)}=s^{AB}s_{AB}=(d-2)$.
Thus $\tilde\psi^{(2)}_+=0$.
Note that $\tilde\psi_-^{(2)}$ is not determined by this equation due to a similar breakdown in the leading coefficient as was seen when solving the Einstein equations in Section \ref{sec:einstein-equations-solution}.
Finally, setting $r=0$ in equation (\ref{eq:witten-equation-4d-3}), left multiplying by $\tilde\psi^{(0)}\Gamma_-$, and taking the real part, we get precisely equation (\ref{eq:psi+-3-condition}).

We have therefore shown that the spinor takes the desired form,
\begin{subequations}
\label{eq:spinor-expansion}
\begin{align}
& \tilde\psi_+ = \tilde\psi_+^{(0)} + r^3\tilde\psi_+^{(3)} + O(r^4) \quad ,\\
& \tilde\psi_- = r\tilde\psi_-^{(1)} + r^2\tilde\psi_-^{(2)} + O(r^3) \quad .
\end{align}
\end{subequations}
\end{proof}

%% file: BondiPositivity4D.tex
\section{Bondi Mass Positivity in 4 Dimensions}
\label{sec:bondi-positivity-4d}
\setcounter{section}{1}
\setcounter{theorem}{0}
A spinorial proof is given, for $d=4$, that the Bondi mass is positive.
This demonstrates how our formalism can be used to prove the positivity result.
Up until this point, all calculations have been performed in arbitrary dimension $d\geq4$ (apart from the expansion of the Witten spinor in Section \ref{sec:dirac-spinor-solution}, which is readily generalised).
However the positivity result is more difficult to obtain in the higher dimensional case.
When expanding the Witten spinor, the expansion coefficient $\tilde\psi_-^{(d/2)}$ is undetermined.
This is due to a breakdown similar to that which occurred in the Einstein equations when expanding the metric.
This expansion coefficient causes problems when proving the positivity result since it cannot be evaluated.
In the 4 dimensional case, it cancels from the calculations, meaning the proof can proceed.
In the higher dimensional case, however, it is not clear that these problem terms cancel out or can even be shown to be zero.
For this reason, the proof in the higher dimensional case has not been completed at the time of writing, and so only the case $d=4$ is considered below.

Bondi positivity proofs generally proceed in two main steps.
First a $2$-form is introduced, whose exterior derivative is shown by a standard argument to be positive.
Then the integral of this $2$-form on a cross section of $\scri^+$, called the ``spinor mass'' \cite{Hollands:2005wt}, is shown to equal the Bondi mass.
Finally, it follows from Stokes' theorem that the spinor mass, and therefore the Bondi mass, is positive.

\begin{theorem}
\label{thm:bondi-positivity-4d}
Suppose there exists a spinor $\psi$ satisfying the conditions of Theorem \ref{thm:witten-solution-4d}.
Suppose also that the stress energy tensor $T_{ab}$ satisfies the dominant energy condition, namely that $T^a{}_b t^b$ is future directed for all future directed $t^b$.
Then for a 4 dimensional asymptotically flat spacetime, the Bondi mass as defined by equation (\ref{eq:bondi-mass-definition}) is positive on every cross section $\Sigma(0,u)$ of $\scri^+$.
\end{theorem}

To prove this theorem, we first obtain the expansion of the vector $\xi^a=\bar\psi\Gamma^a\psi$ in Lemma \ref{thm:spinor-xi-expansion}.
Since $\psi$ is a Witten spinor, $\xi^a$ is divergence free.
Indeed,
\begin{align}
\nabla_a\xi^a
& = \nabla_a(\bar\psi\Gamma^a\psi) \nonumber \\
& = 2\Re(\bar\psi\Gamma^a\nabla_a\psi) \nonumber \\
& = 0 \quad .
\end{align}
It then follows that the spinor mass
\begin{equation}
\label{eq:spinor-mass}
\int_{\Sigma(0,u)} B
\end{equation}
is equal to the Bondi mass, where we have defined the $2$-form,
\begin{equation}
\label{eq:spinor-2-form}
B_{a_1a_2} = \nabla^{[d}\xi^{c]} \epsilon_{cda_1a_2} \quad .
\end{equation}
This is verified in Lemma \ref{thm:bondi-spinor-mass}.
We denote the volume element of the surface $\mathscr{S}$ by ${}^{(d-1)}\epsilon_{ba_1a_2}$.
It can then be calculated that
\begin{align}
dB_{b a_1 a_2}
& = [g^{ab}(\nabla_a\psi)^\ast(\nabla_b\psi) + (R_{ab} - \frac{1}{2}Rg_{ab}) N^a \xi^b ]\ {}^{(d-1)} \epsilon_{ba_1a_2} \nonumber \\
& = [h^{ab}(\nabla_a\psi)^\ast(\nabla_b\psi) + 8\pi T_{ab} N^a \xi^b ]\ {}^{(d-1)}\epsilon_{ba_1a_2} \quad ,
\end{align}
where in the second line we used the $3+1$ decomposition of $g_{ab}$, and the fact that $N^a\nabla_a\psi=0$ on $\mathscr{S}$.
We also used Einstein's equation to replace the Einstein tensor $G_{ab}\equiv R_{ab} - (1/2)Rg_{ab}$ with the mass energy tensor $T_{ab}$.

\noindent
Since $h_{ab}$ is a Riemannian metric, it follows that the first term is positive.
The second is positive by the dominant energy condition, provided $\xi^a$ is future directed.
Therefore we have
\begin{equation}
dB_{ba_1\ldots a_{d-2}} \geq 0
\end{equation}
on $\mathscr{S}$.
Thus to complete the positivity proof, we need only prove Lemmas \ref{thm:spinor-xi-expansion} and \ref{thm:bondi-spinor-mass}, and then apply Stokes' theorem to show that the Bondi mass equals an integral of $dB$.

\begin{lemma}
\label{thm:spinor-xi-expansion}
The components of the vector $\xi^a = \bar\psi\Gamma^a\psi$ have the following expansions.
\begin{subequations}
\label{eq:spinor-xi-expansion}
\begin{align}
\label{eq:spinor-xi-expansion-r}
\xi^r & = r^3[\alpha^{(3)} - 2\sqrt{2}\tilde\psi_-^{(1)}{}^\ast\tilde\psi_-^{(2)}] + O(r^4) \quad , \\
\label{eq:spinor-xi-expansion-u}
\xi^u & = 1 + 2\sqrt{2}r^3\tilde\psi_+^{(0)}{}^\ast\tilde\psi_+^{(3)} + O(r^4) \quad , \\
\label{eq:spinor-xi-expansion-A}
\xi^A &= r^2\left[\frac{1}{2}\beta^{A(2)} - 2\Re(\tilde\psi_+^{(0)}{}^\ast\tilde\sigma^{A(0)}\tilde\psi_-^{(2)})\right] + O(r^3) \quad .
\end{align}
\end{subequations}
\end{lemma}
\begin{proof}
Note that we can equivalently write the vector $\xi^a$ in terms of either physical or unphysical spinor quantities,
\begin{equation}
\xi^a = \bar\psi\Gamma^a\psi \equiv \bar{\tilde\psi}\tilde\Gamma^a\tilde\psi \quad .
\end{equation}
Using the spinor expansions (\ref{eq:half-spinor-expansions}), the normalisation conditions (\ref{eq:minkowski-spinor-normalisation}), the explicit forms (\ref{eq:gamma-matrices-representation}) and (\ref{eq:gamma-matrices-explicit-curved}) of the gamma matrices in our representation, and the decomposition of $\tilde\psi$ into $\tilde\psi_\pm$, we see that
\begin{subequations}
\begin{align}
\bar{\tilde\psi}\Gamma_+\tilde\psi
& = \sqrt{2}\tilde\psi_-^\ast\tilde\psi_- \nonumber \\
\iftoggle{complete}{
& = \sqrt{2}r^2\tilde\psi_-^{(1)}{}^\ast\tilde\psi_-^{(1)} + 2\sqrt{2}r^3\Re(\tilde\psi_-^{(1)}{}^\ast\tilde\psi_-^{(2)}) + O(r^4) \nonumber \\
}{}
& = \frac{\lambda}{2} r^2 + 2\sqrt{2}r^3\Re(\tilde\psi_-^{(1)}{}^\ast\tilde\psi_-^{(2)}) + O(r^4) \quad , \\
\bar{\tilde\psi}\Gamma_-\tilde\psi
& = \sqrt{2}\tilde\psi_+^\ast\tilde\psi_+ \nonumber \\
\iftoggle{complete}{
& = \sqrt{2}\tilde\psi_+^{(0)}{}^\ast\tilde\psi_+^{(0)} + 2\sqrt{2}r^3\Re(\tilde\psi_+^{(0)}\tilde\psi_+^{(3)}) + O(r^4) \nonumber \\
}{}
& = 1 + 2\sqrt{2}r^3\Re(\tilde\psi_+^{(0)}\tilde\psi_+^{(3)}) + O(r^4) \quad , \\
\bar{\tilde\psi}\tilde\Lambda^A\tilde\psi
& = -2r\Re(\tilde\psi_+^{(0)}{}^\ast\tilde\sigma^{A(0)}\tilde\psi_-^{(1)}) - 2r^2\Re(\tilde\psi_+^{(0)}{}^\ast\tilde\sigma^{A(0)}\tilde\psi_-^{(2)} + O(r^3) \nonumber \\
& = - 2r^2\Re(\tilde\psi_+^{(0)}{}^\ast\tilde\sigma^{A(0)}\tilde\psi_-^{(2)} + O(r^3) \quad .
\end{align}
\end{subequations}
We can then obtain the components of the vector $\xi^a$.
\begin{align}
\label{eq:spinor-symmetry-components-r}
\xi^r
& = \bar{\tilde\psi}\tilde\Gamma^r\tilde\psi \nonumber \\
& = \bar{\tilde\psi}\tilde(\Gamma_+ + \alpha\Gamma_- + \beta_A\tilde\Lambda^A)\tilde\psi \nonumber \\
& = r^3[\alpha^{(3)} - 2\sqrt{2}\tilde\psi_-^{(1)}{}^\ast\tilde\psi_-^{(2)}] + O(r^4) \quad , \\
\label{eq:spinor-symmetry-components-u}
\xi^u
& = \bar{\tilde\psi}\tilde\Gamma^u\tilde\psi \nonumber \\
& = \bar{\tilde\psi}\Gamma_-\tilde\psi \nonumber \\
& = 1 + 2\sqrt{2}r^3\tilde\psi_+^{(0)}{}^\ast\tilde\psi_+^{(3)} + O(r^4) \quad .
\end{align}
Since the $A$ component involves the curved space gamma matrix $\tilde\Lambda^A$, a little more care is necessary.
Instead of directly substituting the spinor expansions, we take $r$ derivatives of $\xi^A$ and set $r=0$ to determine the expansion coefficient at each order.
Since $\xi^A$ is a scalar, it does not matter which derivative operator we use; we choose the covariant derivative $\tilde\nabla_r$.
Firstly, setting $r=0$ in $\xi^A$ itself,
\begin{equation}
\xi^{A(0)} = \bar{\tilde\psi}^{(0)}\tilde\Lambda^{A(0)}\tilde\psi^{(0)} = -2\Re(\tilde\psi_+^{(0)}{}^\ast\tilde\sigma^{A(0)}\tilde\psi_-^{(0)}) = 0 \quad .
\end{equation}
Then taking the first derivative, and setting $r=0$, we find
\begin{align}
\xi^{A(1)}
& = 2\Re\bar{\tilde\psi}^{(0)}\tilde\Lambda^{A(0)}\tilde\psi^{(1)} + \tilde{\bar\psi}^{(0)}\left[\tilde\nabla_r\tilde\Lambda^A\right]_{r=0}\tilde\psi^{(0)} \nonumber \\
& = -2\Re(\tilde\psi_+^{(0)}{}^\ast\tilde\sigma^{A(0)}\tilde\psi_-^{(1)}) + \tilde{\bar\psi}^{(0)}\left(-\frac{1}{2}s^{AB}\gamma_{BC}^{(1)}\tilde\Lambda^{C(0)}\right)\tilde\psi^{(0)} \nonumber \\
& = 0 \quad .
\end{align}
Finally, taking a second derivative and setting $r=0$, we get
\begin{align}
2\ \xi^{A(2)}
& = \bar{\tilde\psi}^{(0)}\left[\tilde\nabla_r^2\tilde\Lambda^A\right]_{r=0}\tilde\psi^{(0)} + 3\Re\left(\bar{\tilde\psi}^{(0)}\left[\tilde\nabla_r\tilde\Lambda^A\right]_{r=0}\tilde\psi^{(1)}\right) \nonumber \\
& + 4\Re(\tilde\psi^{(0)}\tilde\Lambda^{A(0)}\tilde\psi^{(2)}) - 4\Re(\tilde\psi_+^{(1)}\tilde\sigma^{A(0)}\tilde\psi_-^{(1)}) \nonumber \\
\iftoggle{complete}{
& = \bar{\tilde\psi}^{(0)}\left(\beta^{A(2)}\Gamma_- - \frac{3}{4}\gamma^{AB(1)}\gamma_{BC}^{(1)}\tilde\Lambda^{C(0)} - s^{AB}\gamma_{BC}^{(2)}\tilde\Lambda^{C(0)}\right)\tilde\psi^{(0)} \nonumber \\
& + \frac{3}{2}s^{AB}\gamma_{BC}^{(1)}\Re(\tilde\psi_+^{(0)}\tilde\sigma^{C(0)}\tilde\psi_-^{(1)}) - 4\Re(\tilde\psi_+^{(0)}\tilde\sigma^{A(0)}\tilde\psi_-^{(2)}) \nonumber \\
& = \beta^{A(2)}\sqrt{2}\tilde\psi_+^{(0)}{}^\ast\tilde\psi_+^{(0)} - 4\Re(\tilde\psi_+^{(0)}\tilde\sigma^{A(0)}\tilde\psi_-^{(2)}) \nonumber \\
}{}
& = \beta^{A(2)} - 4\Re(\tilde\psi_+^{(0)}\tilde\sigma^{A(0)}\tilde\psi_-^{(2)}) \quad .
\end{align}
Summarising these results, we obtain equation (\ref{eq:spinor-xi-expansion-A}), as required.
\end{proof}
\begin{remark}
The expansion we have obtained for $\xi^a$ is indeed a special case of the expansions for asymptotic symmetries we wrote down in Theorem \ref{thm:asymptotic-symmetries}.
\end{remark}

\begin{lemma}
\label{thm:bondi-spinor-mass}
The Bondi mass on the cross section $\Sigma(0,0)$ is equal to the spinor mass,
\begin{equation}
m(0) = \int_{\Sigma(0,0)} B \quad .
\end{equation}
\end{lemma}
\begin{proof}
To prove that the Bondi mass formula (\ref{eq:bondi-mass-formula}) is equivalent to the spinor mass (\ref{eq:spinor-mass}), we simply evaluate the spinor mass in coordinates and compare the result to the Bondi mass formula.
\medskip\noindent
Now,
\begin{align}
\int_{\Sigma(r,0)} B
& = \int_{\Sigma(r,0)} \nabla^{[d}\xi^{c]}\epsilon_{cda_1a_2} \nonumber \\
& = \int_{\Sigma(r,0)} \nabla^{[d}\xi^{c]} \sqrt{g} (dr\wedge du\wedge dx^1\wedge dx^2)_{cda_1a_2} \nonumber \\
& = 2 \int_{\Sigma(r,0)} (\partial_{[a} g_{b]e}\xi^e) g^{ad}g^{cb} dr_c du_d \sqrt{g} (dx^1\wedge dx^2)_{a_1a_2} \quad ,
\end{align}
where in the last line we used that $\nabla_{[a}t_{b]}=\partial_{[a}t_{b]}$.
Rewriting in terms of unphysical quantities, we find
\begin{align}
\int_{\Sigma(r,0)} B
& = 2 \int_{\Sigma(r,0)} (\partial_{[a} r^{-2} \tilde{g}_{b]e}\xi^e) \ r^4 \ \tilde{g}^{ad}\tilde{g}^{cb} dr_c du_d \frac{1}{r^4} \sqrt{\tilde{g}} (dx^1\wedge dx^2)_{a_1a_2} \nonumber \\
& = 2 \int_{\Sigma(r,0)} (\partial_{[a} r^{-2} \tilde{g}_{b]c}\xi^c) \tilde{g}^{au}\tilde{g}^{br} \sqrt{\gamma} (dx^1\wedge dx^2)_{a_1a_2} \nonumber \\
& = 2 \int_{\Sigma(r,0)} (\partial_{[a} r^{-2} \tilde{g}_{b]c}\xi^c) \tilde{g}^{au}\tilde{g}^{br} \ {}^{(2)}\tilde\epsilon_{a_1a_2} \quad .
\end{align}
Finally, we calculate the integrand and find that
\begin{align}
\partial_{[a}(r^{-2}\tilde{g}_{b]c}\xi^c) \tilde{g}^{au} \tilde{g}^{br}
\iftoggle{complete}{
& = \partial_{[r}(r^{-2}\tilde{g}_{u]c}\xi^c) + \beta^A\partial_r{[r}(r^{-2}\tilde{g}_{A]}\xi^c) \nonumber \\
& = \frac{1}{2}\partial_r(r^{-2}\tilde{g}_{uc}\xi^c) - \frac{1}{2}r^{-2}\partial_u(\tilde{g}_{rc}\xi^c) + \frac{1}{2}\beta^A\partial_r(r^{-2}\tilde{g}_{Ac}\xi^c) - \frac{1}{2}r^{-2}\beta^A\partial_A(\tilde{g}_{rc}\xi^c) \nonumber \\
& = -r^{-3}(\xi^r  - 2\alpha\xi^0) + \frac{1}{2}r{-2}\partial_r(\xi^r - 2\alpha\xi^0) - \frac{1}{2}r^{-2}\tilde\nabla_u\xi^u + \frac{1}{2}r^{-2}(\partial_r\alpha\xi^0) \nonumber \\
& = - 6\sqrt{2}\tilde\psi_+^{(0)}{}^\ast\tilde\psi_+^{(3)} - \sqrt{2}\tilde\psi_-^{(1)}{}^\ast\tilde\psi_-^{(2)} + \mathscr{D}_A(\tilde\psi_+^{(0)}\tilde\sigma^{A(0)}\tilde\psi_-^{(2)}) + \mathscr{D}_A v^A + O(r) \nonumber \\
}{}
& = -6\sqrt{2}\tilde\psi_+^{(0)}{}^\ast\tilde\psi_+^{(3)} + \tilde\psi_+^{(0)}\tilde\sigma^{A(0)}\mathscr{D}_A\tilde\psi_-^{(2)} + \mathscr{D}_A v^A + O(r) \quad ,
\end{align}
for some vector $v^A$ on the surfaces $\Sigma(r,0)$.
If we apply relation (\ref{eq:psi+-3-condition}), we find that
\begin{align}
\partial_{[a}(r^{-2}\tilde{g}_{b]c}\xi^c) \tilde{g}^{au} \tilde{g}^{br}
& = \frac{1}{8}\gamma^{AB(1)}\partial_u\gamma_{AB}^{(1)} - \alpha^{(3)} + \mathscr{D}_A w^A + O(r) \quad ,
\end{align}
for another vector $w^A$ on the surfaces $\Sigma(r,0)$.
Integrating over $\Sigma(r,0)$ and taking the limit at $\scri^+$, we find that the spinor mass is
\begin{equation}
\lim_{r\rightarrow0} \int_{\Sigma(r,0)} B = \int_{\Sigma(0,0)} \left( \frac{1}{8}\gamma^{AB(1)}\partial_u\gamma_{AB}^{(1)} - \alpha^{(3)} \right) \quad ,
\end{equation}
since the total divergence and terms $O(r)$ have dropped out.
Up to a constant factor, this is precisely the Bondi mass formula (\ref{eq:bondi-mass-formula}) for $d=4$.
\end{proof}
\noindent
We are now able to complete the proof of the positivity result.
\begin{proof}[Proof of Theorem \ref{thm:bondi-positivity-4d}]
Define $\mathscr{S}(r)$ to be the section of $\mathscr{S}$ bounded by $\Sigma(r,r/2)$.
Therefore, applying Stokes' theorem, we have
\begin{align}
m(0)
& = \lim_{r\rightarrow0}\int_{\Sigma(r,r/2)}B \nonumber \\
& = \lim_{r\rightarrow0}\int_{\mathscr{S}(r)} dB \nonumber \\
& \geq 0 \quad .
\end{align}
Inner boundary terms might also arise from applying Stokes' theorem, but they can be eliminated by imposing suitable boundary conditions on the spinor $\psi$ \cite{Horowitz}.
We have thus proven the positivity of the Bondi mass $m(0)$ on the cross section $\Sigma(0,0)$ of $\scri^+$.
Since our initial choice of the cross section of $\scri^+$ on which $u=0$ was arbitrary, it follows that the Bondi mass is positive on all cross sections,
\begin{equation}
m(u) \geq 0 \quad ,
\end{equation}
for all $u$.
\end{proof}

%% file: Conclusion.tex
\section{Conclusion}
\label{sec:conclusion}

We first gave a definition of asymptotic flatness within the framework of conformal infinity.
We specified the minimal conditions necessary to ensure that the metric expansion coefficients could be determined recursively by the vacuum Einstein equations.
Starting from formulae given in \cite{Hollands:2003ie}, we then obtained simple coordinate formulae for the Bondi mass, news and flux in terms of these expansion coefficients, in even dimension $d\geq4$, .
We cannot apply the same procedure in odd dimensions, since the framework of conformal null infinity does not apply in that case.
Even if we did try to apply our formalism in odd dimensions we would find that the News tensor $N_{AB}=0$ and thus radiating spacetimes would be excluded from consideration.

Having obtained the asymptotic expansion of a Witten spinor in four dimensions, we defined the spinor mass and showed that it was positive by a standard argument.
The Bondi mass was shown to equal the spinor mass by comparing the coordinate expressions for both.
This comprised a spinorial proof that the Bondi mass is positive in the four dimensional case.

Since the Bondi mass formula is valid in arbitrary even dimensions, the positivity argument can, in principle, be generalised to higher dimensions.
The spinor formalism readily generalises to the higher dimensional case: we can obtain a spatial Witten equation from the Dirac equation on a spacelike, asymptotically null hypersurface; and we can solve for the spinor expansion at each order of $r$.
A complication arises, however, when calculating the spinor mass.
Indeed, apparently singular terms cause more of a problem in the higher dimensional case than in four dimensions.
Most such terms can be shown to be total divergences and so they do not affect the integral.
However, it is not clear that terms of the form $\tilde\psi_-^{(1)}{}^\ast\tilde\psi_-^{(d/2)}$ can be dealt with in this way.
In four dimensions, these terms cancelled, but this does not appear to happen in the higher dimensional case.
The result is that in higher dimensions the spinor mass may not even be convergent, let alone equal to the Bondi mass.
If this is the case, and not the result of a calculation error, it would mean that a spinorial proof of positivity is simply not possible in higher dimensions.
This could be a problem with the spinor method, or perhaps is a sign that the positivity result simply does not hold for the Bondi mass in higher dimensions.
It may turn out, however, that these terms do indeed cancel after corrections to the calculations have been made.
Should this be the case then the proof could proceed as in four dimensions, which would demonstrate the positivity of the Bondi mass in even dimension $d\geq4$.

%% file: Notes.tex
\section{Supplementary Notes}
\label{app:notes}
\setcounter{section}{1}

\begin{note}
\label{note:asymptotic-flatness}
Standard definitions of asymptotic flatness \cite{Hollands:2003ie,Wald} invoke the notion of a background spacetime.
As shown in Corollary \ref{thm:background-metric}, we can rewrite the asymptotic expansions of Theorem \ref{thm:metric-expansion-low-order} in the notation of these definitions of asymptotic flatness.

\noindent
If we define an unphysical ``background metric'',
\begin{equation}
\label{eq:background-metric}
\tilde{\bar{g}}_{ab} = 2dr_{(a}du_{b)} - r^2du_a du_b + s_{AB}dx^A_a dx^B_b \quad ,
\end{equation}
then an asymptotically flat metric $g_{ab}$ satisfies the following conditions,
\begin{subequations}
\label{eq:asymptotic-falloff}
\begin{align}
& \label{eq:asymptotic-falloff-gab}
  \tilde{g}_{ab} - \tilde{\bar{g}}_{ab} = O(r^{(d-2)/2}) \quad , \\
& \label{eq:asymptotic-falloff-gau}
  (\tilde{g}_{ab} - \tilde{\bar{g}}_{ab})(\partial_u)^a = O(r^{d/2}) \quad , \\
& \label{eq:asymptotic-falloff-guu}
  (\tilde{g}_{ab} - \tilde{\bar{g}}_{ab})(\partial_u)^a(\partial_u)^b = O(r^{(d+2)/2}) \quad ,
\end{align}
and when $d>4$ we have in addition that,
\begin{equation}
\label{eq:asymptotic-falloff-trace}
\tilde{g}^{ab}\tilde{g}_{ab} - \tilde{\bar{g}}^{ab}\tilde{\bar{g}}_{ab} = O(r^{2d-2}) \quad ,
\end{equation}
or equivalently,
\begin{align}
& \label{eq:asymptotic-falloff-hollands-trace}
  \tilde\epsilon_{ab\ldots} - \tilde{\bar\epsilon}_{ab\ldots} = O(r^{2d-2}) \quad ,
\end{align}
\end{subequations}
where $\tilde\epsilon_{ab\ldots}$ and $\tilde{\bar\epsilon}_{ab\ldots}$ are the volume elements of the unphysical metric $\tilde{g}_{ab}$ and the unphysical background $\tilde{\bar{g}}_{ab}$, respectively.
These conditions match conditions (3) and (4) in the definition of asymptotic flatness of \cite{Hollands:2003ie}, except that the present condition on the volume element is stronger.
\end{note}

\begin{note}
\label{note:lemma-proof}
\begin{proof}[Proof of Lemma \ref{thm:gamma-asymptotics}]
\begin{enumerate}[i)]
\item
By definition,
\begin{equation}
\label{eq:gamma-properties-lemma-gamma-inverse-definition}
0 = \delta^A{}_D = \gamma^{AC}\gamma_{CD} \quad .
\end{equation}
Differentiation with respect to $r$ gives
\begin{equation}
0 = \gamma^{AC}\partial_r\gamma_{CD} + (\partial_r\gamma^{AC})\gamma_{CD} \quad ,
\end{equation}
and following rearrangement and contraction of both sides with $\gamma^{BD}$, we find that,
\begin{equation}
\partial_r\gamma^{AB}
 = \partial_r\gamma^{AC}\delta_C{}^B
 = -\gamma^{AC}\gamma^{BD}\partial_r\gamma_{CD} = O(r^{k-1}) \quad .
\end{equation}
Integrating with respect to $r$ gives
\begin{equation}
\gamma^{AB} = \gamma^{AB(0)} + O(r^k) \quad .
\end{equation}
Now, equation (\ref{eq:gamma-properties-lemma-gamma-inverse-definition}) at order $1$ gives
\begin{equation}
0 = \delta^A{}_D = \gamma^{AC(0)}s_{CD} \quad .
\end{equation}
Since the inverse of $s_{CD}$ is unique, $\gamma^{AC(0)}=s^{AC}$. This completes the proof of (\ref{eq:gamma-inverse-asymptotics}).
\item
Since $\partial_r\gamma_{AB}=O(r^{k-1})$, it follows from the $rr$ component of the Einstein equations (\ref{eq:grtensors-einstein-rr}) that
\begin{equation}
\gamma^{AB}\partial_r^2\gamma_{AB}=O(r^{2k-2}) \quad .
\end{equation}
Therefore,
\begin{equation}
s^{AB}\partial_r^2\gamma_{AB} = \gamma^{AB}\partial_r^2\gamma_{AB} - (\gamma^{AB} - s^{AB})\partial_r^2\gamma_{AB} = O(r^{2k-2}) \quad ,
\end{equation}
where we used that $\partial_r^2\gamma_{AB}=O(r^{k-2})$ and $\gamma^{AB}-s^{AB}=O(r^k)$ by part (\ref{eq:gamma-inverse-asymptotics}).
Integrating twice with respect to $r$ gives
\begin{equation}
s^{AB}\gamma_{AB} = c_1 + c_2r + O(r^{2k}) \quad .
\end{equation}
But we also know that
\begin{equation}
s^{AB}\gamma_{AB} = s^{AB}s_{AB} + O(r^k) = (d-2) + O(r^k) \quad .
\end{equation}
Since $k>1$, comparing the two expressions yields $c_1=d-2$ and $c_2=0$.
This completes the proof of (\ref{eq:gamma-trace-asymptotics}).
\item
The Christoffel symbols of $\gamma_{AB}$ are given by
\begin{equation}
\tilde\Lambda^C{}_{AB} = \frac{1}{2}\gamma^{CD}(\partial_A\gamma_{BD} + \partial_B\gamma_{AD} - \partial_D\gamma_{AB}) \quad .
\end{equation}
Since this expression contains only $\gamma_{AB}$ with no $r$ derivatives or explicit powers of $r$, the result immediately follows.
Note that
\begin{equation}
\label{eq:gamma-christoffel-symbols-order-1}
\tilde\Lambda^C{}_{AB}^{(0)} = \frac{1}{2}s^{CD}(\partial_A s_{BD} + \partial_B s_{AD} - \partial_D s_{AB}) \quad .
\end{equation}
\item
The Ricci tensor of $\gamma_{AB}$ is defined by
\begin{equation}
\mathcal{R}_{AB} = \partial_C\Lambda^C{}_{AB} - \partial_A\Lambda^C{}_{CB} + \Lambda^D{}_{AB}\Lambda^C{}_{CD} - \Lambda^D{}_{CB}\Lambda^C{}_{DA} \quad .
\end{equation}
By the same argument as in the proof of part (iii), $\mathcal{R}_{AB}$ has the same asymptotic behaviour as $\tilde\Lambda^C{}_{AB}$ and hence $\gamma_{AB}$.
\end{enumerate}
\end{proof}
\end{note}

\begin{note}
\begin{proof}[Relations between physical and unphysical spinor derivatives\quad]
\label{note:conformal-spinor-connection}
Here we provide the calculations leading to the relations (\ref{eq:conformal-spinor-connection}) and (\ref{eq:conformal-spinor-derivative}).
The physical and unphysical derivative operators act on a spinor $\psi$ according to
\begin{align}
& \nabla_a\psi = \partial_a\psi + \omega_a\psi \quad , \\
& \tilde\nabla_a\psi = \partial_a\psi + \tilde\omega_a\psi \quad .
\end{align}
Subtracting one equation from the other yields a the relationship between the derivative operators,
\begin{equation}
\label{eq:conformal-spinor-derivative-halfway}
(\nabla_a - \tilde\nabla_a)\psi = (\omega_a-\tilde\omega_a)\psi \quad .
\end{equation}
Applying the definitions of the physical and unphysical connection coefficients, we have
\begin{equation}
\omega_a - \tilde\omega_a
 = \frac{1}{8}(e^{\mu b}\nabla_a e^\nu_b - \tilde{e}^{\mu b}\tilde\nabla_a\tilde{e}_b^\nu)[\Gamma_\mu,\Gamma_\nu] \quad .
\end{equation}
From equation (\ref{eq:spinor-basis-conformal-relation}), we can express the physical basis $e^\mu_a$ in terms of the unphysical basis $\tilde{e}^\mu_a$.
We find
\begin{align}
\label{eq:conformal-spinor-connection-halfway}
\omega_a - \tilde\omega_a
& = \frac{1}{8}(r\tilde{e}^{\mu b}\nabla_a(r^{-1}\tilde{e}^\nu_b) - \tilde{e}^{\mu b}\tilde\nabla_a\tilde{e}^\nu_b)[\Gamma_\mu,\Gamma_\nu] \nonumber\\
& = \frac{1}{8}(-r^{-1}\tilde{e}^{\mu b}\tilde{e}^\nu_b\nabla_a r + \tilde{e}^{\mu b}\nabla_a\tilde{e}^\nu_b - \tilde{e}^{\mu b}\tilde\nabla_a\tilde{e}^\nu_b)[\Gamma_\mu,\Gamma_\nu] \nonumber \\
& = \frac{1}{8}(r^{-1}\lambda^{\mu\nu}\tilde\nabla_a r + \tilde{e}^{\mu b}(\nabla_a - \tilde\nabla_a)\tilde{e}^\nu_b)[\Gamma_\mu,\Gamma_\nu] \quad ,
\end{align}
where in the first term of the last line, we used the equivalence of derivative operators acting on a scalar,
and the fact that the vectors $\tilde{e}^{\mu a}$ are orthonormal.
We have thus reduced the difference between spinor connection coefficients to a difference between derivative operators acting on an ordinary covector.
We can therefore apply the standard relation for derivatives under conformal transformations, for a covector $t_a$,
\begin{equation}
\nabla_a t_b = \tilde\nabla_a t_b - \tilde{C}^c{}_{ab}t_c \quad ,
\end{equation}
where the ``conformal connection coefficients'' are given by
\begin{align}
\tilde{C}^c{}_{ab}
& = \frac{1}{2}g^{cd}(\tilde\nabla_a g_{bd} + \tilde\nabla_b g_{ad} - \tilde\nabla_d g_{ab}) \nonumber \\
& = \frac{1}{2}r^2\tilde{g}^{cd}(\tilde\nabla_a r^{-2} \tilde{g}_{bd} + \tilde\nabla_b r^{-2} \tilde{g}_{ad} - \tilde\nabla_d r^{-2} \tilde{g}_{ab}) \nonumber \\
& = -r^{-1}(\delta^c{}_a \tilde\nabla_a r + \delta^c{}_b \tilde\nabla_b r - \tilde{g}_{ab} \tilde\nabla^c r) \quad .
\end{align}
We find that
\begin{align}
\nabla_a\tilde{e}^\nu_b
& = \tilde\nabla_a\tilde{e}^\nu_b - \tilde{C}^c{}_{ab}\tilde{e}^\nu_c \nonumber \\
& = \tilde\nabla_a\tilde{e}^\nu_b + r^{-1}(\delta^c{}_a\tilde\nabla_a r + \delta^c{}_b\tilde\nabla_b r - \tilde{g}_{ab}\tilde\nabla^c r)\tilde{e}^\nu_c \nonumber \\
& = r^{-1}(\tilde{e}^\nu_b\tilde\nabla_a r + \tilde{e}^\nu_a\tilde\nabla_b r - \tilde{g}_{ab}\tilde{e}^{\nu c}\tilde\nabla_c r) \quad .
\end{align}
Substituting this back into equation (\ref{eq:conformal-spinor-connection-halfway}), we get
\begin{align}
\omega_a - \tilde\omega_a
& = \frac{1}{8}[-r^{-1}\lambda^{\mu\nu}\tilde\nabla_a r + r^{-1}\tilde{e}^{\mu b}(\tilde{e}^\nu_b\tilde\nabla_a r + \tilde{e}^\nu_a \tilde\nabla_b r - \tilde{g}_{ab}\tilde{e}^{\nu c}\tilde\nabla_c r)][\Gamma_\mu,\Gamma_\nu] \nonumber \\
& = \frac{1}{8}[-r^{-1}\lambda^{\mu\nu}\tilde\nabla_a r + r^{-1}\lambda^{\mu\nu}\tilde\nabla_a r + r^{-1}\tilde{e}^\nu_a\tilde{e}^{\mu b}\tilde\nabla_b r - r^{-1}\tilde{e}^\mu_a\tilde{e}^\nu_b\tilde\nabla^b r][\Gamma_\mu,\Gamma_\nu] \nonumber \\
& = -\frac{1}{4}r^{-1}\tilde{e}^{[\mu}_a\tilde{e}^{\nu]}_b(\tilde\nabla^b r)[\Gamma_\mu,\Gamma_\nu] \nonumber \\
& = \frac{1}{4}r^{-1}\tilde{e}^{[\mu}_a\tilde{e}^{\nu]}_b(\tilde\nabla^b r)[\Gamma_\nu,\Gamma_\mu] \nonumber \\
& = \frac{1}{4}r^{-1}\tilde{e}^{[\mu}_a\tilde{e}^{\nu]}_b(\tilde\nabla^b r)(\{\Gamma_\nu,\Gamma_\mu\} - 2\Gamma_\mu\Gamma_\nu) \nonumber \\
& = \frac{1}{2}r^{-1}\tilde{e}^\mu_a\tilde{e}^\nu_b(\tilde\nabla^b r)\lambda_{\nu\mu} - \frac{1}{2}r^{-1}\tilde{e}^\mu_a\tilde{e}^\nu_b(\tilde\nabla^b r)\Gamma_\mu\Gamma_\nu \nonumber \\
& = \frac{1}{2}r^{-1}\tilde{g}_{ab}\tilde\nabla^b r - \frac{1}{2}r^{-1}\tilde\Gamma_a\tilde\Gamma_b\tilde\nabla^r b \nonumber \\
& = \frac{1}{2}r^{-1}(\tilde\nabla_a r - \tilde\Gamma_a\tilde\Gamma_b\tilde\nabla^b r) \quad .
\end{align}
Substituting this back into equation (\ref{eq:conformal-spinor-derivative-halfway}) gives
\begin{equation}
\label{eq:conformal-spinor-derivative-halfway-2}
\nabla_a\psi = \tilde\nabla_a\psi + \frac{1}{2}r^{-1}(\tilde\nabla_a r - \tilde\Gamma_a\tilde\Gamma_b\tilde\nabla^b r)\psi \quad ,
\end{equation}
which is the required relation between physical and unphysical derivatives.
All that remains is to introduce the unphysical spinor $\tilde\psi = r^{1/2}\psi$.
It follows that
\begin{align}
\tilde\nabla_a\psi
& = \tilde\nabla_a(r^{-1/2}\tilde\psi) \nonumber \\
& = r^{-1/2}\tilde\nabla_a\tilde\psi - \frac{1}{2}r^{-3/2}(\tilde\nabla_a r)\tilde\psi \quad .
\end{align}
As a result, equation (\ref{eq:conformal-spinor-derivative-halfway-2}) becomes
\begin{align}
\nabla_a\psi
& = r^{-1/2}\tilde\nabla_a\tilde\psi - \frac{1}{2}r^{-3/2}(\tilde\nabla_a r)\tilde\psi + \frac{1}{2}r^{-1}(\tilde\nabla_a r - \tilde\Gamma_a\tilde\Gamma_b\tilde\nabla^b r)(r^{-1/2}\tilde\psi) \nonumber \\
& = r^{-1/2}\left[\tilde\nabla_a\tilde\psi - \frac{1}{2}r^{-1}(\tilde\nabla_a r)\tilde\psi + \frac{1}{2}r^{-1}(\tilde\nabla_a r)\tilde\psi - \frac{1}{2}r^{-1}\tilde\Gamma_a\tilde\Gamma_b(\tilde\nabla^b r)\tilde\psi\right] \nonumber \\
& = r^{-1/2}\left[\tilde\nabla_a\tilde\psi - \frac{1}{2}r^{-1}\tilde\Gamma_a\tilde\Gamma_b(\tilde\nabla^b r)\tilde\psi\right] \quad , 
\end{align}
which is the desired result.
\end{proof}
\end{note}

\begin{note}
\label{note:gamma-matrices-derivatives}
\begin{proof}[Derivatives of Gamma Matrices]
$\\$Here, we derive the expressions (\ref{eq:grspinors-gamma-matrices-derivatives}).
First we must calculate the $r$ derivatives of the components of the curved space Gamma matrices, $\tilde\Gamma^a$.
By definition, $\tilde\nabla_a\tilde\Gamma_b=0$.
Therefore,
\begin{align}
\label{eq:gamma-matrix-r-derivative-r}
\tilde\nabla_r\tilde\Gamma^r
& = \tilde\nabla_r(\tilde\Gamma^bdr_b) \nonumber \\
& = \tilde\Gamma^b\tilde\nabla_r dr_b \nonumber \\
& = \tilde\Gamma^b(\partial_r dr_b - \tilde\Gamma^c{}_{rb}dr_c) \nonumber \\
& = -\tilde\Gamma^b\tilde\Gamma^r{}_{rb} \nonumber \\
& = -\tilde\Gamma^r\tilde\Gamma^r{}_{rr} - \tilde\Gamma^u\tilde\Gamma^r{}_{ru} - \tilde\Gamma^A\tilde\Gamma^r{}_{rA} \nonumber \\
& = \left(\partial_r\alpha + \frac{1}{2}\beta^A\partial_r\beta_A\right)\Gamma_- + \left(\frac{1}{2}\partial_r\beta_A - \frac{1}{2}\beta^B\partial_r\gamma_{AB}\right)\tilde\Lambda^A \quad .
\end{align}
Proceeding in the same way for the $u$ component, we find,
\begin{align}
\label{eq:gamma-matrix-r-derivative-u}
\tilde\nabla_r\tilde\Gamma^u
& = \tilde\nabla_r(\tilde\Gamma^bdu_b) \nonumber \\
& = \tilde\Gamma^b\tilde\nabla_r du_b \nonumber \\
& = \tilde\Gamma^b(\partial_r dr_b - \tilde\Gamma^c{}_{rb}du_c) \nonumber \\
& = -\tilde\Gamma^b\tilde\Gamma^u{}_{rb} \nonumber \\
& = 0 \quad ,
\end{align}
since $\tilde\Gamma^u{}_{r\mu}=0$ for all $\mu$.
Finally, for the $A$ component,
\begin{align}
\label{eq:gamma-matrix-r-derivative-A}
\tilde\nabla_r\tilde\Gamma^A
& = \tilde\nabla_r(\tilde\Gamma^bdx^A_b) \nonumber \\
& = \tilde\Gamma^b\tilde\nabla_r dx^A_b \nonumber \\
& = \tilde\Gamma^b(\partial_r dr_b - \tilde\Gamma^c{}_{rb}dx^A_c) \nonumber \\
& = -\tilde\Gamma^b\tilde\Gamma^A{}_{rb} \nonumber \\
& = -\tilde\Gamma^r\tilde\Gamma^A{}_{rr} - \tilde\Gamma^u\tilde\Gamma^A{}_{ru} - \tilde\Gamma^B\tilde\Gamma^A{}_{rB} \nonumber \\
& = \frac{1}{2}\gamma^{AB}(\partial_r\beta_B)\Gamma_- - \frac{1}{2}\gamma^{AB}(\partial_r\gamma_{BC})\tilde\Lambda^C \quad .
\end{align}
\end{proof}
\noindent
Now we can calculate the derivatives of $\Gamma_+,\Gamma_-$ and $\tilde\Lambda^A$.
Indeed, since $\Gamma_- = \tilde\Gamma^u$ and $\tilde\Lambda^A = \tilde\Gamma^A$, the first $r$ derivatives of $\tilde\Gamma^u$ and $\tilde\Lambda^A$ are given by equations (\ref{eq:gamma-matrix-r-derivative-u}) and (\ref{eq:gamma-matrix-r-derivative-A}) respectively.
It immediately follows that $\tilde\nabla_r^n\Gamma_-=0$ for integers $n\geq1$, which is precisely equation (\ref{eq:grspinors-gamma-matrices-derivatives--}).

\noindent
Expressing $\Gamma_+$ in terms of curved Gamma matrices yields
\begin{equation}
\Gamma_+ = \tilde\Gamma^r - \alpha\tilde\Gamma^u - \beta_A\tilde\Gamma^A \quad .
\end{equation}
Combining this with the expressions above for the derivatives of the curved space Gamma matrices, we find
\begin{align}
\tilde\nabla_r\Gamma_+
& = (\tilde\nabla_r\tilde\Gamma^r) - (\partial_r\alpha)\tilde\Gamma^u - \alpha\tilde\nabla_r\tilde\Gamma^u - (\partial_r\beta_A)\tilde\Lambda^A - \beta_A\tilde\nabla_r\tilde\Lambda^A \nonumber \\
& = \left(\partial_r\alpha + \frac{1}{2}\beta^A\partial_r\beta_A\right)\Gamma_- + \left(\frac{1}{2}\partial_r\beta_A - \frac{1}{2}\beta^B\partial_r\gamma_{AB}\right)\tilde\Lambda^A \nonumber \\
& - (\partial_r\alpha)\Gamma_- - (\partial_r\beta_A)\tilde\Lambda^A - \frac{1}{2}\beta_A\gamma^{AB}(\partial_r\beta_B)\Gamma_- + \frac{1}{2}\beta_A\gamma^{AB}(\partial_r\gamma_{BC})\tilde\Lambda^C \nonumber \\
& = -\frac{1}{2}(\partial_r\beta_A)\tilde\Lambda^A \quad .
\end{align}
This is equation (\ref{eq:grspinors-gamma-matrices-derivatives-+}) for $n=1$.
We now proceed by induction.
%
Now, suppose it is true for some $n\geq1$, that is,
\begin{align}
\tilde\nabla_r^n\Gamma_+ = -\frac{1}{2}(\partial_r^n\beta_A)\tilde\Lambda^A + \frac{1}{4}\sum_{k=1}^{n-1}\frac{(n-1)!}{j!(n-1-j)!}\gamma^{AB}(\partial_r^j\gamma_{AC})(\partial_r^{n-j}\beta_B)\tilde\Lambda^C + O(r^{d-n}) \quad .
\end{align}
Then application of $\tilde\nabla_r$ yields,
\begin{align}
\tilde\nabla_r^{n+1}\Gamma_+ 
& = -\frac{1}{2}(\partial_r^{n+1}\beta_A)\tilde\Lambda^A + \frac{n}{4}\gamma^{AB}(\partial_{AC})(\partial_r^n\beta_B)\tilde\Lambda^C + \frac{1}{4}\gamma^{AB}(\partial_r^n\gamma_{AC})(\partial_r\beta_B)\tilde\Lambda^B \nonumber \\ 
& + \frac{1}{4}\sum_{k=1}^{n-1}\frac{(n-1)!}{(k+1)!(n-1-k)!}\gamma^{AB}(\partial_r^{k+1}\gamma_{AC})(\partial_r^{n-k}\beta_B)\tilde\Lambda^C \nonumber \\
& = -\frac{1}{2}(\partial_r^{n+1}\beta_A)\tilde\Lambda^A + \frac{1}{4}\sum_{k=1}^n \frac{n!}{j!(n-k)!}\gamma^{AB}(\partial_r^k\gamma_{AC})(\partial_r^{n-k+1}\beta_B)\tilde\Lambda^C + O(r^{d-n-1}) \quad ,
\end{align}
where additional terms arising from differentiating have been omitted since they are $O(r^{d-n-1})$.
This is the desired formula, completing the proof by induction.

\noindent
Equation (\ref{eq:grspinors-gamma-matrices-derivatives-A}) has already been shown for $n=1$.
Now suppose it is true for some $n\geq1$.
That is,
\begin{align}
\tilde\nabla_r^n\tilde\Lambda^A
& = \frac{1}{2}\gamma^{AB}(\partial_r^n\beta_B)\Gamma_- - \frac{1}{2}\gamma^{AB}(\partial_r^n\gamma_{BC})\tilde\Lambda^{C} \nonumber \\
& + \frac{1}{4}\sum_{k=1}^{n-1}a_{n,k}\gamma^{AB}\gamma^{CD}(\partial_r^k\gamma_{BC})(\partial_r^{n-k}\gamma_{DE})\tilde\Lambda^E + O(r^{d-n-1}) \quad .
\end{align}
Applying $\tilde\nabla_r$ gives,
\begin{align}
\tilde\nabla_r^{n+1}\tilde\Lambda^A
& = \frac{1}{2}\gamma^{AB}(\partial_r^{n+1}\beta_B)\Gamma_- - \frac{1}{2}\gamma^{AB}(\partial_r^{n+1}\gamma_{BC})\tilde\Lambda^C + \frac{1}{2}\gamma^{AD}\gamma^{BE}(\partial_r^n\gamma_{BC})(\partial_r\gamma_{DE})\tilde\Lambda^C \nonumber \\
& + \frac{1}{4}\sum_{k=1}^{n-1}a_{n,k}\gamma^{AB}\gamma^{CD}(\partial_r^{k+1}\gamma_{BC})(\partial_r^{n-k}\gamma_{DE})\tilde\Lambda^E \nonumber \\
& + \frac{1}{4}\sum_{k=1}^{n-1}a_{n,k}\gamma^{AB}\gamma^{CD}(\partial_r^k\gamma_{BC})(\partial_r^{n+1-k}\gamma_{DE})\tilde\Lambda^E + O(r^{d-n-2}) \nonumber \\
& = \frac{1}{2}\gamma^{AB}(\partial_r^{n+1}\beta_B)\Gamma_- - \frac{1}{2}\gamma^{AB}(\partial_r^{n+1}\gamma_{BC})\tilde\Lambda^C + \frac{1}{2}\gamma^{AB}\gamma^{CD}(\partial_r\gamma_{BC})(\partial_r^n\gamma_{DE})\tilde\Lambda^E \nonumber\\
& + \frac{1}{4}\sum_{k=2}^{n}a_{n,k-1}\gamma^{AB}\gamma^{CD}(\partial_r^{k}\gamma_{BC})(\partial_r^{n+1-k}\gamma_{DE})\tilde\Lambda^E \nonumber \\
& + \frac{1}{4}\sum_{k=1}^{n-1}a_{n,k}\gamma^{AB}\gamma^{CD}(\partial_r^k\gamma_{BC})(\partial_r^{n+1-k}\gamma_{DE})\tilde\Lambda^E + O(r^{d-n-2}) \nonumber \\
& = \frac{1}{2}\gamma^{AB}(\partial_r^{n+1}\beta_B)\Gamma_- - \frac{1}{2}\gamma^{AB}(\partial_r^{n+1}\gamma_{BC})\tilde\Lambda^C \nonumber \\
& + \frac{1}{4}\sum_{k=1}^{n}a_{n+1,k}\gamma^{AB}\gamma^{CD}(\partial_r^{k}\gamma_{BC})(\partial_r^{n+1-k}\gamma_{DE})\tilde\Lambda^E
\end{align}
This is the desired result and so the proof by induction is complete.
\end{note}

\begin{note}
\label{note:grtensors-lemmas}
The following results are useful in the calculation of the Ricci tensor.
\begin{enumerate}[1)]
\item
\begin{equation}
\partial_a\gamma^{AB} = -\gamma^{AC}\gamma^{BD}\partial_a\gamma_{CD}
\end{equation}
\item
\begin{equation}
\partial_a\beta^A = \gamma^{AB}\partial_a\beta_B - \gamma^{AB}\beta^C\partial_a\gamma_{BC}
\end{equation}
\item
\begin{equation}
\partial_a(\beta_A\beta^A) = 2\beta^A\partial_a\beta_A - \beta^A\beta^B\partial_a\gamma_{AB}
\end{equation}
\item
\begin{equation}
\partial_a\tilde\Lambda^C{}_{CA} = \frac{1}{2}\gamma^{BC} D_A \partial_a \gamma_{BC}
\end{equation}
\end{enumerate}
\proof{These results can be verified by direct calculation.}
\end{note}

%% file: GRTensors.tex
\footnotesize

\section{Tensors of General Relativity in GNCs}
\label{app:grtensors}
Standard tensors of general relativity are given in Gaussian null coordinates (GNCs).

\subsection{Metric Tensor}
\label{app:grtensors-metric}

\begin{align}
\label{eq:metric}
g_{ab} = r^{-2}\tilde{g}_{ab} &=
r^{-2}[2(dr_{(a} - \alpha du_{(a} - \beta_A dx^A_{(a})du_{b)} + \gamma_{AB} dx^A_a dx^B_b]
\\
\label{eq:metric-inverse}
g^{ab} = r^2\tilde{g}^{ab} &= r^2\left[2\left(\left\{\alpha + \frac{1}{2}\beta_C\beta^C\right\}\partial_r^{(a} + \partial_u^{(a} + \beta^A\partial_A^{(a}\right)\partial_r^{b)} + \gamma^{AB}\partial_A^a\partial_B^b\right]
\end{align}
%

\subsection{Christoffel Symbols}
\label{app:grtensors-christoffel}

\begin{equation}
\label{eq:christoffel}
\tilde\Gamma^c{}_{ab} = \frac{1}{2}\tilde{g}^{cd}(\partial_a\tilde{g}_{bd} + \partial_b\tilde{g}_{ad} - \partial_d\tilde{g}_{ab})
\end{equation}

\begin{equation}
\label{eq:christoffel-gamma}
\tilde\Lambda^C{}_{AB} = \frac{1}{2}\gamma^{CD}(\partial_A\gamma_{BD} + \partial_B\gamma_{AD} - \partial_D\gamma_{AB})
\end{equation}


\begin{equation}
\label{eq:christoffel-r-rr}
\tilde\Gamma^r{}_{rr} = 0\\
\end{equation}

\begin{equation}
\label{eq:christoffel-r-ru}
\tilde\Gamma^r{}_{ru} = -\partial_r\alpha - \frac{1}{2}\beta^A\partial_r\beta_A
\end{equation}

\begin{equation}
\label{eq:christoffel-r-rA}
\tilde\Gamma^r{}_{rA} = -\frac{1}{2}\partial_r\beta_A + \frac{1}{2}\beta^B\partial_r\gamma_{AB}
\end{equation}

\begin{equation}
\label{eq:christoffel-r-uu}
\tilde\Gamma^r{}_{uu} = -\partial_u\alpha + (2\alpha + \beta_C\beta^C) \partial_r\alpha + \beta^C(D_C\alpha - \partial_u\beta_C)
\end{equation}

\begin{equation}
\label{eq:christoffel-r-uA}
\tilde\Gamma^r{}_{uA} = -D_A\alpha + \alpha\partial_r\beta_A + \frac{1}{2}\beta^B\partial_u\gamma_{AB} - \beta^B D_{[A}\beta_{B]} + \frac{1}{2}\beta_C\beta^C\partial_r\beta_A
\end{equation}

\begin{equation}
\label{eq:christoffel-r-AB}
\tilde\Gamma^r{}_{AB} = -\frac{1}{2}\partial_u\gamma_{AB} - D_{(A}\beta_{B)} - \alpha\partial_r\gamma_{AB} - \frac{1}{2}\beta_C\beta^C\partial_r\gamma_{AB}
\end{equation}

\begin{equation}
\label{eq:christoffel-u-rr}
\tilde\Gamma^u{}_{rr} = 0
\end{equation}

\begin{equation}
\label{eq:christoffel-u-ru}
\tilde\Gamma^u{}_{ru} = 0
\end{equation}

\begin{equation}
\label{eq:christoffel-u-rA}
\tilde\Gamma^u{}_{rA} = 0
\end{equation}

\begin{equation}
\label{eq:christoffel-u-uu}
\tilde\Gamma^u{}_{uu} = \partial_r\alpha
\end{equation}

\begin{equation}
\label{eq:christoffel-u-uA}
\tilde\Gamma^u{}_{uA} = \frac{1}{2}\partial_r\beta_A
\end{equation}

\begin{equation}
\label{eq:christoffel-u-AB}
\tilde\Gamma^u{}_{AB} = -\frac{1}{2}\partial_r\gamma_{AB}
\end{equation}

\begin{equation}
\label{eq:christoffel-C-rr}
\tilde\Gamma^C{}_{rr} = 0
\end{equation}

\begin{equation}
\label{eq:christoffel-C-ru}
\tilde\Gamma^C{}_{ru} = -\frac{1}{2}\gamma^{CD}\partial_r\beta_D
\end{equation}

\begin{equation}
\label{eq:christoffel-C-rA}
\tilde\Gamma^C{}_{rA} = \frac{1}{2}\gamma^{CD}\partial_r\gamma_{AD}
\end{equation}

\begin{equation}
\label{eq:christoffel-C-uu}
\tilde\Gamma^C{}_{uu} = D^C\alpha - \gamma^{CD}\partial_u\beta_D + \beta^C\partial_r\alpha
\end{equation}

\begin{equation}
\label{eq:christoffel-C-uA}
\tilde\Gamma^C{}_{uA} = \frac{1}{2}\gamma^{CD}\partial_u\gamma_{AD} - \gamma^{CD}D_{[A}\beta_{D]} + \frac{1}{2}\beta^C\partial_r\beta_A
\end{equation}

\begin{equation}
\label{eq:christoffel-C-AB}
\tilde\Gamma^C{}_{AB} = \tilde\Lambda^C{}_{AB} - \frac{1}{2}\beta^C\partial_r\gamma_{AB}
\end{equation}

\subsection{Ricci Tensor}
\label{app:grtensors-ricci}
The results in Note \ref{note:grtensors-lemmas} are used when calculating the Ricci tensor.

\begin{equation}
\label{eq:grtensors-ricci}
\tilde{R}_{ab} = \partial_c\tilde\Gamma^c{}_{ab} - \partial_a\tilde\Gamma^c{}_{cb} + \tilde\Gamma^d{}_{ab}\tilde\Gamma^c{}_{cd} - \tilde\Gamma^d{}_{cb}\tilde\Gamma^c{}_{da}
\end{equation}

\begin{equation}
\label{eq:grtensors-ricci-gamma}
\mathcal{R}_{AB} = \partial_C\Lambda^C{}_{AB} - \partial_A\Lambda^C{}_{CB} + \Lambda^D{}_{AB}\Lambda^C{}_{CD} - \Lambda^D{}_{CB}\Lambda^C{}_{DA}
\end{equation}


\begin{equation}
\label{eq:grtensors-ricci-rr}
\tilde{R}_{rr} = -\frac{1}{2}\gamma^{AB}\partial_r^2\gamma_{AB} + \frac{1}{4}\gamma^{AB}\gamma^{CD}(\partial_r\gamma_{AC})\partial_r\gamma_{BD}
\end{equation}

\begin{align}
\label{eq:grtensors-ricci-ru}
\tilde{R}_{ru} &= -\partial_r^2\alpha - \frac{1}{2}\gamma^{AB}\partial_r\partial_u\gamma_{AB} + \frac{1}{4}\gamma^{AB}\gamma^{CD}(\partial_r\gamma_{AC})\partial_u\gamma_{BD} - \frac{1}{2}D^A\partial_r\beta_A\nonumber\\
& - \frac{1}{2}\gamma^{AB}(\partial_r\beta_A)\partial_r\beta_B - \frac{1}{2}\beta^A\partial_r^2\beta_A - \frac{1}{2}(\partial_r\alpha)\gamma^{AB}\partial_r\gamma_{AB}\nonumber\\
& - \frac{1}{4}\gamma^{AB}\beta^C(\partial_r\gamma_{AB})\partial_r\beta_C + \frac{1}{2}\gamma^{AB}\beta^C(\partial_r\gamma_{AC})\partial_r\beta_B
\end{align}

\begin{align}
\label{eq:grtensors-ricci-rA}
\tilde{R}_{rA} &= -\frac{1}{2}\partial_r^2\beta_A - \gamma^{BC}D_{[A}\partial_r\gamma_{B]C} + \frac{1}{2}\beta^B\partial_r^2\gamma_{AB}\nonumber\\
& + \frac{1}{2}\gamma^{BC}(\partial_r\gamma_{AB})\partial_r\beta_C - \frac{1}{4}(\partial_r\beta_A)\gamma^{BC}\partial_r\gamma_{BC}\nonumber\\
& - \frac{1}{2}\gamma^{BC}\beta^D(\partial_r\gamma_{AB})\partial_r\gamma_{CD} + \frac{1}{4}\gamma^{BC}\beta^D(\partial_r\gamma_{AD})\partial_r\gamma_{BC}
\end{align}

\begin{align}
\label{eq:grtensors-ricci-uu}
\tilde{R}_{uu}
& = D^2\alpha - D^A\partial_u\beta_A + 2\beta^A D_A\partial_r\alpha + (D^A\beta_A)\partial_r\alpha + 2\alpha\partial_r^2\alpha\nonumber\\
& - \frac{1}{2}\gamma^{AB}\partial_u^2\gamma_{AB} + \frac{1}{4}\gamma^{AB}\gamma^{CD}(\partial_u\gamma_{AC})\partial_u\gamma_{BD}\nonumber\\
&  - \frac{1}{2}(\partial_u\alpha)\gamma^{AB}\partial_r\gamma_{AB} + \frac{1}{2}(\partial_r\alpha)\gamma^{AB}\partial_u\gamma_{AB} + \alpha(\partial_r\alpha)\gamma^{AB}\partial_r\gamma_{AB}\nonumber\\
& - \beta^A\partial_r\partial_u\beta_A -(D^A\alpha)\partial_r\beta_A + \alpha\gamma^{AB}(\partial_r\beta_A)\partial_r\beta_B + \gamma^{AB}\gamma^{CD}(D_{[A}\beta_{C]})D_{[B}\beta_{D]}\nonumber\\
& + \beta_A\beta^A\partial_r^2\alpha - (D^A\alpha)\beta^B\partial_r\gamma_{AB} + \frac{1}{2}(\partial_r\alpha)\beta_C\beta^C\gamma^{AB}\partial_r\gamma_{AB} - (\partial_r\alpha)\beta^A\beta^B\partial_r\gamma_{AB}\nonumber\\
& + \frac{1}{2}\beta^C(D_C\alpha)\gamma^{AB}\partial_r\gamma_{AB} - \frac{1}{2}\beta^C(\partial_u\beta_C)\gamma^{AB}\partial_r\gamma_{AB} + \gamma^{AB}\beta^C(\partial_r\gamma_{AC})\partial_u\beta_B\nonumber\\
& - 2\gamma^{AB}\beta^C(D_{[A}\beta_{C]})\partial_r\beta_B - \frac{1}{2}(\beta^A\partial_r\beta_A)^2 + \frac{1}{2}\beta_C\beta^C\gamma^{AB}(\partial_r\beta_A)\partial_r\beta_B
\end{align}

\begin{align}
\label{eq:grtensors-ricci-uA}
\tilde{R}_{uA}
& = -D_A\partial_r\alpha + \alpha\partial_r^2\beta_A - D^B D_{[A}\beta_{B]} + \frac{1}{2}\partial_r\partial_u\beta_A - \gamma^{BC} D_{[A}\partial_u\gamma_{B]C}\nonumber\\
& + \frac{1}{2}\beta^B D_B \partial_r\beta_A + \frac{1}{2} (D^B \beta_B) \partial_r\beta_A - \frac{1}{2}(D_A\alpha)\gamma^{BC}\partial_r\gamma_{BC} + (D^B\alpha)\partial_r\gamma_{AB}\nonumber\\
& + (\partial_r\alpha)\beta^B\partial_r\gamma_{AB} + \frac{1}{2}\alpha\gamma^{BC}(\partial_r\gamma_{BC})\partial_r\beta_A - \alpha\gamma^{BC}(\partial_r\gamma_{AB})\partial_r\beta_C\nonumber\\
& + \frac{1}{4}\gamma^{BC}(\partial_u\gamma_{BC})\partial_r\beta_A - \frac{1}{2}\gamma^{BC}(\partial_r\gamma_{AB})\partial_u\beta_C + \frac{1}{2}\beta^B\partial_r\partial_u\gamma_{AB}\nonumber\\
& - \frac{1}{2}\gamma^{BC}(D_A\beta_B)\partial_r\beta_C - \beta^B\partial_r D_{[A}\beta_{B]} - \frac{1}{2}\gamma^{BC}\beta^D(\partial_u\gamma_{AB})\partial_r\gamma_{CD} \nonumber\\
& + \frac{1}{2}\beta_B\beta^B\partial_r^2\beta_A + \frac{1}{2}\beta^B(\partial_r\beta_B)\partial_r\beta_A + \frac{1}{4}\gamma^{BC}\beta^D(\partial_u\gamma_{AD})\partial_r\gamma_{BC}\nonumber\\
& + \gamma^{BC}\beta^D(D_{[A}\beta_{B]})\partial_r\gamma_{CD} + \gamma^{BC}\beta^D(\partial_r\gamma_{AB})D_{[C}\beta_{D]} - \frac{1}{2}\gamma^{BC}\beta^D(D_{[A}\beta_{D]})\partial_r\gamma_{BC}\nonumber\\
& + \beta^B\beta^C(\partial_r\gamma_{C[A})\partial_r\beta_{B]} + \frac{1}{4}\beta_D\beta^D\gamma^{BC}(\partial_r\gamma_{BC})\partial_r\beta_A - \frac{1}{2}\beta_D\beta^D\gamma^{BC}(\partial_r\gamma_{AB})\partial_r\beta_C\nonumber\\
\end{align}

\begin{align}
\label{eq:grtensors-ricci-AB}
\tilde{R}_{AB} &= -\partial_r\partial_u\gamma_{AB} - \partial_r D_{(A}\beta_{B)} + \mathcal{R}_{AB} - (\partial_r\alpha)\partial_r\gamma_{AB} - \alpha\partial_r^2\gamma_{AB}\nonumber\\
& - \frac{1}{2}(D^C\beta_C)\partial_r\gamma_{AB} - \frac{1}{2}\beta^C D_C\partial_r\gamma_{AB} - \frac{1}{2}(\partial_r\beta_A)\partial_r\beta_B\nonumber\\
& - \frac{1}{4}\gamma^{CD}\left((\partial_u\gamma_{AB})\partial_r\gamma_{CD} + (\partial_r\gamma_{AB})\partial_u\gamma_{CD}\right)\nonumber\\
& + \frac{1}{2}\gamma^{CD}\left((\partial_r\gamma_{AC})\partial_u\gamma_{BD} + (\partial_u\gamma_{AC})\partial_r\gamma_{BD}\right)\nonumber\\
& - \frac{1}{2}\alpha\gamma^{CD}(\partial_r\gamma_{AB})\partial_r\gamma_{CD} + \alpha\gamma^{CD}(\partial_r\gamma_{AC})\partial_r\gamma_{BD}\nonumber\\
& + \beta^C(\partial_r\beta_{(A})\partial_r\gamma_{B)C} - \beta^C(\partial_r\beta_C)\partial_r\gamma_{AB}\nonumber\\
& - \frac{1}{2}\gamma^{CD}(D_{(A}\beta_{B)})\partial_r\gamma_{CD} + (D^C\beta_{(A})\partial_r\gamma_{B)C}\nonumber\\
& - \frac{1}{2}\beta_C\beta^C\partial_r^2\gamma_{AB} + \beta^C\beta^D(\partial_r\gamma_{A[B})\partial_r\gamma_{C]D}\nonumber\\
& - \frac{1}{4}\beta_E\beta^E\gamma^{CD}(\partial_r\gamma_{AB})\partial_r\gamma_{CD} + \frac{1}{2}\beta_E\beta^E\gamma^{CD}(\partial_r\gamma_{AC})\partial_r\gamma_{BD}
\end{align}

\subsection{Ricci Scalar}
\label{app:grtensors-ricci-scalar}

\begin{align}
\label{eq:grtensors-ricci-scalar}
\tilde{R}
& = \mathcal{R} - 2\partial_r^2\alpha - D^A\partial_r\beta_A - \gamma^{AB}\partial_r D_A\beta_B - \frac{1}{2}D_C(\beta^C\gamma^{AB}\partial_r\gamma_{AB})\nonumber\\
& - 2\gamma^{AB}\partial_r\partial_u\gamma_{AB} - 2\alpha\gamma^{AB}\partial_r^2\gamma_{AB} - 2(\partial_r\alpha)\gamma^{AB}\partial_r\gamma_{AB} \nonumber\\
& + \frac{3}{2}\gamma^{AB}\gamma^{CD}(\partial_r\gamma_{AC})\partial_u\gamma_{BD} - \frac{1}{2}\gamma^{AB}\gamma^{CD}(\partial_r\gamma_{AB})\partial_u\gamma_{CD}\nonumber\\
& + \frac{3}{2}\alpha\gamma^{AB}\gamma^{CD}(\partial_r\gamma_{AC})\partial_r\gamma_{BD} - \frac{1}{2}\alpha\gamma^{AB}\gamma^{CD}(\partial_r\gamma_{AB})\partial_r\gamma_{CD}\nonumber\\
& + \gamma^{AB}\gamma^{CD}(\partial_r\gamma_{AC})D_{(B}\beta_{D)} - \frac{1}{2}(D^C\beta_C)\gamma^{AB}\partial_r\gamma_{AB}\nonumber\\
& - \frac{3}{2}\gamma^{AB}(\partial_r\beta_A)\partial_r\beta_B - 2\beta^A\partial_r^2\beta_A - 2\gamma^{AB}\beta^C D_{[C}\partial_r\gamma_{A]B}\nonumber\\
& + 3\gamma^{AB}\beta^C(\partial_r\gamma_{AC})\partial_r\beta_B - 2\gamma^{AB}\beta^C(\partial_r\gamma_{AB})\partial_r\beta_C\nonumber\\
& + \beta^A\beta^B\partial_r^2\gamma_{AB} - \beta_C\beta^C\gamma^{AB}\partial_r^2\gamma_{AB}\nonumber\\
& - \frac{3}{2}\gamma^{AB}\beta^C\beta^D(\partial_r\gamma_{AC})\partial_r\gamma_{BD} + \gamma^{AB}\beta^C\beta^D(\partial_r\gamma_{AB})\partial_r\gamma_{CD}\nonumber\\
& + \frac{3}{4}\beta_E\beta^E\gamma^{AB}\gamma^{CD}(\partial_r\gamma_{AC})\partial_r\gamma_{BD} - \frac{1}{4}\beta_E\beta^E\gamma^{AB}\gamma^{CD}(\partial_r\gamma_{AB})\partial_r\gamma_{CD}
\end{align}

\subsection{Schouten Tensor}
\label{app:grtensors-schouten}

\begin{equation}
\label{eq:grtensors-schouten}
\tilde{S}_{ab} \equiv \frac{2}{d-2}\tilde{R}_{ab} - \frac{1}{(d-1)(d-2)}\tilde{R}\tilde{g}_{ab}
\end{equation}

\begin{equation}
\label{eq:grtensors-schouten-compact-rr}
\tilde{S}_{rr} = \frac{2}{d-2}\tilde{R}_{rr}
\end{equation}

\begin{equation}
\label{eq:grtensors-schouten-compact-ru}
\tilde{S}_{ru} = \frac{1}{(d-1)(d-2)}[2(d-1)\tilde{R}_{ru} - \tilde{R}]
\end{equation}

\begin{equation}
\label{eq:grtensors-schouten-compact-rA}
\tilde{S}_{rA} = \frac{2}{d-2}\tilde{R}_{rA}
\end{equation}

\begin{equation}
\label{eq:grtensors-schouten-compact-uu}
\tilde{S}_{uu} = \frac{1}{(d-1)(d-2)}[2(d-1)\tilde{R}_{uu} + 2\alpha\tilde{R}]
\end{equation}

\begin{equation}
\label{eq:grtensors-schouten-compact-uA}
\tilde{S}_{uA} = \frac{1}{(d-1)(d-2)}[2(d-1)\tilde{R}_{uA} + \beta_A\tilde{R}]
\end{equation}

\begin{equation}
\label{eq:grtensors-schouten-compact-AB}
\tilde{S}_{AB} = \frac{1}{(d-1)(d-2)}[2(d-1)\tilde{R}_{AB} - \gamma_{AB}\tilde{R}]
\end{equation}

\noindent
For brevity we define
\begin{equation*}
D \equiv (d-1)(d-2)
\end{equation*}


\begin{equation}
\label{eq:grtensors-schouten-rr}
\tilde{S}_{rr} = \frac{1}{d-2}\bigg[
 - \gamma^{AB}\partial_r^2\gamma_{AB} + \frac{1}{2}\gamma^{AB}\gamma^{CD}(\partial_r\gamma_{AC})\partial_r\gamma_{BD}
\bigg]
\end{equation}

\begin{align}
\label{eq:grtensors-schouten-ru}
\tilde{S}_{ru} = \frac{1}{D} \bigg[
& - \mathcal{R} - 2(d-2)\partial_r^2\alpha - (d-2)D^A\partial_r\beta_A + \gamma^{AB}\partial_r D_A\beta_B + \frac{1}{2}D_C(\beta^C\gamma^{AB}\partial_r\gamma_{AB})\nonumber\\
& - (d-3)\gamma^{AB}\partial_r\partial_u\gamma_{AB} + 2\alpha\gamma^{AB}\partial_r^2\gamma_{AB} - (d-3)(\partial_r\alpha)\gamma^{AB}\partial_r\gamma_{AB} \nonumber\\
& + \left(\frac{d}{2} - 2\right)\gamma^{AB}\gamma^{CD}(\partial_r\gamma_{AC})\partial_u\gamma_{BD} + \frac{1}{2}\gamma^{AB}\gamma^{CD}(\partial_r\gamma_{AB})\partial_u\gamma_{CD}\nonumber\\
& - \frac{3}{2}\alpha\gamma^{AB}\gamma^{CD}(\partial_r\gamma_{AC})\partial_r\gamma_{BD} + \frac{1}{2}\alpha\gamma^{AB}\gamma^{CD}(\partial_r\gamma_{AB})\partial_r\gamma_{CD}\nonumber\\
& - \gamma^{AB}\gamma^{CD}(\partial_r\gamma_{AC})D_{(B}\beta_{D)} + \frac{1}{2}(D^C\beta_C)\gamma^{AB}\partial_r\gamma_{AB}\nonumber\\
& - \left(d - \frac{5}{2}\right)\gamma^{AB}(\partial_r\beta_A)\partial_r\beta_B - (d-3)\beta^A\partial_r^2\beta_A + 2\gamma^{AB}\beta^C D_{[C}\partial_r\gamma_{A]B}\nonumber\\
& + (d-4)\gamma^{AB}\beta^C(\partial_r\gamma_{AC})\partial_r\beta_B - \left(\frac{d}{2} - \frac{5}{2}\right)\gamma^{AB}\beta^C(\partial_r\gamma_{AB})\partial_r\beta_C\nonumber\\
& - \beta^A\beta^B\partial_r^2\gamma_{AB} + \beta_C\beta^C\gamma^{AB}\partial_r^2\gamma_{AB}\nonumber\\
& + \frac{3}{2}\gamma^{AB}\beta^C\beta^D(\partial_r\gamma_{AC})\partial_r\gamma_{BD} - \gamma^{AB}\beta^C\beta^D(\partial_r\gamma_{AB})\partial_r\gamma_{CD}\nonumber\\
& - \frac{3}{4}\beta_E\beta^E\gamma^{AB}\gamma^{CD}(\partial_r\gamma_{AC})\partial_r\gamma_{BD} + \frac{1}{4}\beta_E\beta^E\gamma^{AB}\gamma^{CD}(\partial_r\gamma_{AB})\partial_r\gamma_{CD}
\bigg]
\end{align}

\begin{align}
\label{eq:grtensors-schouten-rA}
\tilde{S}_{rA} = \frac{1}{d-2}\bigg[
& - \partial_r^2\beta_A - 2\gamma^{BC}D_{[A}\partial_r\gamma_{B]C} + \beta^B\partial_r^2\gamma_{AB}\nonumber\\
& + \gamma^{BC}(\partial_r\gamma_{AB})\partial_r\beta_C - \frac{1}{2}(\partial_r\beta_A)\gamma^{BC}\partial_r\gamma_{BC}\nonumber\\
& - \gamma^{BC}\beta^D(\partial_r\gamma_{AB})\partial_r\gamma_{CD} + \frac{1}{2}\gamma^{BC}\beta^D(\partial_r\gamma_{AD})\partial_r\gamma_{BC}
\bigg]
\end{align}

\begin{align}
\label{eq:grtensors-schouten-uu}
\tilde{S}_{uu} = \frac{1}{D} \bigg[
& \ 2(d-1)D^2\alpha - 2(d-1)D^A\partial_u\beta_A + 4(d-1)\beta^A D_A\partial_r\alpha\nonumber\\
& + 2\alpha\mathcal{R} - 2\alpha D^A\partial_r\beta_A - 2\alpha\gamma^{AB}\partial_r D_A\beta_B - \alpha D_C(\beta^C\gamma^{AB}\partial_r\gamma_{AB})\nonumber\\
& + 2(d-1)(D^A\beta_A)\partial_r\alpha + 4(d-2)\alpha\partial_r^2\alpha - (d-1)\gamma^{AB}\partial_u^2\gamma_{AB}\nonumber\\
& + \frac{d-1}{2}\gamma^{AB}\gamma^{CD}(\partial_u\gamma_{AC})\partial_u\gamma_{BD} - (d-1)(\partial_u\alpha)\gamma^{AB}\partial_r\gamma_{AB}\nonumber\\
& - 4\alpha\gamma^{AB}\partial_r\partial_u\gamma_{AB} - 4\alpha^2\gamma^{AB}\partial_r^2\gamma_{AB} \nonumber\\
& + (d-1)(\partial_r\alpha)\gamma^{AB}\partial_u\gamma_{AB} + 2(d-3)\alpha(\partial_r\alpha)\gamma^{AB}\partial_r\gamma_{AB}\nonumber\\
& + 3\alpha\gamma^{AB}\gamma^{CD}(\partial_r\gamma_{AC})\partial_u\gamma_{BD} - \alpha\gamma^{AB}\gamma^{CD}(\partial_r\gamma_{AB})\partial_u\gamma_{CD}\nonumber\\
& + 3\alpha^2\gamma^{AB}\gamma^{CD}(\partial_r\gamma_{AC})\partial_r\gamma_{BD} - \alpha^2\gamma^{AB}\gamma^{CD}(\partial_r\gamma_{AB})\partial_r\gamma_{CD}\nonumber\\
& + 2\alpha\gamma^{AB}\gamma^{CD}(\partial_r\gamma_{AC})D_{(B}\beta_{D)} - \alpha(D^C\beta_C)\gamma^{AB}\partial_r\gamma_{AB}\nonumber\\
& - 2(d-1)\beta^A\partial_r\partial_u\beta_A - 2(d-1)(D^A\alpha)\partial_r\beta_A\nonumber\\
& + (2d-5)\alpha\gamma^{AB}(\partial_r\beta_A)\partial_r\beta_B + 2(d-1)\gamma^{AB}\gamma^{CD}(D_{[A}\beta_{C]})D_{[B}\beta_{D]}\nonumber\\
& - 4\alpha\beta^A\partial_r^2\beta_A - 4\alpha\gamma^{AB}\beta^C D_{[C}\partial_r\gamma_{A]B}\nonumber\\
& + 2(d-1)\beta_A\beta^A\partial_r^2\alpha - 2(d-1)(D^A\alpha)\beta^B\partial_r\gamma_{AB}\nonumber\\
& + 6\alpha\gamma^{AB}\beta^C(\partial_r\gamma_{AC})\partial_r\beta_B - 4\alpha\gamma^{AB}\beta^C(\partial_r\gamma_{AB})\partial_r\beta_C\nonumber\\
& + 2\alpha\beta^A\beta^B\partial_r^2\gamma_{AB} - 2\alpha\beta_C\beta^C\gamma^{AB}\partial_r^2\gamma_{AB}\nonumber\\
& + (d-1)(\partial_r\alpha)\beta_C\beta^C\gamma^{AB}\partial_r\gamma_{AB} - 2(d-1)(\partial_r\alpha)\beta^A\beta^B\partial_r\gamma_{AB}\nonumber\\
& + (d-1)\beta^C(D_C\alpha)\gamma^{AB}\partial_r\gamma_{AB} - (d-1)\beta^C(\partial_u\beta_C)\gamma^{AB}\partial_r\gamma_{AB}\nonumber\\
& + 2(d-1)\gamma^{AB}\beta^C(\partial_r\gamma_{AC})\partial_u\beta_B - 4(d-1)\gamma^{AB}\beta^C(D_{[A}\beta_{C]})\partial_r\beta_B\nonumber\\
& - 3\alpha\gamma^{AB}\beta^C\beta^D(\partial_r\gamma_{AC})\partial_r\gamma_{BD} + 2\alpha\gamma^{AB}\beta^C\beta^D(\partial_r\gamma_{AB})\partial_r\gamma_{CD}\nonumber\\
& + \frac{3}{2}\alpha\beta_E\beta^E\gamma^{AB}\gamma^{CD}(\partial_r\gamma_{AC})\partial_r\gamma_{BD} - \frac{1}{2}\alpha\beta_E\beta^E\gamma^{AB}\gamma^{CD}(\partial_r\gamma_{AB})\partial_r\gamma_{CD}\nonumber\\
& - (d-1)(\beta^A\partial_r\beta_A)^2 + (d-1)\beta_C\beta^C\gamma^{AB}(\partial_r\beta_A)\partial_r\beta_B
\bigg]
\end{align}

\begin{align}
\label{eq:grtensors-schouten-uA}
\tilde{S}_{uA} = \frac{1}{D} \bigg[
& - 2(d-1)D_A\partial_r\alpha + 2(d-1)\alpha\partial_r^2\beta_A - 2(d-1)D^B D_{[A}\beta_{B]} + (d-1)\partial_r\partial_u\beta_A\nonumber\\
& + \beta_A\mathcal{R} - 2(d-1)\gamma^{BC} D_{[A}\partial_u\gamma_{B]C} + (d-1)\beta^B D_B \partial_r\beta_A + (d-1)(D^B \beta_B) \partial_r\beta_A\nonumber\\
& - 2\beta_A\partial_r^2\alpha - \beta_AD^B\partial_r\beta_B - \beta_A\gamma^{BC}\partial_r D_B\beta_C - \frac{1}{2}\beta_A D_D(\beta^D\gamma^{BC}\partial_r\gamma_{BC})\nonumber\\
& - (d-1)(D_A\alpha)\gamma^{BC}\partial_r\gamma_{BC} + 2(d-1)(D^B\alpha)\partial_r\gamma_{AB} + 2(d-1)(\partial_r\alpha)\beta^B\partial_r\gamma_{AB}\nonumber\\
& + (d-1)\alpha\gamma^{BC}(\partial_r\gamma_{BC})\partial_r\beta_A - 2(d-1)\alpha\gamma^{BC}(\partial_r\gamma_{AB})\partial_r\beta_C\nonumber\\
& + \frac{d-1}{2}\gamma^{BC}(\partial_u\gamma_{BC})\partial_r\beta_A - (d-1)\gamma^{BC}(\partial_r\gamma_{AB})\partial_u\beta_C + (d-1)\beta^B\partial_r\partial_u\gamma_{AB}\nonumber\\
& - 2\beta_A\gamma^{BC}\partial_r\partial_u\gamma_{BC} - 2\alpha\beta_A\gamma^{BC}\partial_r^2\gamma_{BC} - 2(\partial_r\alpha)\beta_A\gamma^{BC}\partial_r\gamma_{BC} \nonumber\\
& + \frac{3}{2}\beta_A\gamma^{BC}\gamma^{DE}(\partial_r\gamma_{BD})\partial_u\gamma_{CE} - \frac{1}{2}\beta_A\gamma^{BC}\gamma^{DE}(\partial_r\gamma_{BC})\partial_u\gamma_{DE}\nonumber\\
& + \frac{3}{2}\alpha\beta_A\gamma^{BC}\gamma^{DE}(\partial_r\gamma_{BD})\partial_r\gamma_{CE} - \frac{1}{2}\alpha\beta_A\gamma^{BC}\gamma^{DE}(\partial_r\gamma_{BC})\partial_r\gamma_{DE}\nonumber\\
& - (d-1)\gamma^{BC}(D_A\beta_B)\partial_r\beta_C - 2(d-1)\beta^B\partial_r D_{[A}\beta_{B]} - (d-1)\gamma^{BC}\beta^D(\partial_u\gamma_{AB})\partial_r\gamma_{CD} \nonumber\\
& + \beta_A\gamma^{BC}\gamma^{DE}(\partial_r\gamma_{BD})D_{(C}\beta_{E)} - \frac{1}{2}\beta_A(D^D\beta_D)\gamma^{BC}\partial_r\gamma_{BC}\nonumber\\
& + (d-1)\beta_B\beta^B\partial_r^2\beta_A + (d-1)\beta^B(\partial_r\beta_B)\partial_r\beta_A + \frac{d-1}{2}\gamma^{BC}\beta^D(\partial_u\gamma_{AD})\partial_r\gamma_{BC}\nonumber\\
& - \frac{3}{2}\beta_A\gamma^{BC}(\partial_r\beta_B)\partial_r\beta_C - 2\beta_A\beta^B\partial_r^2\beta_B - 2\beta_A\gamma^{BC}\beta^D D_{[D}\partial_r\gamma_{B]C}\nonumber\\
& + 2(d-1)\gamma^{BC}\beta^D(D_{[A}\beta_{B]})\partial_r\gamma_{CD} + 2(d-1)\gamma^{BC}\beta^D(\partial_r\gamma_{AB})D_{[C}\beta_{D]}\nonumber\\
& - (d-1)\gamma^{BC}\beta^D(D_{[A}\beta_{D]})\partial_r\gamma_{BC}\nonumber\\
& + 3\beta_A\gamma^{BC}\beta^D(\partial_r\gamma_{BD})\partial_r\beta_C - 2\beta_A\gamma^{BC}\beta^D(\partial_r\gamma_{BC})\partial_r\beta_D\nonumber\\
& + \beta_A\beta^B\beta^C\partial_r^2\gamma_{BC} - \beta_A\beta_D\beta^D\gamma^{BC}\partial_r^2\gamma_{BC} + 2(d-1)\beta^B\beta^C(\partial_r\gamma_{C[A})\partial_r\beta_{B]}\nonumber\\
& + \frac{d-1}{2}\beta_D\beta^D\gamma^{BC}(\partial_r\gamma_{BC})\partial_r\beta_A - (d-1)\beta_D\beta^D\gamma^{BC}(\partial_r\gamma_{AB})\partial_r\beta_C\nonumber\\
& - \frac{3}{2}\beta_A\gamma^{BC}\beta^D\beta^E(\partial_r\gamma_{BD})\partial_r\gamma_{CE} + \beta_A\gamma^{BC}\beta^D\beta^E(\partial_r\gamma_{BC})\partial_r\gamma_{DE}\nonumber\\
& + \frac{3}{4}\beta_A\beta_F\beta^F\gamma^{BC}\gamma^{DE}(\partial_r\gamma_{BD})\partial_r\gamma_{CE} - \frac{1}{4}\beta_A\beta_F\beta^F\gamma^{BC}\gamma^{DE}(\partial_r\gamma_{BC})\partial_r\gamma_{DE}\bigg]
\end{align}

\begin{align}
\label{eq:grtensors-schouten-AB}
\tilde{S}_{AB} = \frac{1}{D} \bigg[
& - 2(d-1)\partial_r\partial_u\gamma_{AB} - 2(d-1)\partial_r D_{(A}\beta_{B)} + 2(d-1)\mathcal{R}_{AB}\nonumber\\
& - 2(d-1)(\partial_r\alpha)\partial_r\gamma_{AB} - 2(d-1)\alpha\partial_r^2\gamma_{AB}\nonumber\\
& - (d-1)(D^C\beta_C)\partial_r\gamma_{AB} - (d-1)\beta^C D_C\partial_r\gamma_{AB} - (d-1)(\partial_r\beta_A)\partial_r\beta_B\nonumber\\
& - \gamma_{AB}\mathcal{R} + 2\gamma_{AB}\partial_r^2\alpha + \gamma_{AB}D^C\partial_r\beta_C + \gamma_{AB}\gamma^{CD}\partial_r D_C\beta_D + \frac{1}{2}\gamma_{AB}D_E(\beta^E\gamma^{CD}\partial_r\gamma_{CD})\nonumber\\
& + 2\gamma_{AB}\gamma^{CD}\partial_r\partial_u\gamma_{CD} + 2\alpha\gamma_{AB}\gamma^{CD}\partial_r^2\gamma_{CD} + 2(\partial_r\alpha)\gamma_{AB}\gamma^{CD}\partial_r\gamma_{CD} \nonumber\\
& - \frac{d-1}{2}\gamma^{CD}\left((\partial_u\gamma_{AB})\partial_r\gamma_{CD} + (\partial_r\gamma_{AB})\partial_u\gamma_{CD}\right)\nonumber\\
& + (d-1)\gamma^{CD}\left((\partial_r\gamma_{AC})\partial_u\gamma_{BD} + (\partial_u\gamma_{AC})\partial_r\gamma_{BD}\right)\nonumber\\
& - \frac{3}{2}\gamma_{AB}\gamma^{CD}\gamma^{EF}(\partial_r\gamma_{CE})\partial_u\gamma_{DF} + \frac{1}{2}\gamma_{AB}\gamma^{CD}\gamma^{EF}(\partial_r\gamma_{CD})\partial_u\gamma_{EF}\nonumber\\
& - (d-1)\alpha\gamma^{CD}(\partial_r\gamma_{AB})\partial_r\gamma_{CD} + 2(d-1)\alpha\gamma^{CD}(\partial_r\gamma_{AC})\partial_r\gamma_{BD}\nonumber\\
& - \frac{3}{2}\alpha\gamma_{AB}\gamma^{CD}\gamma^{EF}(\partial_r\gamma_{CE})\partial_r\gamma_{DF} + \frac{1}{2}\alpha\gamma_{AB}\gamma^{CD}\gamma^{EF}(\partial_r\gamma_{CD})\partial_r\gamma_{EF}\nonumber\\
& + 2(d-1)\beta^C(\partial_r\beta_{(A})\partial_r\gamma_{B)C} - 2(d-1)\beta^C(\partial_r\beta_C)\partial_r\gamma_{AB}\nonumber\\
& - (d-1)\gamma^{CD}(D_{(A}\beta_{B)})\partial_r\gamma_{CD} + 2(d-1)(D^C\beta_{(A}\partial_r\gamma_{B)C}\nonumber\\
& - \gamma_{AB}\gamma^{CD}\gamma^{EF}(\partial_r\gamma_{CE})D_{(D}\beta_{F)} + \frac{1}{2}\gamma_{AB}(D^E\beta_E)\gamma^{CD}\partial_r\gamma_{CD}\nonumber\\
& + \frac{3}{2}\gamma_{AB}\gamma^{CD}(\partial_r\beta_C)\partial_r\beta_D + 2\gamma_{AB}\beta^C\partial_r^2\beta_C + 2\gamma_{AB}\gamma^{CD}\beta^E D_{[E}\partial_r\gamma_{C]D}\nonumber\\
& - (d-1)\beta_C\beta^C\partial_r^2\gamma_{AB} + 2(d-1)\beta^C\beta^D(\partial_r\gamma_{A[B})\partial_r\gamma_{C]D}\nonumber\\
& - 3\gamma_{AB}\gamma^{CD}\beta^E(\partial_r\gamma_{CE})\partial_r\beta_D + 2\gamma_{AB}\gamma^{CD}\beta^E(\partial_r\gamma_{CD})\partial_r\beta_E\nonumber\\
& - \gamma_{AB}\beta^C\beta^D\partial_r^2\gamma_{CD} + \gamma_{AB}\beta_E\beta^E\gamma^{CD}\partial_r^2\gamma_{CD}\nonumber\\
& - \frac{d-1}{2}\beta_E\beta^E\gamma^{CD}(\partial_r\gamma_{AB})\partial_r\gamma_{CD} + (d-1)\beta_E\beta^E\gamma^{CD}(\partial_r\gamma_{AC})\partial_r\gamma_{BD}\nonumber\\
& + \frac{3}{2}\gamma_{AB}\gamma^{CD}\beta^E\beta^F(\partial_r\gamma_{CE})\partial_r\gamma_{DF} - \gamma_{AB}\gamma^{CD}\beta^E\beta^F(\partial_r\gamma_{CD})\partial_r\gamma_{EF}\nonumber\\
& - \frac{3}{4}\gamma_{AB}\beta_G\beta^G\gamma^{CD}\gamma^{EF}(\partial_r\gamma_{CE})\partial_r\gamma_{DF} + \frac{1}{4}\gamma_{AB}\beta_G\beta^G\gamma^{CD}\gamma^{EF}(\partial_r\gamma_{CD})\partial_r\gamma_{EF}
\bigg]
\end{align}

\subsection{Einstein Equations}
\label{app:grtensors-einstein}

The Einstein equations are given by
\begin{equation}
0 = \tilde{E}_{ab} = \tilde{S}_{ab} + 2r^{-1}\tilde\nabla_a\tilde\nabla_b r - r^{-2}\tilde{g}_{ab}\tilde{g}^{cd}(\tilde\nabla_c r)\tilde\nabla_d r \quad .
\end{equation}

To express this in coordinate form, we note that
\begin{align}
& \tilde\nabla_a\tilde\nabla_b r = \tilde\nabla_a \partial_b r = \partial_a\partial_b r - \tilde{\Gamma}^c{}_{ab} \partial_c r = -\tilde{\Gamma}^r{}_{ab} \quad , \\
& \tilde{g}^{cd}(\tilde\nabla_c r)\tilde\nabla_d r = \tilde{g}^{rr} = 2\alpha + \beta_C\beta^C \quad .
\end{align}
We can therefore rewrite the vacuum Einstein equation as
\begin{equation}
\label{eq:grtensors-einstein}
0 = \tilde{E}_{ab} \equiv \tilde{S}_{ab} - 2r^{-1}\tilde\Gamma^r{}_{ab} - r^{-2}(2\alpha + \beta_C\beta^C)\tilde{g}_{ab}
\end{equation}

The components are given by,
\begin{equation}
\label{eq:grtensors-einstein-compact-rr}
0 = \tilde{E}_{rr} = \tilde{S}_{rr} \quad ,
\end{equation}

\begin{equation}
\label{eq:grtensors-einstein-compact-ru}
0 = \tilde{E}_{ru} = \tilde{S}_{ru} + 2r^{-1}(\partial_r\alpha - r^{-1}\alpha) + r^{-1}\beta^A(\partial_r\beta_A - r^{-1}\beta_A) \quad ,
\end{equation}

\begin{equation}
\label{eq:grtensors-einstein-compact-rA}
0 = \tilde{E}_{rA} = \tilde{S}_{rA} + r^{-1}\partial_r\beta_A - r^{-1}\beta^B\partial_r\gamma_{AB} \quad ,
\end{equation}

\begin{equation}
\label{eq:grtensors-einstein-compact-uu}
0 = \tilde{E}_{uu} = \tilde{S}_{uu} + 2r^{-1}\partial_u\alpha - 2r^{-1}(\partial_r\alpha - r^{-1}\alpha)(2\alpha + \beta_A\beta^A) - 2r^{-1}\beta^A(D_A\alpha - \partial_u\beta_A) \quad ,
\end{equation}

\begin{equation}
\label{eq:grtensors-einstein-compact-uA}
0 = \tilde{E}_{uA} = \tilde{S}_{uA} + 2r^{-1}D_A\alpha - r^{-1}(\partial_r\beta_A - r^{-1}\beta_A)(2\alpha + \beta_B\beta^B) - r^{-1}\beta^B\partial_u\gamma_{AB} + 2r^{-1}\beta^B D_{[A}\beta_{B]} \quad ,
\end{equation}

\begin{equation}
\label{eq:grtensors-einstein-compact-AB}
0 = \tilde{E}_{AB} = \tilde{S}_{AB} + r^{-1}\partial_u\gamma_{AB} + 2r^{-1}D_{(A}\beta_{B)} + r^{-1}(\partial_r\gamma_{AB} - r^{-1}\gamma_{AB})(2\alpha + \beta_C\beta^C) \quad .
\end{equation}

The components of the Einstein tensor $\tilde{E}_{ab}$ are now written out in full.

\begin{equation}
\label{eq:grtensors-einstein-rr}
\tilde{E}_{rr} = \frac{1}{d-2}\bigg[
- \gamma^{AB}\partial_r^2\gamma_{AB} + \frac{1}{2}\gamma^{AB}\gamma^{CD}(\partial_r\gamma_{AC})\partial_r\gamma_{BD}
\bigg]
\end{equation}

\begin{align}
\label{eq:grtensors-einstein-ru}
\tilde{E}_{ru} = \frac{1}{D} \bigg[
& - \mathcal{R} - 2(d-2)\partial_r^2\alpha + 2(d-1)(d-2)r^{-1}(\partial_r\alpha - r^{-1}\alpha)\nonumber\\
& - (d-2)D^A\partial_r\beta_A + \gamma^{AB}\partial_r D_A\beta_B + \frac{1}{2}D_C(\beta^C\gamma^{AB}\partial_r\gamma_{AB})\nonumber\\
& - (d-3)\gamma^{AB}\partial_r\partial_u\gamma_{AB} + 2\alpha\gamma^{AB}\partial_r^2\gamma_{AB} - (d-3)(\partial_r\alpha)\gamma^{AB}\partial_r\gamma_{AB} \nonumber\\
& + \left(\frac{d}{2} - 2\right)\gamma^{AB}\gamma^{CD}(\partial_r\gamma_{AC})\partial_u\gamma_{BD} + \frac{1}{2}\gamma^{AB}\gamma^{CD}(\partial_r\gamma_{AB})\partial_u\gamma_{CD}\nonumber\\
& - \frac{3}{2}\alpha\gamma^{AB}\gamma^{CD}(\partial_r\gamma_{AC})\partial_r\gamma_{BD} + \frac{1}{2}\alpha\gamma^{AB}\gamma^{CD}(\partial_r\gamma_{AB})\partial_r\gamma_{CD}\nonumber\\
& - \gamma^{AB}\gamma^{CD}(\partial_r\gamma_{AC})D_{(B}\beta_{D)} + \frac{1}{2}(D^C\beta_C)\gamma^{AB}\partial_r\gamma_{AB}\nonumber\\
& - \left(d - \frac{5}{2}\right)\gamma^{AB}(\partial_r\beta_A)\partial_r\beta_B - (d-3)\beta^A\partial_r^2\beta_A\nonumber\\
& + (d-1)(d-2)r^{-1}\beta^A(\partial_r\beta_A - r^{-1}\beta_A) + 2\gamma^{AB}\beta^C D_{[C}\partial_r\gamma_{A]B}\nonumber\\
& + (d-4)\gamma^{AB}\beta^C(\partial_r\gamma_{AC})\partial_r\beta_B - \left(\frac{d}{2} - \frac{5}{2}\right)\gamma^{AB}\beta^C(\partial_r\gamma_{AB})\partial_r\beta_C\nonumber\\
& - \beta^A\beta^B\partial_r^2\gamma_{AB} + \beta_C\beta^C\gamma^{AB}\partial_r^2\gamma_{AB}\nonumber\\
& + \frac{3}{2}\gamma^{AB}\beta^C\beta^D(\partial_r\gamma_{AC})\partial_r\gamma_{BD} - \gamma^{AB}\beta^C\beta^D(\partial_r\gamma_{AB})\partial_r\gamma_{CD}\nonumber\\
& - \frac{3}{4}\beta_E\beta^E\gamma^{AB}\gamma^{CD}(\partial_r\gamma_{AC})\partial_r\gamma_{BD} + \frac{1}{4}\beta_E\beta^E\gamma^{AB}\gamma^{CD}(\partial_r\gamma_{AB})\partial_r\gamma_{CD}
\bigg]
\end{align}

\begin{align}
\label{eq:grtensors-einstein-rA}
\tilde{E}_{rA} = \frac{1}{d-2} \bigg[
& - \partial_r^2\beta_A + (d-2)r^{-1}\partial_r\beta_A - 2\gamma^{BC}D_{[A}\partial_r\gamma_{B]C}\nonumber\\
& + \beta^B\partial_r^2\gamma_{AB} - (d-2)r^{-1}\beta^B\partial_r\gamma_{AB}\nonumber\\
& + \gamma^{BC}(\partial_r\gamma_{AB})\partial_r\beta_C - \frac{1}{2}(\partial_r\beta_A)\gamma^{BC}\partial_r\gamma_{BC}\nonumber\\
& - \gamma^{BC}\beta^D(\partial_r\gamma_{AB})\partial_r\gamma_{CD} + \frac{1}{2}\gamma^{BC}\beta^D(\partial_r\gamma_{AD})\partial_r\gamma_{BC}
\bigg]
\end{align}

\begin{align}
\label{eq:grtensors-einstein-uu}
\tilde{E}_{uu} = \frac{1}{D} \bigg[
& \ 2(d-1)(d-2)r^{-1}\partial_u\alpha + 2(d-1)D^2\alpha - 2(d-1)D^A\partial_u\beta_A\nonumber\\
& + 4(d-1)\beta^A D_A\partial_r\alpha - 2(d-1)(d-2)r^{-1}\beta^A D_A\alpha\nonumber\\
& + 2\alpha\mathcal{R} - 2\alpha D^A\partial_r\beta_A - 2\alpha\gamma^{AB}\partial_r D_A\beta_B - \alpha D_C(\beta^C\gamma^{AB}\partial_r\gamma_{AB})\nonumber\\
& + 4(d-2)\alpha\partial_r^2\alpha - 4(d-1)(d-2)r^{-1}\alpha(\partial_r\alpha - r^{-1}\alpha)\nonumber\\
& + 2(d-1)(D^A\beta_A)\partial_r\alpha - (d-1)\gamma^{AB}\partial_u^2\gamma_{AB}\nonumber\\
& + \frac{d-1}{2}\gamma^{AB}\gamma^{CD}(\partial_u\gamma_{AC})\partial_u\gamma_{BD} - (d-1)(\partial_u\alpha)\gamma^{AB}\partial_r\gamma_{AB}\nonumber\\
& - 4\alpha\gamma^{AB}\partial_r\partial_u\gamma_{AB} - 4\alpha^2\gamma^{AB}\partial_r^2\gamma_{AB} \nonumber\\
& + (d-1)(\partial_r\alpha)\gamma^{AB}\partial_u\gamma_{AB} + 2(d-3)\alpha(\partial_r\alpha)\gamma^{AB}\partial_r\gamma_{AB}\nonumber\\
& + 3\alpha\gamma^{AB}\gamma^{CD}(\partial_r\gamma_{AC})\partial_u\gamma_{BD} - \alpha\gamma^{AB}\gamma^{CD}(\partial_r\gamma_{AB})\partial_u\gamma_{CD}\nonumber\\
& + 3\alpha^2\gamma^{AB}\gamma^{CD}(\partial_r\gamma_{AC})\partial_r\gamma_{BD} - \alpha^2\gamma^{AB}\gamma^{CD}(\partial_r\gamma_{AB})\partial_r\gamma_{CD}\nonumber\\
& + 2\alpha\gamma^{AB}\gamma^{CD}(\partial_r\gamma_{AC})D_{(B}\beta_{D)} - \alpha(D^C\beta_C)\gamma^{AB}\partial_r\gamma_{AB}\nonumber\\
& - 2(d-1)\beta^A\partial_r\partial_u\beta_A + 2(d-1)(d-2)r^{-1}\beta^A\partial_u\beta_A - 2(d-1)(D^A\alpha)\partial_r\beta_A\nonumber\\
& + (2d-5)\alpha\gamma^{AB}(\partial_r\beta_A)\partial_r\beta_B + 2(d-1)\gamma^{AB}\gamma^{CD}(D_{[A}\beta_{C]})D_{[B}\beta_{D]}\nonumber\\
& - 4\alpha\beta^A\partial_r^2\beta_A - 2(d-1)(d-2)r^{-1}\beta_A\beta^A(\partial_r\alpha - r^{-1}\alpha)\nonumber\\
& - 4\alpha\gamma^{AB}\beta^C D_{[C}\partial_r\gamma_{A]B} + 2(d-1)\beta_A\beta^A\partial_r^2\alpha - 2(d-1)(D^A\alpha)\beta^B\partial_r\gamma_{AB}\nonumber\\
& + 6\alpha\gamma^{AB}\beta^C(\partial_r\gamma_{AC})\partial_r\beta_B - 4\alpha\gamma^{AB}\beta^C(\partial_r\gamma_{AB})\partial_r\beta_C\nonumber\\
& + 2\alpha\beta^A\beta^B\partial_r^2\gamma_{AB} - 2\alpha\beta_C\beta^C\gamma^{AB}\partial_r^2\gamma_{AB}\nonumber\\
& + (d-1)(\partial_r\alpha)\beta_C\beta^C\gamma^{AB}\partial_r\gamma_{AB} - 2(d-1)(\partial_r\alpha)\beta^A\beta^B\partial_r\gamma_{AB}\nonumber\\
& + (d-1)\beta^C(D_C\alpha)\gamma^{AB}\partial_r\gamma_{AB} - (d-1)\beta^C(\partial_u\beta_C)\gamma^{AB}\partial_r\gamma_{AB}\nonumber\\
& + 2(d-1)\gamma^{AB}\beta^C(\partial_r\gamma_{AC})\partial_u\beta_B - 4(d-1)\gamma^{AB}\beta^C(D_{[A}\beta_{C]})\partial_r\beta_B\nonumber\\
& - 3\alpha\gamma^{AB}\beta^C\beta^D(\partial_r\gamma_{AC})\partial_r\gamma_{BD} + 2\alpha\gamma^{AB}\beta^C\beta^D(\partial_r\gamma_{AB})\partial_r\gamma_{CD}\nonumber\\
& + \frac{3}{2}\alpha\beta_E\beta^E\gamma^{AB}\gamma^{CD}(\partial_r\gamma_{AC})\partial_r\gamma_{BD} - \frac{1}{2}\alpha\beta_E\beta^E\gamma^{AB}\gamma^{CD}(\partial_r\gamma_{AB})\partial_r\gamma_{CD}\nonumber\\
& - (d-1)(\beta^A\partial_r\beta_A)^2 + (d-1)\beta_C\beta^C\gamma^{AB}(\partial_r\beta_A)\partial_r\beta_B
\bigg]
\end{align}

\begin{align}
\label{eq:grtensors-einstein-uA}
\tilde{E}_{uA} = \frac{1}{D} \bigg[
& \ 2(d-1)(d-2)r^{-1}D_A\alpha - 2(d-1)D_A\partial_r\alpha\nonumber\\
& + 2(d-1)\alpha\partial_r^2\beta_A - 2(d-1)(d-2)r^{-1}\alpha(\partial_r\beta_A - r^{-1}\beta_A)\nonumber\\
& - 2(d-1)D^B D_{[A}\beta_{B]} + (d-1)\partial_r\partial_u\beta_A\nonumber\\
& + \beta_A\mathcal{R} - 2(d-1)\gamma^{BC} D_{[A}\partial_u\gamma_{B]C} + (d-1)(D^B \beta_B) \partial_r\beta_A\nonumber\\
& +2(d-1)(d-2)r^{-1}\beta^B D_{[A}\beta_{B]} + (d-1)\beta^B D_B \partial_r\beta_A\nonumber\\
& - 2\beta_A\partial_r^2\alpha - \beta_AD^B\partial_r\beta_B - \beta_A\gamma^{BC}\partial_r D_B\beta_C - \frac{1}{2}\beta_A D_D(\beta^D\gamma^{BC}\partial_r\gamma_{BC})\nonumber\\
& - (d-1)(D_A\alpha)\gamma^{BC}\partial_r\gamma_{BC} + 2(d-1)(D^B\alpha)\partial_r\gamma_{AB}\nonumber\\
& + 2(d-1)(\partial_r\alpha)\beta^B\partial_r\gamma_{AB} - (d-1)(d-2)r^{-1}\beta^B\partial_u\gamma_{AB}\nonumber\\
& + (d-1)\alpha\gamma^{BC}(\partial_r\gamma_{BC})\partial_r\beta_A - 2(d-1)\alpha\gamma^{BC}(\partial_r\gamma_{AB})\partial_r\beta_C\nonumber\\
& + \frac{d-1}{2}\gamma^{BC}(\partial_u\gamma_{BC})\partial_r\beta_A - (d-1)\gamma^{BC}(\partial_r\gamma_{AB})\partial_u\beta_C + (d-1)\beta^B\partial_r\partial_u\gamma_{AB}\nonumber\\
& - 2\beta_A\gamma^{BC}\partial_r\partial_u\gamma_{BC} - 2\alpha\beta_A\gamma^{BC}\partial_r^2\gamma_{BC} - 2(\partial_r\alpha)\beta_A\gamma^{BC}\partial_r\gamma_{BC} \nonumber\\
& + \frac{3}{2}\beta_A\gamma^{BC}\gamma^{DE}(\partial_r\gamma_{BD})\partial_u\gamma_{CE} - \frac{1}{2}\beta_A\gamma^{BC}\gamma^{DE}(\partial_r\gamma_{BC})\partial_u\gamma_{DE}\nonumber\\
& + \frac{3}{2}\alpha\beta_A\gamma^{BC}\gamma^{DE}(\partial_r\gamma_{BD})\partial_r\gamma_{CE} - \frac{1}{2}\alpha\beta_A\gamma^{BC}\gamma^{DE}(\partial_r\gamma_{BC})\partial_r\gamma_{DE}\nonumber\\
& - (d-1)\gamma^{BC}(D_A\beta_B)\partial_r\beta_C - 2(d-1)\beta^B\partial_r D_{[A}\beta_{B]} - (d-1)\gamma^{BC}\beta^D(\partial_u\gamma_{AB})\partial_r\gamma_{CD} \nonumber\\
& + \beta_A\gamma^{BC}\gamma^{DE}(\partial_r\gamma_{BD})D_{(C}\beta_{E)} - \frac{1}{2}\beta_A(D^D\beta_D)\gamma^{BC}\partial_r\gamma_{BC}\nonumber\\
& + (d-1)\beta_B\beta^B\partial_r^2\beta_A - (d-1)(d-2)r^{-1}\beta_B\beta^B(\partial_r\beta_A - r^{-1}\beta_A)\nonumber\\
& + (d-1)\beta^B(\partial_r\beta_B)\partial_r\beta_A + \frac{d-1}{2}\gamma^{BC}\beta^D(\partial_u\gamma_{AD})\partial_r\gamma_{BC}\nonumber\\
& - \frac{3}{2}\beta_A\gamma^{BC}(\partial_r\beta_B)\partial_r\beta_C - 2\beta_A\beta^B\partial_r^2\beta_B - 2\beta_A\gamma^{BC}\beta^D D_{[D}\partial_r\gamma_{B]C}\nonumber\\
& + 2(d-1)\gamma^{BC}\beta^D(D_{[A}\beta_{B]})\partial_r\gamma_{CD} + 2(d-1)\gamma^{BC}\beta^D(\partial_r\gamma_{AB})D_{[C}\beta_{D]}\nonumber\\
& - (d-1)\gamma^{BC}\beta^D(D_{[A}\beta_{D]})\partial_r\gamma_{BC}\nonumber\\
& + 3\beta_A\gamma^{BC}\beta^D(\partial_r\gamma_{BD})\partial_r\beta_C - 2\beta_A\gamma^{BC}\beta^D(\partial_r\gamma_{BC})\partial_r\beta_D\nonumber\\
& + \beta_A\beta^B\beta^C\partial_r^2\gamma_{BC} - \beta_A\beta_D\beta^D\gamma^{BC}\partial_r^2\gamma_{BC} + 2(d-1)\beta^B\beta^C(\partial_r\gamma_{C[A})\partial_r\beta_{B]}\nonumber\\
& + \frac{d-1}{2}\beta_D\beta^D\gamma^{BC}(\partial_r\gamma_{BC})\partial_r\beta_A - (d-1)\beta_D\beta^D\gamma^{BC}(\partial_r\gamma_{AB})\partial_r\beta_C\nonumber\\
& - \frac{3}{2}\beta_A\gamma^{BC}\beta^D\beta^E(\partial_r\gamma_{BD})\partial_r\gamma_{CE} + \beta_A\gamma^{BC}\beta^D\beta^E(\partial_r\gamma_{BC})\partial_r\gamma_{DE}\nonumber\\
& + \frac{3}{4}\beta_A\beta_F\beta^F\gamma^{BC}\gamma^{DE}(\partial_r\gamma_{BD})\partial_r\gamma_{CE} - \frac{1}{4}\beta_A\beta_F\beta^F\gamma^{BC}\gamma^{DE}(\partial_r\gamma_{BC})\partial_r\gamma_{DE}\bigg]
\end{align}

\begin{align}
\label{eq:grtensors-einstein-AB}
\tilde{E}_{AB} = \frac{1}{D} \bigg[
& - 2(d-1)\partial_r\partial_u\gamma_{AB} + (d-1)(d-2)r^{-1}\partial_u\gamma_{AB} + 2(d-1)\mathcal{R}_{AB}\nonumber\\
& - 2(d-1)\partial_r D_{(A}\beta_{B)} + 2(d-1)(d-2)r^{-1}D_{(A}\beta_{B)}\nonumber\\
& - 2(d-1)(\partial_r\alpha)\partial_r\gamma_{AB} - 2(d-1)\alpha\partial_r^2\gamma_{AB}\nonumber\\
& + 2(d-1)(d-2)r^{-1}\alpha(\partial_r\gamma_{AB} - r^{-1}\gamma_{AB})\nonumber\\
& - (d-1)(D^C\beta_C)\partial_r\gamma_{AB} - (d-1)\beta^C D_C\partial_r\gamma_{AB} - (d-1)(\partial_r\beta_A)\partial_r\beta_B\nonumber\\
& - \gamma_{AB}\mathcal{R} + 2\gamma_{AB}\partial_r^2\alpha + \gamma_{AB}D^C\partial_r\beta_C + \gamma_{AB}\gamma^{CD}\partial_r D_C\beta_D\nonumber\\
& + \frac{1}{2}\gamma_{AB}D_E(\beta^E\gamma^{CD}\partial_r\gamma_{CD})\nonumber\\
& + 2\gamma_{AB}\gamma^{CD}\partial_r\partial_u\gamma_{CD} + 2\alpha\gamma_{AB}\gamma^{CD}\partial_r^2\gamma_{CD} + 2(\partial_r\alpha)\gamma_{AB}\gamma^{CD}\partial_r\gamma_{CD} \nonumber\\
& - \frac{d-1}{2}\gamma^{CD}\left((\partial_u\gamma_{AB})\partial_r\gamma_{CD} + (\partial_r\gamma_{AB})\partial_u\gamma_{CD}\right)\nonumber\\
& + (d-1)\gamma^{CD}\left((\partial_r\gamma_{AC})\partial_u\gamma_{BD} + (\partial_u\gamma_{AC})\partial_r\gamma_{BD}\right)\nonumber\\
& - \frac{3}{2}\gamma_{AB}\gamma^{CD}\gamma^{EF}(\partial_r\gamma_{CE})\partial_u\gamma_{DF} + \frac{1}{2}\gamma_{AB}\gamma^{CD}\gamma^{EF}(\partial_r\gamma_{CD})\partial_u\gamma_{EF}\nonumber\\
& - (d-1)\alpha\gamma^{CD}(\partial_r\gamma_{AB})\partial_r\gamma_{CD} + 2(d-1)\alpha\gamma^{CD}(\partial_r\gamma_{AC})\partial_r\gamma_{BD}\nonumber\\
& - \frac{3}{2}\alpha\gamma_{AB}\gamma^{CD}\gamma^{EF}(\partial_r\gamma_{CE})\partial_r\gamma_{DF} + \frac{1}{2}\alpha\gamma_{AB}\gamma^{CD}\gamma^{EF}(\partial_r\gamma_{CD})\partial_r\gamma_{EF}\nonumber\\
& + 2(d-1)\beta^C(\partial_r\beta_{(A})\partial_r\gamma_{B)C} - 2(d-1)\beta^C(\partial_r\beta_C)\partial_r\gamma_{AB}\nonumber\\
& - (d-1)\gamma^{CD}(D_{(A}\beta_{B)})\partial_r\gamma_{CD} + 2(d-1)(D^C\beta_{(A})\partial_r\gamma_{B)C}\nonumber\\
& - \gamma_{AB}\gamma^{CD}\gamma^{EF}(\partial_r\gamma_{CE})D_{(D}\beta_{F)} + \frac{1}{2}\gamma_{AB}(D^E\beta_E)\gamma^{CD}\partial_r\gamma_{CD}\nonumber\\
& + \frac{3}{2}\gamma_{AB}\gamma^{CD}(\partial_r\beta_C)\partial_r\beta_D + 2\gamma_{AB}\beta^C\partial_r^2\beta_C + 2\gamma_{AB}\gamma^{CD}\beta^E D_{[E}\partial_r\gamma_{C]D}\nonumber\\
& - (d-1)\beta_C\beta^C\partial_r^2\gamma_{AB} + (d-1)(d-2)r^{-1}\beta_C\beta^C(\partial_r\gamma_{AB} - r^{-1}\gamma_{AB})\nonumber\\
& + 2(d-1)\beta^C\beta^D(\partial_r\gamma_{A[B})\partial_r\gamma_{C]D}\nonumber\\
& - 3\gamma_{AB}\gamma^{CD}\beta^E(\partial_r\gamma_{CE})\partial_r\beta_D + 2\gamma_{AB}\gamma^{CD}\beta^E(\partial_r\gamma_{CD})\partial_r\beta_E\nonumber\\
& - \gamma_{AB}\beta^C\beta^D\partial_r^2\gamma_{CD} + \gamma_{AB}\beta_E\beta^E\gamma^{CD}\partial_r^2\gamma_{CD}\nonumber\\
& - \frac{d-1}{2}\beta_E\beta^E\gamma^{CD}(\partial_r\gamma_{AB})\partial_r\gamma_{CD} + (d-1)\beta_E\beta^E\gamma^{CD}(\partial_r\gamma_{AC})\partial_r\gamma_{BD}\nonumber\\
& + \frac{3}{2}\gamma_{AB}\gamma^{CD}\beta^E\beta^F(\partial_r\gamma_{CE})\partial_r\gamma_{DF} - \gamma_{AB}\gamma^{CD}\beta^E\beta^F(\partial_r\gamma_{CD})\partial_r\gamma_{EF}\nonumber\\
& - \frac{3}{4}\gamma_{AB}\beta_G\beta^G\gamma^{CD}\gamma^{EF}(\partial_r\gamma_{CE})\partial_r\gamma_{DF} + \frac{1}{4}\gamma_{AB}\beta_G\beta^G\gamma^{CD}\gamma^{EF}(\partial_r\gamma_{CD})\partial_r\gamma_{EF}
\bigg]
\end{align}

\subsection{Riemann Tensor}
\label{app:grtensors-riemann}
Components of the Riemann tensor are required when commuting derivatives on spinors in the Dirac equation.
Two components are also needed to calculate components of the Weyl tensor used in the calculation of the Bondi mass.
These components are provided here, up to the order of $r$ to which they are needed.
\medskip\noindent
The Riemann tensor is given by
\begin{equation}
\label{eq:grtensors-riemann}
\tilde{R}_{abc}{}^d = \partial_b\tilde\Gamma^d{}_{ac} - \partial_a\tilde\Gamma^d{}_{ab} + \tilde\Gamma^e{}_{ac}\tilde\Gamma^d{}_{eb} - \tilde\Gamma^e{}_{bc}\tilde\Gamma^d{}_{ea} \quad .
\end{equation}
\medskip\noindent
The components we require are
\begin{subequations}
\begin{align}
\label{eq:grtensors-riemann-rur-r}
\tilde{R}_{rur}{}^r &= \partial_r^2\alpha + O(r^{d-2}) \quad , \\
\label{eq:grtensors-riemann-rur-A}
\tilde{R}_{rur}{}^B &= O(r^{(d-4)/2}) \quad ,
\end{align}
\begin{align}
\label{eq:grtensors-riemann-rAr-r}
\tilde{R}_{rAr}{}^r &= \frac{1}{2}(\partial_r^2\beta_A) - \frac{1}{2}\beta^B(\partial_r^2\gamma_{AB}) - \frac{1}{4}\gamma^{BC}(\partial_r\gamma_{AB})\partial_r\beta_C \nonumber \\
& + \frac{1}{4}\gamma^{BC}\beta^D(\partial_r\gamma_{AB})\partial_r\gamma_{CD} \quad ,
\end{align}
\begin{equation}
\label{eq:grtensors-riemann-rAu-r}
\tilde{R}_{rAu}{}^r = -\alpha\partial_r^2\beta_A - \frac{1}{2}\beta^B\partial_r\partial_u\gamma_{AB} + O(r^{d-1}) \quad ,
\end{equation}
\begin{align}
\label{eq:grtensors-riemann-rAB-r}
\tilde{R}_{rAB}{}^r &=  \frac{1}{2}\mathscr{D}_B\partial_r\beta_A + \frac{1}{2}\partial_r\partial_u\gamma_{AB} + \frac{1}{2}(\partial_r\alpha)\partial_r\gamma_{AB} + \alpha\partial_r^2\gamma_{AB} \nonumber \\
& - \frac{1}{4}\gamma^{CD}(\partial_u\gamma_{AC})\partial_r\gamma_{BD} \quad ,
\end{align}
\begin{equation}
\label{eq:grtensors-riemann-rAr-u}
\tilde{R}_{rAr}{}^u = 0 \quad ,
\end{equation}
\begin{equation}
\label{eq:grtensors-riemann-rAu-u}
\tilde{R}_{rAu}{}^u = -\frac{1}{2}\partial_r^2\beta_A + \frac{1}{4}\gamma^{BC}(\partial_r\gamma_{AB})\partial_r\beta_C \quad ,
\end{equation}
\begin{equation}
\label{eq:grtensors-riemann-rAB-u}
\tilde{R}_{rAB}{}^u = \frac{1}{2}\partial_r^2\gamma_{AB} - \frac{1}{4}\gamma^{CD}(\partial_r\gamma_{AC})\partial_r\gamma_{BD} \quad ,
\end{equation}
\begin{equation}
\label{eq:grtensors-riemann-rAr-C}
\tilde{R}_{rAr}{}^C = -\frac{1}{2}\gamma^{BC}\partial_r^2\gamma_{AB} + \frac{1}{4}\gamma^{BC}\gamma^{DE}(\partial_r\gamma_{AD})\partial_r\gamma_{BE} \quad ,
\end{equation}
\begin{align}
\tilde{R}_{rAu}{}^C &= -\frac{1}{2}\mathscr{D}^C\partial_r\beta_A - \frac{1}{2}\gamma^{BC}\partial_r\partial_u\gamma_{AB} - \frac{1}{2}(\partial_r\alpha)\gamma^{BC}\partial_r\gamma_{AC} \nonumber \\
& + \frac{1}{4}\gamma^{BC}\gamma^{DE}(\partial_u\gamma_{AD})\partial_r\gamma_{BE} + O(r^{d-2}) \quad ,
\end{align}
\begin{align}
\tilde{R}_{rAB}{}^C &= \frac{1}{2}\gamma^{CD}D_A\partial_r\gamma_{BD} - \partial_r\tilde\Lambda^C{}_{AB} + \frac{1}{2}\beta^C(\partial_r^2\gamma_{AB}) \nonumber \\
& + \frac{1}{2}\gamma^{CD}(\partial_r\gamma_{A[B})\partial_r\beta_{D]} + O(r^{d-1}) \quad ,
\end{align}
\end{subequations}
With all indices lowered, the Riemann tensor components are as follows.
\begin{subequations}
\begin{align}
\label{eq:grtensors-riemann-ruru}
\tilde{R}_{ruru}
& = \tilde{R}_{rur}{}^d\tilde{g}_{ud} \nonumber \\
& = \tilde{R}_{rur}{}^r \nonumber \\
& = \partial_r^2\alpha + O(r^{d-2}) \quad ,
\end{align}
\begin{align}
\label{eq:grtensors-riemann-rurA}
\tilde{R}_{rurA}
& = \tilde{R}_{rur}{}^d\tilde{g}_{dA} \nonumber \\
& = -\beta_A\tilde{R}_{rur}{}^u + \gamma_{AB}\tilde{R}_{rur}{}^B \nonumber\\
& = O(r^{(d-4)/2}) \quad .
\end{align}
\begin{align}
\label{eq:grtensors-riemann-rAru}
\tilde{R}_{rAru} = \frac{1}{2}\partial_r^2\beta_A - \frac{1}{4}\gamma^{BC}(\partial_r\gamma_{AB})\partial_r\beta_C + O(r^{d-1}) \quad ,
\end{align}
\begin{align}
\label{eq:grtensors-riemann-rArB}
\tilde{R}_{rArB} = -\frac{1}{2}\partial_r^2\gamma_{AB} + \frac{1}{4}\gamma^{CD}(\partial_r\gamma_{AC})\partial_r\gamma_{BD} + O(r^{d-1}) \quad ,
\end{align}
\begin{align}
\label{eq:grtensors-riemann-rAuB}
\tilde{R}_{rAuB} &= - \frac{1}{2}\mathscr{D}_B\partial_r\beta_A - \frac{1}{2}\partial_r\partial_u\gamma_{AB} - \frac{1}{2}(\partial_r\alpha)\partial_r\gamma_{AB} \nonumber \\
& + \frac{1}{4}\gamma^{CD}(\partial_u\gamma_{AC})\partial_r\gamma_{BD} + O(r^{d-2}) \quad ,
\end{align}
\begin{align}
\label{eq:grtensors-riemann-rABD}
\tilde{R}_{rABD} = -D_{[B}\partial_r\gamma_{D]A} + \frac{1}{2}(\partial_r\gamma_{A[B})\partial_r\beta_{D]} + O(r^{d-2}) \quad .
\end{align}
\end{subequations}

\subsection{Weyl Tensor}
\label{app:grtensors-weyl}
Two components of the unphysical Weyl tensor are required in the calculation of the Bondi mass formula (\ref{eq:bondi-mass-formula}).
Those components are provided here, only to the order of $r$ to which they are required.
The Weyl tensor is given by
\begin{equation}
\label{eq:grtensors-weyl}
\tilde{C}_{abcd} = \tilde{R}_{abcd} - \frac{2}{d-2}(\tilde{g}_{a[c}\tilde{R}_{d]b} - \tilde{g}_{b[c}\tilde{R}_{d]a} + \frac{2}{(d-1)(d-2)}\tilde{R}\tilde{g}_{a[c}\tilde{g}_{d]b} \quad .
\end{equation}
The components we want are
\begin{subequations}
\begin{align}
\tilde{C}_{ruru}
\iftoggle{complete}{
& = \tilde{R}_{ruru} - \frac{2}{d-2}(\tilde{g}_{r[r}\tilde{R}{u]u} - \tilde{g}_{u[r}\tilde{R}{u]r}) + \frac{2}{(d-1)(d-2)}\tilde{R}\tilde{g}_{r[r}\tilde{g}{u]u} \nonumber \\
& = \tilde{R}_{ruru} + \frac{2}{d-2}\tilde{R}_{ru} - \frac{1}{(d-1)(d-2)}\tilde{R} \nonumber\\
}{}
& = \tilde{R}_{ruru} + \tilde{S}_{ru} \nonumber\\
& = \partial_r^2\alpha - 2r^{-1}(\partial_\alpha - r^{-1}\alpha) + O(r^{d-2}) \\
\tilde{C}_{rurA}
\iftoggle{complete}{
& = \tilde{R}_{rurA} - \frac{2}{d-2}(\tilde{g}_{r[r}\tilde{R}_{A]u} - \tilde{g}_{u[r}\tilde{R}_{A]r}) + \frac{2}{(d-1)(d-2)}\tilde{R}\tilde{g}_{r[r}\tilde{g}_{A]u} \nonumber \\
}{}
& = -\beta_A\tilde{R}_{rur}{}^r + \gamma_{AB}\tilde{R}_{rur}{}^B + \frac{1}{2}\tilde{S}_{rA} \nonumber \\
& = O(r^{(d-4)/2}) \quad .
\end{align}
\end{subequations}

\normalsize

%% file: GRSpinors.tex
\section{Spinor Formalism in GNCs}
\label{app:grspinors}

Basic spinor formalism and results are given in Gaussian null coordinates.

\subsection{Spinor Basis}
\label{app:grspinors-basis}

\begin{minipage}{.5\columnwidth}
\begin{subequations}
\begin{align}
\label{eq:grspinors-spinor-basis-+_a}
\tilde{e}^+_a &= du_a \\
\label{eq:grspinors-spinor-basis--_a}
\tilde{e}^-_a &= dr_a - \alpha du_a - \beta_A dx^A_a \\
\label{eq:grspinors-spinor-basis-A_a}
\tilde{e}^{\hat{A}}_a &= \hat{\tilde{e}}^{\hat{A}}_A dx^A_a
\end{align}
\end{subequations}
\end{minipage}
\begin{minipage}{.5\columnwidth}
\begin{subequations}
\begin{align}
\label{eq:grspinors-spinor-basis-+^a}
\tilde{e}^{+a} &= \partial_r^a \\
\label{eq:grspinors-spinor-basis--^a}
\tilde{e}^{-a} &= \alpha\partial_r^a + \partial_u^a \\
\label{eq:grspinors-spinor-basis-A^a}
\tilde{e}^{\hat{A}a} &= \beta_A\hat{\tilde{e}}^{\hat{A}A}\partial_r^a + \hat{\tilde{e}}^{\hat{A}A}\partial_A^a
\end{align}
\end{subequations}
\end{minipage}

\subsection{Gamma Matrices}
\label{app:grspinors-gamma-matrices}

\begin{subequations}
\begin{alignat}{4}
\label{eq:grspinors-gamma-matrices}
\tilde{\Gamma}^a &= \Gamma_\mu \tilde{e}^{\mu a} &&= (\Gamma_+ + \alpha\Gamma_- + \beta_A\tilde\Lambda^A)\partial_r^a + \Gamma_-\partial_u^a + \tilde\Lambda^A\partial_A^a \\
\label{eq:grspinors-gamma-matrices-down}
\tilde{\Gamma}_a &= \Gamma_\mu \tilde{e}^\mu_a &&= \Gamma_-dr_a + (\Gamma_+ - \alpha\Gamma_-)du_a + (-\beta_A\Gamma_- + \tilde\Lambda_A)dx^A_a
\end{alignat}
\end{subequations}

\subsection{Spinor Connection}
\label{app:grspinors-spinor-connection}

The covariant derivative satisfying $\tilde\nabla_a\tilde\Gamma_b=0$ is defined by the connection coefficients,
\begin{equation}
\label{eq:grspinors-spinor-connection-coefficients-formula}
\tilde\omega_a = \tilde\omega^{\mu\nu}_a [\Gamma_\mu, \Gamma_\nu] = \frac{1}{8} \tilde{e}^{\mu b} (\tilde\nabla_a \tilde{e}^\nu_b) [\Gamma_\mu, \Gamma_\nu]
\end{equation}
Its components in Gaussian null coordinates are given below.

\begin{subequations}
\label{eq:grspinors-spinor-connection-coefficients}

\begin{align}
\label{eq:grspinors-spinor-connection-coefficient-r}
\tilde\omega_r &= -\frac{1}{4}(\partial_r\beta_A)\tilde\Lambda^A\Gamma_- - \frac{1}{8}\gamma^{AB}\partial_r\gamma_{AB} + \frac{1}{4}\tilde\Lambda^A\partial_r\tilde\Lambda_A
\end{align}

\begin{align}
\label{eq:grspinors-spinor-connection-coefficient-u}
\tilde\omega_u 
& = \frac{1}{2}(\partial_r\alpha)(P_+ - P_-) - \frac{1}{4}(\partial_r\beta_A)\tilde\Lambda^A\Gamma_+ + \frac{1}{4}\tilde\Lambda^A\partial_u\tilde\Lambda_A\nonumber\\
& + \left( -\frac{1}{2}\tilde{D}_A\alpha + \frac{1}{2}\partial_u\beta_A + \frac{1}{4}\alpha\partial_r\beta_A - \frac{1}{2}\beta_A\partial_r\alpha\right)\tilde\Lambda^A\Gamma_-\nonumber\\
& + \left(\frac{1}{4}\beta_{[A}\partial_r\beta_{B]} + \frac{1}{4}\tilde{D}_{[A}\beta_{B]}\right)\tilde\Lambda^A\tilde\Lambda^B - \frac{1}{8}\gamma^{AB}\partial_u\gamma_{AB}
\end{align}

\begin{align}
\label{eq:grspinors-spinor-connection-coefficient-A}
\tilde\omega_A
& = \tilde\Omega_A + \frac{1}{4}(\partial_r\beta_A)(P_+ - P_-) - \frac{1}{4}(\partial_r\gamma_{AB})\tilde\Lambda^B\Gamma_+\nonumber\\
& + \left(\frac{1}{4}\partial_u\gamma_{AB} - \frac{1}{2}\tilde{D}_{[A}\beta_{B]} + \frac{1}{4}\alpha\partial_r\gamma_{AB} + \frac{1}{4}(\partial_r\beta_A)\beta_B\right)\tilde\Lambda^B\Gamma_-\nonumber\\
& + \frac{1}{4}(\partial_r\gamma_{A[B})\beta_{C]}\tilde\Lambda^B\tilde\Lambda^C
\end{align}
\end{subequations}

\iftoggle{complete}{
\subsubsection*{Components}

\begin{subequations}
\begin{align}
\label{eq:grspinors-spinor-connection-coefficient-r+-}
\tilde\omega^{+-}_r &= 0\\
\label{eq:grspinors-spinor-connection-coefficient-r+A}
\tilde\omega^{+\hat{A}}_r &= 0\\
\label{eq:grspinors-spinor-connection-coefficient-r-A}
\tilde\omega^{-\hat{A}}_r &= \frac{1}{16}\hat{\tilde{e}}^{\hat{A}A}\partial_r\beta_A\\
\label{eq:grspinors-spinor-connection-coefficient-rAB}
\tilde\omega^{\hat{A}\hat{B}}_r &= \frac{1}{8}\hat{\tilde{e}}^{\hat{A}A}\partial_r\hat{\tilde{e}}^{\hat{B}}_A - \frac{1}{16}\hat{\tilde{e}}^{\hat{A}A}\hat{\tilde{e}}^{\hat{B}B}\partial_r\gamma_{AB}
\end{align}
\end{subequations}

\begin{subequations}
\begin{align}
\label{eq:grspinors-spinor-connection-coefficient-u+-}
\tilde\omega^{+-}_u &= \frac{1}{8}\partial_r\alpha\\
\label{eq:grspinors-spinor-connection-coefficient-u+A}
\tilde\omega^{+\hat{A}}_u &= \frac{1}{16}\hat{\tilde{e}}^{\hat{A}A}\partial_r\beta_A\\
\label{eq:grspinors-spinor-connection-coefficient-u-A}
\tilde\omega^{-\hat{A}}_u &= \left(-\frac{1}{8}\tilde{D}_A\alpha + \frac{1}{8}\partial_u\beta_A + \frac{1}{16}\alpha\partial_r\beta_A - \frac{1}{8}\beta_A\partial_r\alpha\right)\hat{\tilde{e}}^{\hat{A}A}\\
\label{eq:grspinors-spinor-connection-coefficient-uAB}
\tilde\omega^{\hat{A}\hat{B}}_u &= \frac{1}{8}\hat{\tilde{e}}^{\hat{A}A}\partial_u\hat{\tilde{e}}^{\hat{B}}_A + \left( -\frac{1}{16}\partial_u\gamma_{AB} + \frac{1}{8}\beta_{[A}\partial_r\beta_{B]} + \frac{1}{8}\tilde{D}_{[A}\beta_{B]}\right)\hat{\tilde{e}}^{\hat{A}A}\hat{\tilde{e}}^{\hat{B}B}
\end{align}
\end{subequations}

\begin{subequations}
\begin{align}
\label{eq:grspinors-spinor-connection-coefficient-C+-}
\tilde\omega^{+-}_C &= \frac{1}{16}\partial_r\beta_C\\
\label{eq:grspinors-spinor-connection-coefficient-C+A}
\tilde\omega^{+\hat{A}}_C &= - \frac{1}{16}\hat{\tilde{e}}^{\hat{A}A}\partial_r\gamma_{AC}\\
\label{eq:grspinors-spinor-connection-coefficient-C-A}
\tilde\omega^{-\hat{A}}_C &= \left(-\frac{1}{8}\tilde{D}_{[A}\beta_{C]} - \frac{1}{16}\partial_u\gamma_{AC} - \frac{1}{16}\alpha\partial_r\gamma_{AC} - \frac{1}{16}\beta_A\partial_r\beta_C\right)\hat{\tilde{e}}^{\hat{A}A}\\
\label{eq:grspinors-spinor-connection-coefficient-CAB}
\tilde\omega^{\hat{A}\hat{B}}_C &= \tilde\Omega^{\hat{A}\hat{B}} - \frac{1}{8}\hat{\tilde{e}}^{\hat{A}A}\hat{\tilde{e}}^{\hat{B}B}\beta_{[A}\partial_r\gamma_{B]C}
\end{align}
\end{subequations}
}{}

\subsection{Derivatives of Gamma Matrices}
\label{app:grspinors-gamma-matrices-derivatives}
Here we give expressions for the $r$ derivatives of the Gamma matrices $\Gamma_+,\Gamma_-,\tilde\Lambda^A$.
The derivations of these expressions are given in Note \ref{note:gamma-matrices-derivatives}.
The $n^{th}$ $r$ derivatives of the flat space Gamma matrices $\Gamma_+$ and $\Gamma_-$ are given by, for $n\geq1$,
\begin{subequations}
\label{eq:grspinors-gamma-matrices-derivatives}
\begin{align}
\label{eq:grspinors-gamma-matrices-derivatives-+}
& \tilde\nabla_r^n\Gamma_+ = - \frac{1}{2}(\partial_r^n\beta_A)\tilde\Lambda^A + \frac{1}{4}\sum_{k=1}^{n-1}\frac{(n-1)!}{j!(n-1-j)!}\gamma^{AB}(\partial_r\gamma_{AC}^{(j)})(\partial_r^{n-j}\beta_B)\tilde\Lambda^C + O(r^{d-n}) \quad , \\
\label{eq:grspinors-gamma-matrices-derivatives--}
& \tilde\nabla_r^n\Gamma_- = 0 \quad .
\end{align}
The the $n^{th}$ $r$ derivative of the curved space gamma matrix $\tilde\Lambda^A$ is given by, for $n\geq1$,
\begin{align}
\label{eq:grspinors-gamma-matrices-derivatives-A}
\tilde\nabla_r^n \tilde\Lambda^A
& = \frac{1}{2}\gamma^{AB}(\partial_r^n\beta_B)\Gamma_- - \frac{1}{2}\gamma^{AB}(\partial_r^n\gamma_{BC})\tilde\Lambda^{C(0)} \nonumber \\
& + \frac{1}{4}\sum_{k=1}^{n-1}a_{n,k}\gamma^{AB}\gamma^{CD}(\partial_r^j\gamma_{BC})(\partial_r^{n-j}\gamma_{DE})\tilde\Lambda^E + O(r^{d-n-1}) \quad ,
\end{align}
\end{subequations}
where $a_{n,k}$ are integer coefficients (whose values were difficult to determine, but are not important to the arguments where they are used).

\subsection{Commutation of Derivatives on Spinors}
\label{app:grspinors-derivative-commutation}
The commutation of unphysical covariant $r$ and $A$ derivatives on a spinor $\tilde\psi$ is given by
\begin{equation}
\tilde\nabla_r^n\tilde\nabla_A\tilde\psi = \tilde\nabla_A\tilde\nabla_r^n\tilde\psi + \sum_{k=1}^n \frac{n!}{k!(n-k)!}(\tilde\nabla_r^{k-1}\tilde{Q}_A)\tilde\nabla_r^{n-k}\tilde\psi \quad ,
\end{equation}
where
\begin{equation}
\tilde{Q}_A = \frac{1}{8}\tilde{R}_{rAcd}[\tilde\Gamma^c, \tilde\Gamma^d] \quad .
\end{equation}
Thus, we have
\begin{align}
\label{eq:grspinors-derivative-commutator}
\tilde{Q}_A
& = \frac{1}{8}\tilde{R}_{rAcd}[\tilde\Gamma^c,\tilde\Gamma^d] \nonumber \\
& = \frac{1}{4}\tilde{R}_{rAru}[\tilde\Gamma^r,\tilde\Gamma^u] + \frac{1}{4}\tilde{R}_{rArB}[\tilde\Gamma^r,\tilde\Gamma^B]  + \frac{1}{4}\tilde{R}_{rAuB}[\tilde\Gamma^u,\tilde\Gamma^B] + \frac{1}{8}\tilde{R}_{rABD}[\tilde\Gamma^A,\tilde\Gamma^B] \nonumber \\
& = \frac{1}{2}(P_+ - P_- + \beta_A\tilde\Lambda^A\Gamma_-) - \frac{1}{2}\tilde{R}_{rArB}(\tilde\Lambda^B\Gamma_+ + \alpha\tilde\Lambda^B\Gamma_- - \frac{1}{2}\beta_A[\tilde\Lambda^B,\tilde\Lambda^A]) \nonumber \\
& - \frac{1}{4}\tilde{R}_{rAuB}\tilde\Lambda^B\Gamma_- + \frac{1}{4}\tilde{R}_{rABD}\tilde\Lambda^B\tilde\Lambda^D \quad .
\end{align}
Substituting in the expressions from Appendix \ref{app:grtensors-riemann}, we get the final expression for $\tilde{Q}_A$.
\begin{align}
\label{eq:grspinors-derivative-commutator-formula}
\tilde{Q}_A
& = \left(\frac{1}{4}\partial_r^2\beta_A - \frac{1}{8}\gamma^{BC}(\partial_r\gamma_{AB})\partial_r\beta_C\right)(P_+ - P_-) \nonumber \\
& + \left(\frac{1}{4}\partial_r^2\gamma_{AB} - \frac{1}{8}\gamma^{CD}(\partial_r\gamma_{AC})\partial_r\gamma_{BD}\right)\tilde\Lambda^B\Gamma_+ \nonumber \\
& + \left(\frac{1}{4}\alpha\partial_r^2\gamma_{AB} + \frac{1}{4}(\partial_r\alpha)\partial_r\gamma_{AB}\right)\tilde\Lambda^B\Gamma_- \nonumber \\
& + \left(\frac{1}{4}\mathscr{D}_B\partial_r\beta_A + \frac{1}{4}\partial_r\partial_u\gamma_{AB} - \frac{1}{8}\gamma^{CD}(\partial_u\gamma_{AC})\partial_r\gamma_{BD}\right)\tilde\Lambda^B\Gamma_- \nonumber \\
& + \left(-\frac{1}{4}D_{[B}\partial_r\gamma_{D]A} + \frac{1}{4}(\partial_r^2\gamma_{A[B}\beta_{D]} + \frac{1}{8}(\partial_r\gamma_{A[B})\partial_r\beta_{D]}\right)\tilde\Lambda^B\tilde\Lambda^D \nonumber \\
& + O(r^{d-2}) \quad .
\end{align}
For $0\leq n\leq d-2$, we have
\begin{align}
\tilde\nabla_r^n\tilde{Q}_A
& = \left(\frac{1}{4}\partial_r^{n+2}\beta_A - \frac{1}{8}\sum_{k=0}^n \frac{n!}{k!(n-k)!} \gamma^{BC}(\partial_r^{k+1}\gamma_{AB})\partial_r^{n+1-k}\beta_C\right)(P_+ - P_-) \nonumber \\
& + \left(\frac{1}{4}\partial_r^{n+2}\gamma_{AB} - \frac{1}{8}\sum_{k=0}^n \frac{n!}{k!(n-k)!}\gamma^{BC}(\partial_r^{k+1}\gamma_{AC})\partial_r^{n+1-k}\gamma_{BD}\right)\tilde\Lambda^B\Gamma_+ \nonumber \\
& + \frac{1}{4}\sum_{k=1}^n \frac{n!}{k!(n-k)!}(\partial_r^{n+2-k}\gamma_{AB})(\tilde\nabla_r^k\tilde\Lambda^B)\Gamma_+ \nonumber \\
& + \frac{1}{4}\sum_{k=1}^n\frac{n!}{k!(n-k)!}(\partial_r^{n+2-k}\gamma_{AB})\tilde\Lambda^B\tilde\nabla_r^k\Gamma_+ \nonumber \\
& + \left(\frac{n(n+1)}{8}(\partial_r^2\alpha)\partial_r^n\gamma_{AB} + \frac{n+1}{4}(\partial_r\alpha)\partial_r^{n+1}\gamma_{AB} + \frac{1}{4}\alpha\partial_r^{n+2}\gamma_{AB}\right)\tilde\Lambda^B\Gamma_- \nonumber \\
& + \frac{1}{4}(\mathscr{D}_B\partial_r^{n+1}\beta_A)\tilde\Lambda^B\Gamma_- + \frac{1}{4}\sum_{k=0}^n\frac{n!}{k!(n-k)!}(\partial_r^{k+1}\partial_u\gamma_{AB})(\tilde\nabla_r^{n-k}\tilde\Lambda^B)\Gamma_- \nonumber \\
& - \frac{1}{8}\sum_{k=0}^n \frac{n!}{k!(n-k)!}\gamma^{CD}(\partial_r^k\partial_u\gamma_{AC})(\partial_r^{n+1-k}\gamma_{BD})\tilde\Lambda^B\Gamma_- \nonumber \\
& + \frac{1}{8}\sum_{k=0}^{n+1} b_{k,n} (\partial_r^{n+2-k}\gamma_{A[B})(\partial_r^k\beta_{D]}\tilde\Lambda^B\tilde\Lambda^D \nonumber \\
& - \frac{1}{4}\sum_{k=1}^n \frac{n!}{k!(n-k)!}(\partial_r^{n-k}D_{[B}\partial_r\gamma_{D]A})[(\tilde\nabla^k\tilde\Lambda^B)\tilde\Lambda^D + \tilde\Lambda^B(\tilde\nabla_r^k\tilde\Lambda^D)] \nonumber \\
& - \frac{1}{4}(\partial_r^nD_{[B}\partial_r\gamma_{D]A})\tilde\Lambda^B\tilde\Lambda^D + O(r^{d-2-n}) \quad ,
\end{align}
where in the second to last line $b_{k,n}$ are integer coefficients (which have not been determined since they are not needed).
This equation has already been proved for $n=0$ as in that case it is simply equation (\ref{eq:grspinors-derivative-commutator-formula}) when $n=0$.
It can then be proved for $n>0$ by induction.